\documentclass[10pt,twocolumn,aps,prx,superscriptaddress,floatfix]{revtex4-1}

\usepackage{epsfig}
\usepackage{cancel}
\usepackage{soul} 
\usepackage{mathtools}
\usepackage{multirow}
\usepackage{amsmath, amssymb, amsfonts, mathrsfs}
\usepackage{amsthm}
\usepackage{slashbox} 

\theoremstyle{definition}

\theoremstyle{remark}

\theoremstyle{plain}

\newtheorem{claim}{Claim}
\newtheorem{lemma}{Lemma}

\usepackage{graphicx}
\usepackage[usenames, dvipsnames]{color}
\usepackage{dsfont} 
\usepackage{xcolor}
\usepackage{tcolorbox}
\usepackage{hyperref}
\hypersetup{
    colorlinks,
    linkcolor={rgb:red,2; blue,1},
    citecolor={blue!60!black},
    urlcolor={blue!60!black}
}

\usepackage{esvect} 

\tcbset{colback=orange!5,colframe=orange!50!blue,fonttitle=\bfseries, float=htb}

\newcommand{\ket}[1]{\vert{#1}\rangle} 
\newcommand{\bra}[1]{\langle{#1}\vert} 
 
\DeclareMathOperator{\Tr}{Tr}

\newcommand{\beq}{\begin{equation}}
\newcommand{\eeq}{\end{equation}}

\newcommand{\figref}[1]{Fig.~\ref{#1}}
\newcommand{\secref}[1]{Sec.~\ref{#1}}
\newcommand{\appref}[1]{App.~\ref{#1}}

\newcommand{\jona}[1]{{\color{Black} #1}}

\newcommand{\fab}[1]{{\color{Black} #1}}

\newcommand{\nico}[1]{{\color{Black} #1}}
\newcommand{\gh}[1]{{\color{Black} #1}}
\usepackage{changes}
\tcbset{colback=green!5,colframe=green!50!blue,
fonttitle=\bfseries, float=htb}

\newcommand{\defeq}{\vcentcolon=}
\newtcolorbox[blend into=figures]{boxfigure}[3][]
{ float*=ht,width=\textwidth,lower separated=false, center upper, 
title={#2},label= fig:#3,#1}

\begin{document}

\title{Unifying paradigms of quantum refrigeration:\\
fundamental limits of cooling and associated work costs}
\author{Fabien Clivaz}
\affiliation{Department of Applied Physics, University of Geneva, 1211 Geneva 4, Switzerland}
\author{Ralph Silva}
\affiliation{Department of Applied Physics, University of Geneva, 1211 Geneva 4, Switzerland}
\author{G\'eraldine Haack}
\affiliation{Department of Applied Physics, University of Geneva, 1211 Geneva 4, Switzerland}
\author{Jonatan Bohr Brask}
\affiliation{Department of Applied Physics, University of Geneva, 1211 Geneva 4, Switzerland}
\author{Nicolas Brunner}
\affiliation{Department of Applied Physics, University of Geneva, 1211 Geneva 4, Switzerland}
\author{Marcus Huber}
\affiliation{Institute for Quantum Optics and Quantum Information (IQOQI), Austrian Academy of Sciences, Boltzmanngasse 3, A-1090 Vienna, Austria}

\date{\today}

\begin{abstract}
In classical thermodynamics the work cost of control can typically be neglected. On the contrary, in quantum thermodynamics the cost of control constitutes a fundamental contribution to the total work cost. Here, focusing on quantum refrigeration, we investigate how the level of control determines the fundamental limits to cooling and how much work is expended in the corresponding process. \jona{We compare two extremal levels of control. First coherent operations, where the entropy of the resource is left unchanged, and second incoherent operations, where only energy at maximum entropy (i.e. heat) is extracted from the resource. For minimal machines, we find that the lowest achievable temperature and associated work cost depend strongly on the type of control, in both single-cycle and asymptotic regimes. We also extend our analysis to general machines.} Our work provides a unified picture of the different approaches to quantum refrigeration developed in the literature, including algorithmic cooling, autonomous quantum refrigerators, and the resource theory of quantum thermodynamics.
\end{abstract}

\maketitle

\section{Introduction}

Characterizing the ultimate performance limits of thermal machines is directly connected to the problem of understanding the fundamental laws of thermodynamics. The development of classical thermodynamics was instrumental for the realization of efficient thermal machines. Similarly, understanding the thermodynamics of quantum systems is closely related to the fundamental limits of quantum thermal machines. An intense research effort has been devoted to these questions \cite{book,review,millen,janet}, resulting in the formulation of the basic laws of quantum thermodynamics, a resource theory perspective, and a large body of work on quantum thermal machines, including first experimental demonstrations. 

When trying to establish fundamental limits on quantum thermodynamics tasks, one is always faced with the problem of identifying the relevant resources. For instance, one may consider different classes of allowed operations on a quantum system, or equivalently different levels of control. This challenge is particular to the quantum regime, where monitoring and manipulating systems generally affects the dynamics. Conceptually different approaches have been pursued in parallel to explore this question. 

\gh{One approach is via} the development of a general theory of quantum thermodynamics \gh{that} aims at placing upper bounds on the performance limits of quantum thermal machines. By establishing fundamental laws, this abstract perspective provides limits that hold for any possible quantum process (hence to all transformations achievable by quantum thermal machines). Typically, such upper bounds are obtained by characterising possible state transitions, focusing on the single-cycle regime. The intuition being that a machine cannot perform better than a perfect cycle. Here one can distinguish two paradigms. In the first, free operations are given by ``thermal operations''~\cite{wit,resource,TO,coherence1,coherence2,coherence3}, i.e. energy conserving unitaries applied to the system and a thermal bath. The implicit assumptions are access to i) a perfect timing device, ii) arbitrary spectra in the bath, and iii) interaction Hamiltonians of arbitrary complexity. This \gh{perspective}  led to derivations of the second law \cite{secondlaws,paul,yelena}---i.e. the removal of system entropy in a thermally equilibrated environment comes at an inevitable work cost---and general formulations of the third law \cite{thirdlaw,thirdlaw2,thirdlaw3}---cooling to temperatures approaching absolute zero requires a diverging amount of resources. In the second paradigm, one considers an increased amount of classical control over a single quantum system, but no access to bath degrees of freedom. I.e. the implicit assumptions are i) a perfect timing device ii) the ability to implement any cyclic change in the Hamiltonian of a quantum system. This led to the concepts of passive states \cite{passive1,passive2,passive3,skrzypczyk15,passive4} and algorithmic cooling \cite{algo1,algo2,raeisi,algo3,algo4} and more generally to fundamental limits on single-cycle performance of coherently driven quantum machines \cite{ticozzi}. 

Another approach is via explicit models of quantum thermal machines that provide lower bounds on their performance. A wide range of such models have been discussed. In general terms, a quantum thermal machine makes use of external resources (e.g. thermal baths) to accomplish a specific task, such as work extraction or cooling. More formally, these machines are modeled as open quantum systems, where the machine consists of few interacting quantum systems coupled to external baths. Performance is usually evaluated in the asymptotic regime of non-equilibrium steady states. Machines with very different levels of control must be distinguished.

Autonomous quantum thermal machines feature the lowest level of control \cite{auto1,auto0,auto2,virtual,venturelli,autoexp1,patrick,autoexp2,roulet}. Here the machine subsystems are coupled to thermal baths at different temperatures, and interact via time-independent Hamiltonians, thus requiring no external source of work or control. In the opposite regime, machines requiring a high level of control have been considered, such as quantum Otto engines \cite{abah,rossnagel,otto1,otto2}. Here one assumes the ability to implement complex unitary cycles, which generally require time-dependent Hamiltonians or well-timed access to a coherent battery \cite{clock1,clock2,coh1}. Nonetheless similar statements of the second and third law are also possible in this regime \cite{thirdlawmachines,silva16}. 

Each of the above approaches represents a perfectly reasonable paradigm for discussing the ultimate limitations of quantum thermodynamics, each featuring its own merits and drawbacks. 
Comparing these approaches is thus a natural and important question. It is however also a challenging one, due to the fact that each approach works within its own respective framework and set of assumptions. 
Recently, several works established preliminary connections between some of these approaches. Refs \cite{transient1,transient2} studied autonomous machines in the transient regime and showed that a single cycle can achieve more cooling than the steady state regime. Quantum machines powered by finite-size baths have been studied \cite{finiteotto} to understand the impact of finite resources, and the control cost of achieving a shortcut to adiabaticity was studied in \cite{shortcutcost}. In \cite{finitesize} the authors explored the implications of finite size systems, i.e. thermal operations \gh{not at} the thermodynamic limit.
In the single-cycle regime, Refs \cite{Gaussianpassive,realresource,passiveinstability} discussed thermodynamic performance under restricted sets of thermal operations, with limited complexity. Finally, even the assumption of perfect timing control, inherent to all paradigms except autonomous machines, should arguably carry a thermodynamic cost \cite{clock3}.

The above paradigms can instructively be split into two types of assumed control over the quantum system. For a single cycle of a thermodynamic process, we can either assume to be capable of engineering time dependent Hamiltonians, dubbed \emph{coherent} control or just turning on time-independent interactions, which we call \emph{incoherent} control. We explicitly model each bath constituent that we have access to and refer to it as machine size. Thus for an infinite machine size, the incoherent control paradigm exactly captures the resource theory of thermodynamics. On the other hand, the explicit modeling of size adds another layer to the analysis of thermodynamic processes in terms of size/complexity.

\begin{table}[h!]
\centering
    \begin{tabular}{|c|c||c|c|}
    \hline
  \multicolumn{2}{|c||}{\backslashbox{complexity}{control}} & Incoherent& Coherent \\  \hline
   \multirow{ 2}{*}{ancillas}& at \(T_H\)& m&0\\ \cline{2-4}
   & at \(T_R\)& n-m&n\\ \hline
   \multicolumn{2}{|c||}{limit \(n \rightarrow \infty\)}& TO& CPTP\\ \hline
    \end{tabular}
    \caption{We here summarize the important properties of both paradigms. By complexity we mean the number of components the machine is allowed to have. Each component is in principle allowed to be a qudit of arbitrary dimension. In the limit of infinitely many ancillas the single cycle incoherent paradigm become the thermal operations (TO) used in the Resource theory of thermodynamics (RTT) and in the single cycle coherent paradigm one is allowed to apply any CPTP map to the target.}
\end{table}

In the accompanying letter \cite{OurLetter}, we used this framework to derive a universal bound for quantum refrigeration and proved that it could be obtained by all types of control, provided that complex enough machines and corresponding interactions are available. In the present work we dig deeper and reveal the intricate connection between machine complexity, control and add the amount of resources consumed in the process to the picture. The latter, in turn, is connected to the entropy change associated with the energy drawn from the resource. Considering our two extremal levels of control; first the \emph{coherent} scenario, where the entropy of the resource is left unchanged, and second the \emph{incoherent} scenario, where only energy at maximum entropy (i.e. heat) is extracted from the resource. Within each level of control, we investigate the lowest attainable temperature, and the work cost for attaining a certain temperature. These quantities allow us to give a direct and insightful comparison between the different approaches for quantum refrigeration.

To tackle these questions, it is natural to consider machines of a given size (i.e. the number of systems that one has access to), since the size in itself also represents a form of control. We analyze this aspect of control starting from the smallest possible machines. It turns out that the two-qubit machine is the smallest one where the coherent and incoherent scenarios \gh{can be compared in a meaningful way.} We also discuss the case of general machines, and in particular the limit of asymptotically large machines.

Our results clearly demonstrate the expected crucial role of control for quantum cooling performance, but surprisingly unifies the different operational approaches through machine complexity.

\section{Setting and summary of results\label{sec:setting}}

\begin{figure}[t]
	\centering
	\includegraphics[width=9cm]{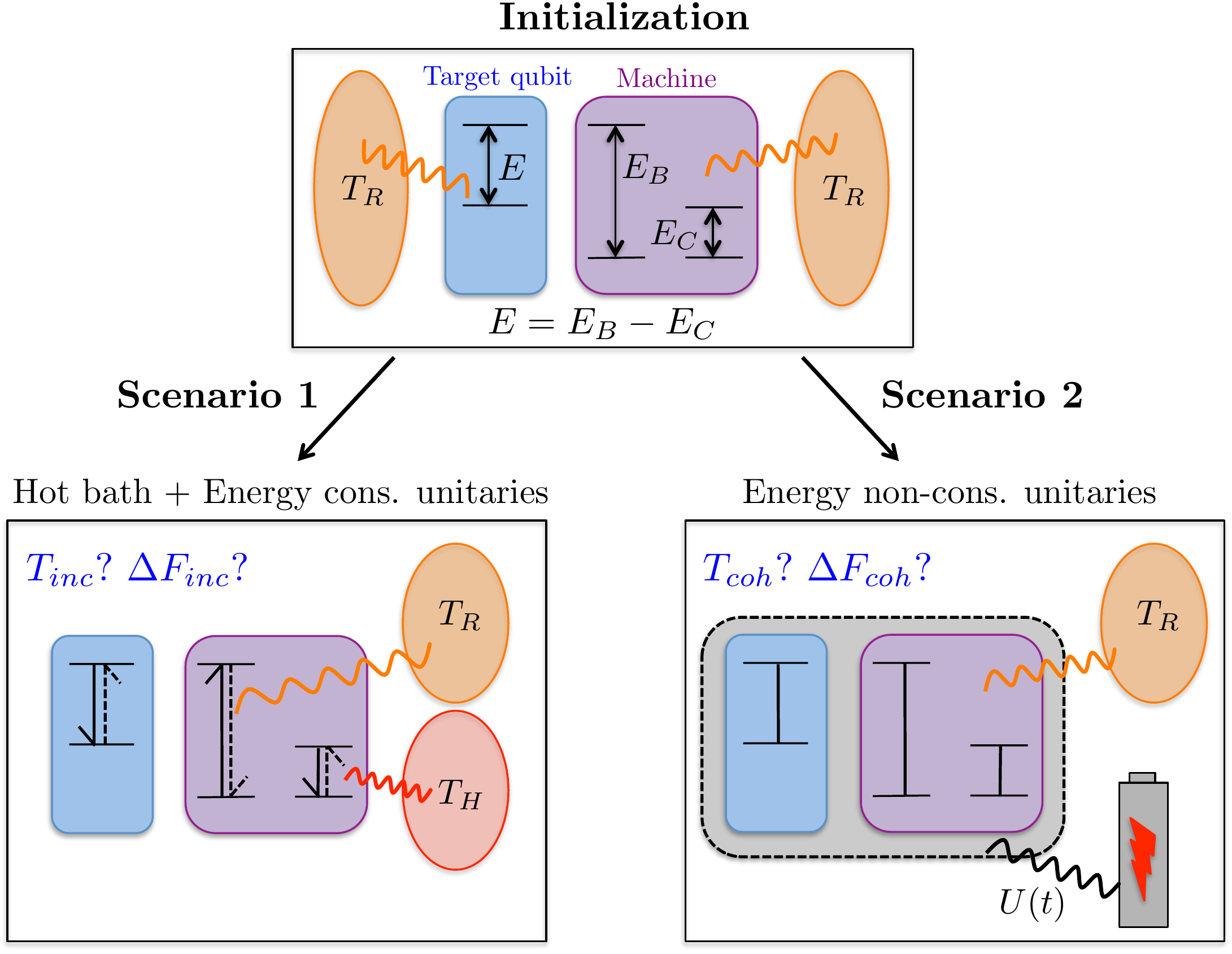}
	\caption{Model for the minimal thermal machine achieving cooling and allowing for the comparison of two paradigmatic scenarios of quantum refrigeration. After initialization of the machine and target qubit with a thermal bath at room temperature $T_R$, two scenarios are proposed. In \gh{scenario} 1, the free energy is provided by a hot bath. This corresponds to a low level of control, i.e. maximal entropy change. In contrast, \gh{scenario} 2 describes a thermal machine requiring a high level of control (e.g. via a coherent battery), that can implement arbitrary unitary operations at zero entropy change.}\label{fig:model}
\end{figure}

\nico{Cooling a quantum system could have several meanings. For a system initially in a thermal state, one can drive it to a thermal state of lower temperature. Alternatively, one could consider increasing the ground-state population, or decreasing the entropy or the energy. These notions are in general inequivalent for target systems of arbitrary dimension. Determining the fundamental limits to cooling is therefore a complex problem in general. It turns out, however, that for the case of a qubit target, all the above notions of cooling coincide. Because of the clarity that this offers but also because the bounds set on target qubits imply bounds for target qudits, see our accompanying article \cite{OurLetter}, we focus on qubit targets only in this article.

Specifically we consider cooling a single qubit which is initially in a thermal state set by the environment temperature $T_R$ and then isolated from any environment \cite{footnote0}. The goal is to increase the ground sate of the qubit (without changing its energy gap).}
In order to cool the target qubit, we couple it to a quantum thermal machine. We consider two scenarios for the operation of this machine, that represent the two extremal levels of control (coherent and incoherent) introduced above. For each of these scenarios we are interested in the limits to cooling performance, see our accompanying article \cite{OurLetter} for a complementary treatment of this,  as well as in the associated work cost. We characterize the work cost by the free energy change, a well-established monotone across thermodynamic paradigms (see e.g. \cite{work,work2}). This quantifies the maximum extractable work from a resource in the presence of an environment at equilibrium, and hence measures to what degree the resource is out of equilibrium with the environment, a property necessary to induce non-trivial transformations of the target system. 

More precisely, the two scenarios are defined as follows:
\begin{itemize}
\item {\bf{Scenario 1: Incoherent operations.}} The source of free energy is a hot bath at a temperature $T_H > T_R$. The machine (or any of its subsystems) can be coupled to the hot bath or rethermalized with the environment at any stage. The machine interacts with the target qubit via an energy conserving unitary operation. The work cost of the operation corresponds to the decrease in free energy of the hot bath.

\item {\bf{Scenario 2: Coherent operations.}} Here the source of free energy is coherent in the sense of allowing for energy non-conserving unitary operations between the machine and the target qubit. This effectively assumes a coherent battery or classical control field as the source of free energy. There is no additional thermal bath, and the machine may only be coupled to the environment (at temperature $T_R$). As the entropy is unaffected, the work cost, i.e. the change in free energy, is simply the change in energy.
\end{itemize}

In order to compare these two scenarios and understand the fundamental limits to cooling performance, we investigate 

\begin{itemize}
\item[i)] the lowest attainable temperature $T^*$,
\item[ii)] the work cost for attaining any given temperature, in particular $T^*$. 
\end{itemize}

In contrast to our accompanying article \cite{OurLetter} where we focus on the unbounded number of cycles regime, we are here interested in the single-cycle, repeated and asymptotic regimes. In the single-cycle regime, an initial thermalisation step is followed by a single unitary operation on the machine and the target qubit (energy conserving or arbitrary, for scenario 1 and 2, respectively). In the repeated operations regime, thermalisation and unitary operations are alternated a finite number of times. In the asymptotic regime, this cycle of steps is repeated indefinitely.

 Turning our attention to the machine more closely, we consider that distinct subsystems of the machine can connect to baths at different temperatures, but we do not allow individual transitions in the machine to be separately thermalised at different temperatures. With that in mind, while bounds on the performance of general machines can be set for both paradigms, see our accompanying article \cite{OurLetter}, the incoherent paradigm is trivial unless the machine has a tensor product structure. Since we are here focusing in comparing both paradigms, in particular with respect to their associated work cost, we will consider machines with such a structure only. Furthermore, besides the more practical aspect of small machines, which are arguably easier to realize, especially in the incoherent scenario where increasing the machine size usually comes at the price of decreased interaction strengths \cite{autoexp1,autoexp2}, they also already suffice to saturate the cooling bounds of each scenario, see our accompanying article \cite{OurLetter}. This as such motivates our interest to focus most our analysis on the minimal settings. The two smallest possible machines consist of either a single qubit or two qubits. Of these, only the latter allows for a non-trivial comparison between the incoherent and coherent scenarios as a single-qubit machine allows for cooling only in scenario 2.

Figure \ref{fig:summary} summarises the results of our comparison, and demonstrates the crucial role of control for the fundamental limits of quantum refrigeration. It shows the minimal achievable temperature of the target qubit vs.~the associated work cost in each scenario and for the single-cycle, finite-repetition, and asymptotic regimes. 

Surprisingly, in the single-cycle regime, we find that neither scenario is universally superior. While scenario 2 always achieves the lowest temperature when no restriction is placed on the work cost, there is a threshold work cost below which scenario 1 outperforms scenario 2. 

For finite repetitions, additional cooling starts from the end points of maximal single-cycle cooling in each scenario. For scenario 1, one can think of this as repeated thermal operations with a locality restriction, i.e. access to a single qubit from each of the two baths in every round, and for scenario 2 it corresponds to multiple cycles of coherently driven quantum machines (such as e.g. quantum Otto cycles). 

In the asymptotic regime scenario 1 corresponds to the minimal autonomous quantum thermal refrigerator, as shown in \cite{Raam} and discussed in our accompanying article \cite{OurLetter}. Scenario 2 \gh{leads to heat bath algorithmic cooling, when augmented with the ability to individually rethermalise the machine qubits to the environment temperature \(T_R\)}. Moreover, like in the single-cycle regime, scenario 2 always achieves a lower temperature, although generally at a higher work cost. 

While minimal machines saturate the cooling bounds, they do so in a very ineffective way from a work cost perspective. Extending our analysis to the case of N-qubit machines, by considering cooling to a fixed target temperature, we finally show that both coherent and incoherent machines can achieve minimal work cost, \textit{i.e} saturate the second law, in the limit of large size.\\

\begin{figure}[t]
	\centering
	\includegraphics[width=9cm]{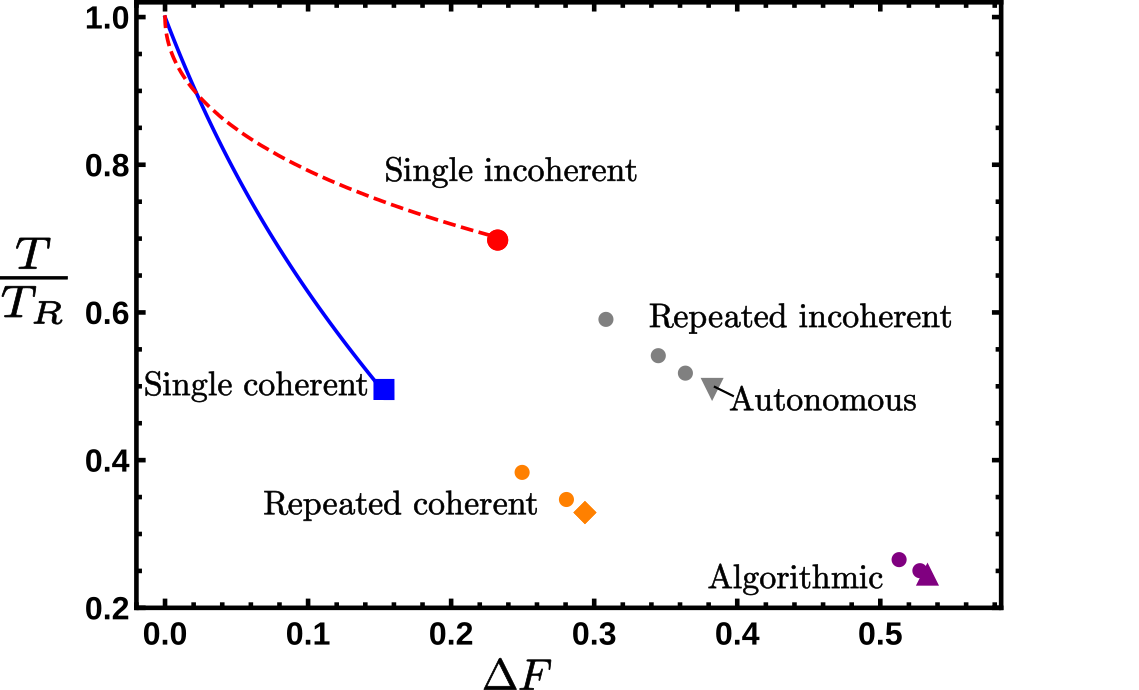}
	\caption{Comparison of achievable temperatures and associated work costs for scenarios 1 and 2 in the single-cycle, finite repetitions, and asymptotic regimes. The ratio $T/T_R$ is the relative cooling, $T$ being the final temperature and $T_R$ the initial one. The symbols (dots, etc) correspond to maximal cooling (i.e. achieving minimal temperature $T^*$) in each scenario. \jona{Here we use $T_R=1$.}
	\label{fig:summary}
	}
\end{figure}

The rest of the paper is organized as follows. In \secref{sec:model}, we introduce notation and definitions. Section \ref{sec:onequbit} deals with the case of the one-qubit machine. In Secs \ref{sec:prelim} and \ref{sec:single-cycle}, we investigate the cooling performance and associated work cost of the two-qubit machine, focusing on the single-cycle regime. In \secref{sec:repeat}, we discuss repeated operations and the asymptotic regime of the two-qubit machine. We then discuss the saturation of the second law by more general machines in \secref{sec:secondlaw} before concluding in \secref{sec:conclusion}.

\gh{
\section{Notation and definitions \label{sec:model}}}

As argued in section \ref{sec:setting}, we consider machines consisting of a given number of qubits. We take the energies of all qubit ground states to be zero, denote the excited state energy of qubit $i$ by $E_i$, and the energy eigenstates by $\ket{0}_i$ and $\ket{1}_i$. Thus, the local Hamiltonian for each qubit is $H_i = E_i \ket{1}_i\bra{1}$, and the total Hamiltonian of target and machine is
\begin{equation}
H = \sum_i E_i \ket{1}_i\bra{1} .
\end{equation}
The initial state, prior to cooling, is the same for the incoherent and coherent scenarios. Every qubit is in a thermal state of its local Hamiltonian at the environment temperature $T_R$. In general, a thermal state of a qubit with energy gap $\varepsilon$ and temperature $T$ is given by
\begin{equation}
\tau(\varepsilon,T) = r(\varepsilon,T) \ket{0}\bra{0} + [1-r(\varepsilon,T)] \ket{1}\bra{1} ,
\end{equation}
where the populations are determined by the Boltzmann distribution (throughout the paper we work in natural units, $k_B=\hbar=1$)
\begin{equation}
\label{eq:bdist}
r(\varepsilon,T) = \frac{1}{1 + e^{-\varepsilon/T}} = \frac{1}{\mathcal{Z}(\varepsilon,T)} ,
\end{equation}
where $\mathcal{Z}(\epsilon,T)$ is the partition function corresponding to the qubit Hamiltonian and temperature.

We denote the ground state populations at the environmental temperature by $r_i = r(E_i,T_R)$, and the corresponding thermal states by $\tau_i$. We will refer to the target to be cooled as qubit $A$, but for convenience, we will generally drop the subscript for the target qubit, such that
\begin{equation}
E \defeq E_A, \hspace{0.5cm} r \defeq r_A, \hspace{0.5cm} \tau \defeq \tau_A .
\end{equation}
Note that we can choose a unit of energy such that \(E=1\) without loss of generality, which we do for all our numerical analysis.

In scenario 1, one (or more) of the machine qubits is first heated to a higher temperature $T_H$. This is followed by an energy conserving unitary acting jointly on the target and the machine, i.e.~any unitary $U$ for which $[U,H]=0$. In scenario 2, an energy non-conserving unitary is applied directly to the initial state of target and machine.

We extract the temperature of the target qubit by reading its ground state population and inverting the relation \eqref{eq:bdist}. When the target qubit is diagonal, which will turn out to be the case for all our relevant operations, the target has a well defined temperature and this is a valid way to extract it. When the target state is not diagonal it strictly speaking has no temperature. One way to nevertheless extend the notion of temperature to these states also is as presented above. The work cost is accounted for from the perspective of the work reservoir, i.e. the free energy change of the resource. This is not necessarily equal to the free energy change of the system itself, but is nonetheless the appropriate way to quantify consumed resources. For completeness, we have also worked out the two scenarios for the two-qubit machine case from a system perspective in \appref{app:internal}.\\

\jona{
\section{One-qubit machine}
\label{sec:onequbit}

Denoting the machine qubit by $B$, the Hamiltonian is $H = H_A + H_B$, and the initial state is \gh{simply}
\begin{equation}
\rho^{in} = \tau \otimes \tau_B .
\end{equation}

\subsection{Scenario 1: incoherent operations}

In this scenario, the machine qubit is first heated to a higher temperature $T_H$, resulting in the state
\begin{equation}
\rho^{H} = \tau \otimes \tau_B^H ,
\end{equation}
where $\tau_B^H = \tau(E_B,T_H)$ is the thermal state of qubit B at the temperature of the hot bath. This is followed by an energy conserving unitary. However there is no such unitary that can cool the target, as we demonstrate now.

For the action of the unitary to be nontrivial (and hence, for any cooling of the target to happen), the spectrum of the joint Hamiltonian $H$ must have some degeneracy, allowing one to shift population between distinct energy eigenstates of the same energy. The only possibilities are that \textit{(i)} one of the energies vanish $E=0$ or $E_B=0$, or \textit{(ii)} the gaps are equal $E=E_B$. In case \textit{(i)}, the thermal state $\rho^H$ will be proportional to the identity in the degenerate subspace, and hence $U\rho^HU^\dagger = \rho^H$ for any energy conserving unitary $U$. In case \textit{(ii)}, because the matrix elements (in the product basis of $H_A$, $H_B$) fulfill $\rho^H_{01,01} = r(1-r_B^H) > r_B^H(1-r) = \rho^H_{10,10}$, unitaries acting on the degenerate subspace can only heat up the target. 

Thus, for the single-qubit machine, cooling is impossible in the incoherent scenario.

\subsection{Scenario 2: coherent operations\label{sec:OneMco}}

Using coherent operations, it is possible to cool, and we now derive the minimal attainable temperature of the target, and the work cost of cooling. This will also provide some intuition for how to tackle the two-qubit machine, where coherent and incoherent cooling can be compared.

Cooling corresponds to increasing the ground state population of the target using an arbitrary joint unitary $U$ on target and machine. This population $r_{coh}$ is given by the sum of the two first diagonal entries of the final state $U\rho^{in}U^\dagger$, when expressed in the product basis of $H_{\text{A}} \otimes H_{\text{B}}$. \fab{From the Schur-Horn theorem one learns that this sum can at most be the sum of the two greatest eigenvalues of \(U \rho^{\text{in}} U^{\dagger}\), which, since \(U\) cannot change the eigenvalues of the state and since \(\rho^{\text{in}}\) is diagonal, are the sum of the two largest diagonal entries of \(\rho^{\text{in}}\). Maximal cooling is thus} achieved when $r_{coh}$ equals the sum of the two largest diagonal entries of $\rho^{in}$. One readily sees that $\rho_{00,00}^{in} = r r_B$ is the largest element and $\rho_{11,11}^{in} = (1-r)(1-r_B)$ the smallest, while
\begin{equation}
\frac{\rho_{01,01}^{in}}{\rho_{10,10}^{in}} = \frac{r(1-r_B)}{(1-r)r_B} = e^{\frac{E-E_B}{T_R}} .
\end{equation}
Cooling is only possible if the initial ground state population $r = \rho_{00,00}^{in} + \rho_{01,01}^{in}$ is not already maximal, i.e.~if $E < E_B$. In this case, the maximal final population is $r_{coh}^* = \rho_{00,00}^{in} + \rho_{10,10}^{in} = r_B$ corresponding to (from \eqref{eq:bdist})
\begin{equation}
T_{coh}^* = \frac{E}{\ln\left(\frac{r_{coh}^*}{1-r_{coh}^*}\right)} = \frac{E}{E_B} T_R .
\end{equation}

This temperature can be achieved by a unitary which swaps the states $\ket{01}$ and $\ket{10}$, and in fact this also minimises the associated work cost. More generally, we can identify an optimal unitary which minimises the work cost of cooling to any temperature in the attainable range, i.e.~any ground state population $r_{coh}$ between $r$ and $r_{coh}^*$. The optimal work cost is given by
\begin{equation}
\label{eq:workcostonequbitcoh}
\Delta F_{coh} = (r_{coh} - r)(E_B - E) ,
\end{equation}
and it is achieved by a unitary of the form
\begin{equation}
\label{eq:optimalUonequbitcoh}
U = e^{-i t L} ,
\end{equation}
where
\begin{equation}
\label{eq:Lonequbit}
L = i \ket{01}\bra{10} - i \ket{10}\bra{01} 
\end{equation}
is a Hamiltonian which generates swapping of excitations between the target and machine qubits, and $t = \arcsin(\sqrt{\mu})$ with
\begin{equation}
\mu = \frac{r_{coh} - r}{r_B - r} .
\end{equation}

The optimality of \eqref{eq:workcostonequbitcoh} can be proven using the Shur-Horn theorem and \fab{majorization} \cite{MarshalMajorization, NielsenMajorization}. \gh{The idea of the proof is as follows.} \fab{By scanning through all the unitarily attainable \(\rho^{\text{coh}}= U \rho^{\text{in}} U^{\dagger}\) we are looking at all the Hermitian matrices \(\rho^{\text{coh}}\) with spectrum \(\vv{\rho}^{\text{in}}\) (given an \(n \times n\)  matrix \(\mu=(\mu_{ij})\) we generically denote its vectorized diagonal \((\mu_{11},\dots,\mu_{nn})\) by \(\vv{\mu}\)). According to the Schur-Horn theorem there exists a Hermitian matrix \(\rho^{\text{coh}}\) with spectrum \(\vv{\rho}^{\text{in}}\) if and only if the majorization condition \(\vv{\rho}^{\text{coh}} \prec \vv{\rho}^{\text{in}}\) holds. Hence, a state $\rho^{coh}$ is reachable by a unitary starting from $\rho^{in}$ if and only if the diagonals fulfill \(\vv{\rho}^{\text{coh}} \prec \vv{\rho}^{\text{in}}\). In the coherent scenario, the free energy difference and hence the work cost is simply given by
\begin{align}
	\Delta F = \Tr \left[(\rho^{coh} - \rho^{in}) H\right] = (\vv{\rho}^{\text{coh}} - \vv{\rho}^{\text{in}})\cdot \vv{H},
\end{align}
where $\vv{H}$ is the diagonal of $H$. The last term is constant, and so minimising the work cost for a given final $r_{coh}$ is equivalent to 
\begin{equation}
\label{eq:workminimisation_onequbit}
\underset{\vv{\rho}\prec\vv{\rho}^{\text{in}}}{\min} \, \vv{\rho}\cdot \vv{H} \hspace{0.5cm} \text{s.t.} \hspace{0.5cm} 
\rho_1 + \rho_2  = r_{coh}.
\end{equation}
As shown in \appref{app:Uopt}, this minimisation can be solved analytically, leading to \eqref{eq:workcostonequbitcoh} and \eqref{eq:optimalUonequbitcoh}.
}}

\jona{

\gh{
\section{Two-qubit machine: Model}
\label{sec:prelim}}

\gh{When considering the two-qubit machine, the total Hamiltonian of target and machine} is $H = H_A + H_B + H_C$, with qubits B and C forming the machine. \gh{The setup, as well as the two scenarios, are illustrated in \figref{fig:model}}. The starting point for both scenarios 1 and 2 is the initial state
\begin{equation}
\label{eq:initialstate}
	\rho^{\text{in}} = \tau \otimes \tau_B \otimes \tau_C.
\end{equation}

In scenario 2, an energy non-conserving unitary is applied directly to the initial state $\rho^{\text{in}}$, while in scenario 1, qubit C is first heated to a higher temperature $T_H$, resulting in the state 
\begin{eqnarray}
	&&\rho^H = \tau \otimes \tau_B \otimes \tau_C^H,
\end{eqnarray}
where $\tau_C^H = \tau(E_C,T_H)$ is the thermal state of qubit C at the temperature of the hot bath. This is followed by an energy conserving unitary acting on the three qubits. To allow for non-trivial energy conserving unitaries, there must be a degeneracy in the spectrum of $H$ with an associated degenerate subspace. In \appref{app:degeneracy}, we show that the only degeneracy which enables cooling of the target is obtained by setting
\begin{equation}\label{eq:deg}
E = E_B - E_C\,.
\end{equation}
Hence, we work with this convention throughout the following.}\\
\\


\gh{
\section{Two-qubit machine: single-cycle regime}
\label{sec:single-cycle}}

In this section, we discuss the single-cycle regime \gh{of the two-qubit machine}. 
We show that scenario 2 (coherent operations) always reaches lower temperatures when the work cost is unrestricted. However, for sufficiently low work cost, it turns out that scenario 1 (incoherent operations) outperforms scenario 2. \\

\subsection{Scenario 1: incoherent operations}
\label{sec:single-cycle1}

We first identify the energy-conserving unitary that is optimal for cooling the target qubit. From the relation \eqref{eq:deg} it follows that there is only one subspace that is degenerate in energy (relevant for cooling), which is spanned by the states $\ket{010}$ and $\ket{101}$. Optimal cooling is simply achieved by swapping these two states, i.e. the unitary is given by (see \appref{app:degeneracy} for more details)
\begin{equation}\label{eq:u_incoh}
	U =  \ket{010} \bra{101} + \ket{101} \bra{010} + \mathds{1}_{\text{non-deg}},
\end{equation}
where $\mathds{1}_{\text{non-deg}}$ is the identity operation on the complement space. We can thus directly compute the final temperature of the target qubit. We first compute the final ground state population $r_{\text{inc}}$

\begin{eqnarray}
\label{eq:rSf}
r_{\text{inc}} (T_H) = r r_B + [(1-r) r_B + r (1-r_B)] (1-r_C^H),
\end{eqnarray}
where $r_C^H = r(E_C,T_H)$ denotes the ground state population of qubit C after heating and $r$ and $r_B$ denote the ground state populations of the target qubit and qubit B at room temperature $T_R$. The final temperature is found by inverting Eq.~\eqref{eq:bdist}
\begin{equation}
\label{eq:TSf}
T_{\text{inc}} (T_H) = \frac{E}{\ln(\frac{r_{\text{inc}}}{1-r_{\text{inc}}})}.
\end{equation}
Not limiting the work cost, optimal cooling is obtained in the limit $T_H \rightarrow \infty$. In this case $r_C^H = \frac{1}{2}$, and thus
\begin{equation}
	r_{\text{inc}}^* = \lim_{T_H \rightarrow \infty} r_{\text{inc}} (T_H) = \frac{1}{2}(r+r_B).
\end{equation}
We thus obtain the lowest achievable temperature for scenario 1:
\begin{equation}
\label{eq:Tminincoh}
T_{\text{inc}}^*  = \lim_{T_H \rightarrow \infty} T_{\text{inc}} (T_H) =  \frac{E}{\ln(\frac{r+r_B}{2-(r+r_B)})}.
\end{equation}

We are now interested in the work cost of cooling. For scenario 1, the hot bath is the only resource, implying that the free energy decrease in the hot bath represents the cooling cost. The free energy difference is $\Delta F = \Delta U - T_R \Delta S$, where \(\Delta U\) is the internal energy change. For a thermal bath \(\Delta U\) is defined as the heat drawn from the bath, \(Q\), which from the first law equals the change in energy of qubit C. We follow the convention of counting as positive what is taken from the bath. The change in entropy $\Delta S$ also takes a simple form for a thermal bath, $\Delta S = Q/T_H$. This gives
\begin{equation}
\label{eq:deltaF1}
\begin{aligned}
\Delta F_{\text{inc}} (T_H) &= Q (1-\frac{T_R}{T_H}) \\
&= E_C (r_C - r_C^H) (1-\frac{T_R}{T_H}) .
\end{aligned}
\end{equation}
The above equation shows that the work cost is determined directly by the hot bath temperature $T_H$. The work cost associated to maximal cooling is given by
\begin{equation}
\label{eq:deltaFmaxincoh}
\Delta F_{\text{inc}}^*= \lim_{T_H \rightarrow \infty} \Delta F_{\text{inc}}(T_H) = E_C (r_C-\frac{1}{2}). 
\end{equation}
 Note that despite appearances, the above expression is not independent of $E_B$, as the machine qubits are mutually constrained by the degeneracy condition \eqref{eq:deg}.

More generally, as the ground state population $r_{\text{inc}}$ is monotonic in $r_C^H$, see Eq.~\eqref{eq:rSf}, and thus in $T_H$, one can cool to any temperature between $T_R$ and $T_{\text{inc}}^*$ by varying $T_H$ continuously between $T_R$ and infinity. The associated work cost is given by Eq.~\eqref{eq:deltaF1}; see \figref{fig:Tf_single}.

Note that the minimum achievable temperature in this scenario is lower bounded away from absolute zero. Taking the limits $T_H \rightarrow \infty$ and then $E_B \rightarrow \infty$, $r_{\text{inc}}$ tends to $(1+r)/2$. The work cost diverges in this limit. This is in contrast to scenario 2 presented in the following section, where for an unbounded work cost, one can cool arbitrarily close to absolute zero.

\subsection{Scenario 2: coherent operations}\label{sec:single-cycle2}

We now turn to the second scenario, where any joint unitary operation can be applied to the target and machine qubits. The freedom in unitary operation means that the resonance condition $E_B = E + E_C$ is in principle not required to allow cooling, in contrast to scenario 1. However, as the cooling in either scenario depends on the choice of machine qubits, the freedom to choose them represents an extra level of control. In order to make a meaningful comparison between coherent and incoherent operations, we will therefore enforce the resonance condition for scenario 2 as well. 

We first investigate the lowest achievable temperature. By definition this is obtained by maximizing the ground state population of the target qubit. If we express the state of all three qubits as a density matrix $\rho$ in the energy eigenbasis, then the initial state is seen to be diagonal from Eq.~\eqref{eq:initialstate} and the reduced state of the target is given by $\text{Tr}_{BC}(\rho)$. Its ground state population is then simply given by adding the populations (diagonal elements) of the 4 following states: $\{\ket{000},\ket{001},\ket{010},\ket{011}\}$. \fab{Making use of the Schur-Horn theorem as argued in \secref{sec:OneMco} one reaches optimal cooling by unitarily rearranging the populations such that the four largest populations of the initial state are mapped to the four levels contributing to the ground state population of the target.

\gh{Labeling the population of the state $\ket{ijk}$ in the corresponding initial density operator \(\rho^{\text{in}}\) by \(p_{ijk}\)}, and arranging them in decreasing order of magnitude, we find}
\begin{align}
p_{000} \!>\! \{ p_{001}, p_{100} \} \!>\! p_{010} \!=\! p_{101} \!>\! \{ p_{011}, p_{110} \} \!>\! p_{111},
\end{align}
where $\{\}$ denotes populations whose ordering depends on whether $E_C > E$ or $E_C < E$. Thus the only change necessary to optimize cooling is to swap the populations of $\ket{100}$ and $\ket{011}$, and this leads to a final ground state population of $r_{\text{coh}}^* = r_B$, corresponding to the remarkably simple final temperature
\begin{equation}\label{eq:Tmincoh}
T_{\text{coh}}^* = T_R \frac{E}{E_B}.
\end{equation}
This is the lowest achievable temperature in scenario 2, when the work cost is unrestricted.\\

We now turn to the question of optimizing the work cost. Indeed, on inspection of the end point of the above procedure, one finds that within the ground and excited subspaces of the target qubit, one can perform unitaries that rearrange populations without affecting cooling, but that extract energy back from the system, hence decreasing the work cost of the cooling procedure.

We illustrate this subtlety with the end point of the simple swap above. The only modified populations after the swap are those of the states $\ket{100}$ and $\ket{011}$. Denoting the new population of energy level \(\ket{ijk}\) by $p_{ijk}^\prime$, we have that $p_{011}^\prime = p_{100}$ and $p_{100}^\prime = p_{011}$, with the rest unchanged. Thus the new ordering is
\begin{align}
p^\prime_{000} \!>\! \{ p^\prime_{001}, p^\prime_{011} \} \!>\! p^\prime_{010} \!=\! p^\prime_{101} \!>\! \{ p^\prime_{100}, p^\prime_{110} \} \!>\! p^\prime_{111}.
\end{align}
Although the ground state population is maximized by this swap, one sees that its energy is not minimal, since e.g. $p^\prime_{011} > p^\prime_{010}$.  As a consequence, one could now extract energy without changing the ground state population by simply swapping the levels $\ket{011}$ and $\ket{010}$. Formally, this implies that within each subspace of the target qubit (ground and excited), the state is not \emph{passive} \cite{passive1}, i.e. the populations are not ordered in decreasing order in energy within each subspace.
\fab{This showcases the general fact  that the state the optimal unitary drives the system to, necessarily has to be passive within each of these subspaces. In the maximal cooling case, as shown in \appref{subsec:endpoint}, this passivity condition remarkably turns out to also be sufficient.}
If one thus follows performing the unitary that reorders each subspace to be passive, and subtracts the energy extracted from the work cost, we arrive at the optimal work cost corresponding to maximal cooling, $\Delta F_{\text{coh}}^*$. \\

We find that there are two cases. If $E_C \leq E$, then
\begin{align}
\label{eq:deltaFmaxcohBig}
\Delta F_{\text{coh}}^* &= E_C \left( r_B - r \right).
\end{align}
Note that this end point can be achieved by simply performing the unitary that swaps the states of qubits $A$ and $B$. On the other hand, if $E_C > E$, then
\begin{align}
\label{eq:deltaFmaxcohSmall}
\Delta F_{\text{coh}}^* &= \left( E_C - E \right) \left( r_C - r \right) + E_C \left( r_B - r_C \right).
\end{align}
The unitary that achieves this result is the sequence of two swaps - first the swap between the target and qubit $C$, followed by the swap between the target and qubit $B$.

Remarkably, these two expressions can be intuitively understood in the following manner. In order to achieve cooling on the target qubit, one would swap its state with a qubit of the machine (or qubit subspace from the machine, also called a "virtual qubit" \cite{virtual}, see \appref{app:virtualqubit}, that has a larger energy gap between its ground and excited states than the target qubit. However, doing so requires moving population against the energy difference between the target and the specific machine qubit. Minimizing the work cost of the cooling procedure therefore amounts  to swapping the state of the target qubit with the state of the machine qubit with the minimal energy gap as long as this one is bigger than the energy $E$ of the target qubit.

If $E_C \leq E$, then the smallest qubit subspace of the machine that has a higher energy gap than the target is qubit $B$, and the optimal procedure is to swap the states of those two qubits. This has a work cost $E_B - E = E_C$ per population. In contrast, when $E_C > E$, then qubit $C$ is the machine qubit with the smallest energy gap bigger than $E$ ($E_C < E_B$ by definition). We thus begin by swapping the target qubit with qubit $C$, at a work cost per population of $E_C - E < E_C$, and only after proceed to cool further by swapping the target qubit with qubit $B$, at higher work cost. This two cases respectively lead to Eqs.~\eqref{eq:deltaFmaxcohBig} and \eqref{eq:deltaFmaxcohSmall} when the work cost is unrestricted.

We now move to the case where the work cost is restricted. Equivalently, we consider the problem of cooling to a certain temperature (above $T_{\text{coh}}^*$), and derive the minimal associated work cost.  Intuitively, as the lowest temperature given by Eq.~\eqref{eq:Tmincoh} can be reached by a full swap (or a sequence of two full swaps if $E_C > E$), we might expect that an optimal strategy for reaching an intermediate temperature will be a partial swap. 

{This is indeed the case. In analogy with \eqref{eq:workminimisation_onequbit}, the minimal work cost for a given target temperature $T_{coh}$ and corresponding ground-state population $r_{coh}$ is given by
\begin{equation}
\label{eq:workminimisation_onequbit}
\underset{\vv{\rho}\prec\vv{\rho^{\text{in}}}}{\min} \, \vv{\rho}\cdot \vv{H} \hspace{0.5cm} \text{s.t.} \hspace{0.5cm} 
\sum_{i=1}^4\rho_i  = r_{coh},
\end{equation}
where $\vv{\rho}$, $\vv{\rho^{\text{in}}}$, and $\vv{H}$ represent the diagonals of $\rho$, $\rho^{in}$, and $H$. This minimisation can be solved analytically, as shown in \appref{app:Uopt}. The optimal unitary and associated work cost depends on whether $E_C \leq E$ or $E_C > E$.

For the case $E_C \leq E$, we can parametrise a partial swap of the target with machine qubit B as in \eqref{eq:Lonequbit} by}
\begin{equation}
U_\leq(\mu) = e^{-i t L_{AB}} ,
\end{equation}
where
\begin{equation}
\label{eq:LAB}
L_{AB} = i \ket{01}_{AB}\bra{10} - i \ket{10}_{AB}\bra{01} 
\end{equation}
It is useful to define $t = \arcsin(\sqrt{\mu})$, where $\mu\in[0,1]$ is a swapping parameter. The ground state population of the target qubit and the free energy cost are given by
\begin{align}
	r_{\text{coh},\leq} (\mu) &= r + \mu \left( r_B - r \right), \label{eq:rAf2leq}\\
	\Delta F_{\text{coh},\leq} (\mu) &= \mu E_C \left( r_B - r \right),\label{eq:deltaF2leq}
\end{align}
with $\mu=0$ corresponding to no swap and $\mu=1$ to a full swap, which is the limit of maximal cooling, as previously discussed, see Eqs.~\eqref{eq:Tmincoh} and \eqref{eq:deltaFmaxcohBig}.

Similarly, for the case $E_C > E$ we employ the unitary that first swaps qubits $A$ with $C$ until the required temperature is reached, and if this is not the case after the full swap, continue by swapping the new state of qubit $A$ with qubit $B$. This unitary can be parametrised as
\begin{equation}
	U_>(\mu) =  e^{-i g(\mu) L_{AB}}  e^{-i f(\mu) L_{AC}} ,
\end{equation}
\fab{where $f(\mu) = \arcsin(\sqrt{\min\{2\mu,1\}})$, $g(\mu) = \arcsin(\sqrt{\max\{2\mu-1,0\}})$,} and $L_{AC}$ is definied analogously to Eq.~\eqref{eq:LAB}. Again $\mu\in[0,1]$ such that for $\mu \leq \frac{1}{2}$, a partial swap between A and C is performed and for $\frac{1}{2}<\mu\leq 1$, an additional partial swap between A and B is performed. The ground state population for the strategy defined by $U_>$ is:
\begin{align}\label{eq:rAf2gtr}
	r_{\text{coh},>}(\mu) &= \begin{cases}
		r + 2\mu(r_C-r), 					 &  \mu\in[0,\frac{1}{2}] \\
		r_C + (2\mu-1)(r_B-r_C),  & \mu\in(\frac{1}{2},1]
	\end{cases},
\end{align}
and the work cost for the same strategy is given by
\begin{align}\label{eq:deltaF2gtr}
	\Delta F_{\text{coh},>}(\mu) &= \begin{cases}
		2\mu (E_C-E) (r_C-r),				 & \mu\in[0,\frac{1}{2}] \\
\begin{aligned}
&(E_C - E) (r_C-r)\\
&+ (2\mu-1) E_C (r_B-r_C)
\end{aligned},			& \mu\in(\frac{1}{2},1]
	\end{cases}.
\end{align}

The final temperature can again be computed by inverting Eq.~\eqref{eq:bdist} using the ground state population $r_{\text{coh}}$ as given by Eq.~\eqref{eq:rAf2leq} or Eq.~\eqref{eq:rAf2gtr} according to the relative size of $E$ and $E_C$. Since both $\Delta F_{\text{coh}}$ and $T_{\text{coh}}$ are given as functions of $\mu$, by varying $\mu$ from 0 to 1, we can parametrically map out the amount of cooling and the associated work cost,  as shown in \figref{fig:Tf_single} and discussed in \secref{subsec:comparison}.

\subsection{Comparison of scenarios 1 and 2}
\label{subsec:comparison}
Our main results in the single-cycle regime are summarised in \figref{fig:Tf_single}. There we map out the amount of cooling vs.~the associated work cost for both scenarios 1 and 2. In the first case, the curve is generated from Eqs.~\eqref{eq:rSf} and \eqref{eq:deltaF1} (inverting Eq.~\eqref{eq:bdist} to extract $T_{\text{inc}}$) and is parametric in the hot bath temperature $T_H$. In the second case, the curve is generated from Eqs.~\eqref{eq:deltaF2leq} and \eqref{eq:rAf2leq} (inverting Eq.~\eqref{eq:bdist} to extract $T_{\text{coh}}$) and is parametrised by the swapping parameter $\mu$. We selected \(E_C \leq E\) for \figref{fig:Tf_single} but note that the behavior of the curve for \(E_C > E\) is similar, changing only by the fact that the coherent curve has a discontinuity in the first derivative at \(\mu = \frac{1}{2}\).

\begin{figure}[t]
	\centering
	\includegraphics[width=8.5cm]{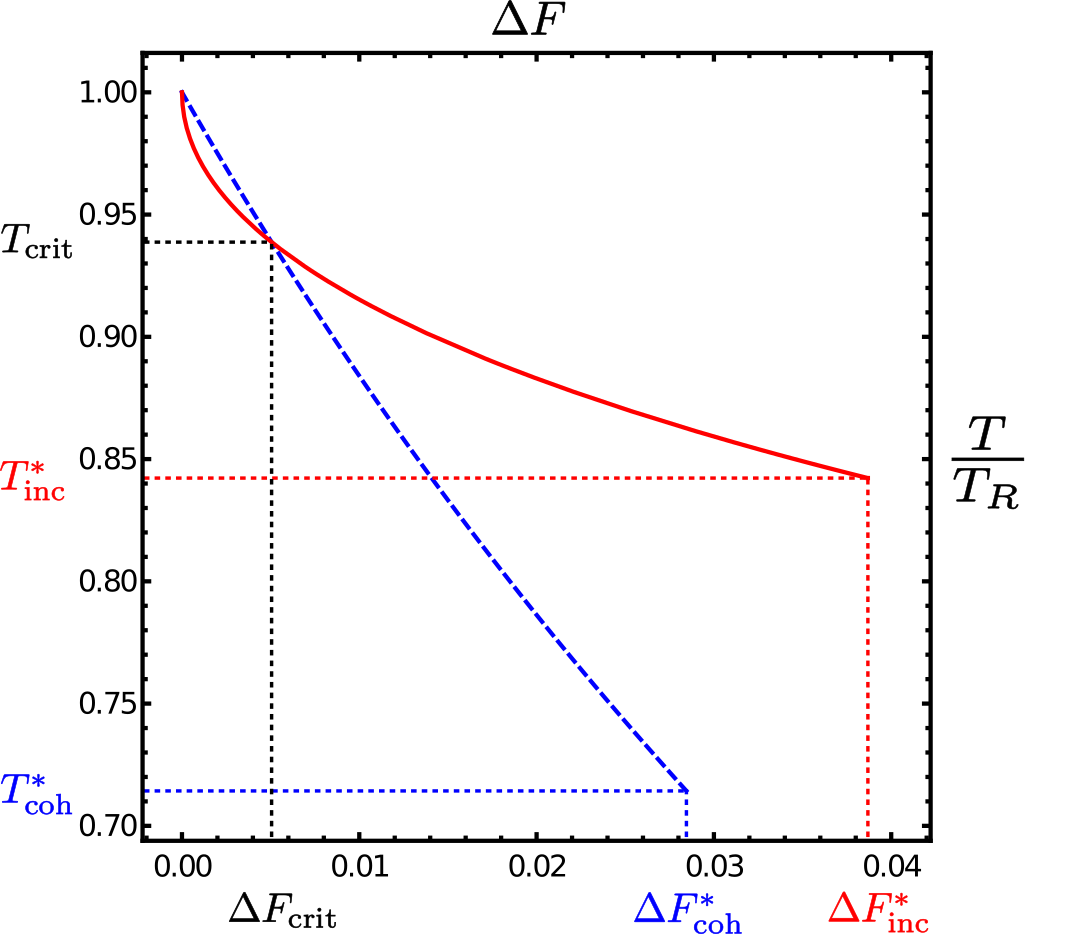}
\caption{Parametric plot of the relative temperature of the target qubit \(\frac{T}{T_R}\) as a function of its work cost $\Delta F$ for $E_C=0.4$ and \(T_R=1\). The red solid curve corresponds to scenario 1 (incoherent operations), the blue dashed, to scenario 2 (coherent operations). When the cooling is maximal (i.e. the work cost is unrestricted), scenario 2 always outperforms scenario 1, $T_{\text{coh}}^{*} < T_{\text{inc}}^{*}$ and $\Delta F_{\text{coh}}^* < \Delta F_{\text{inc}}^*$. However, below a critical work cost $\Delta F_\text{crit}$, scenario 1 always outperforms scenario 2.\label{fig:Tf_single}}
\end{figure}

The plot illustrates several interesting observations. First, comparing the endpoints of the curves, we see that coherent operations achieve a lower minimal temperature (i.e.~stronger cooling) and that the associated work cost is lower than the one for achieving the minimal temperature with incoherent operations. This is true generally. As can be seen by comparing Eqs.~\eqref{eq:Tminincoh} and \eqref{eq:Tmincoh}, $T_{\text{coh}}^* < T_{\text{inc}}^*$ since
\begin{equation}
	\ln \left( \frac{r+r_B}{2-(r+r_B)} \right) < \frac{E_B}{T_R},
\end{equation}
where we use that $E_B>E$. Similarly, comparing Eqs.~\eqref{eq:deltaFmaxincoh} and \eqref{eq:deltaFmaxcohBig}, we see that $\Delta F_{\text{coh}}^* < \Delta F_{\text{inc}}^*$, see \appref{app:endpoint}. Thus, for maximal cooling, coherent operations always perform better than incoherent ones in the single-cycle regime.

Second, perhaps surprisingly, for non-maximal cooling with low work cost, incoherent operations may outperform coherent ones. In fact, for sufficiently low work cost, this is always the case. This can be seen by looking at the derivatives of the two curves in \figref{fig:Tf_single} with respect to $\Delta F$, close to $\Delta F = 0$. For the incoherent scenario, using the parametrization w.r.t. $T_H$, we have
\begin{equation}
	\lim_{\Delta F_{\text{inc}} \rightarrow 0} \frac{dT_{\text{inc}}}{d\Delta F_{\text{inc}}} = \lim_{\Delta F_{\text{inc}} \rightarrow 0} \frac{dT_{\text{inc}}}{dT_H} \left( \frac{d \Delta F_{\text{inc}}}{dT_H} \right)^{-1} = -\infty .
\end{equation}
On the other hand, for the coherent scenario, using the parametrization in terms of $\mu$, we find that
\begin{equation}
\lim_{\Delta F_{\text{coh}} \rightarrow 0} \frac{dT_{\text{coh}}}{d\Delta F_{\text{coh}}} = - \frac{1}{E_C' r (1-r)\ln^2(\frac{1-r}{r})},
\end{equation}
where $E_C' = E_C$ for $E_C\leq E$ and $E_C'=E_C-E$ if $E > E_C$. This expression is negative but finite. Hence, since both curves begin at the same point, the incoherent curve must lie below the coherent one for sufficiently small $\Delta F$. From the previous observations, it follows that the curves must cross at least once. Numerically we find that there is always exactly one such crossing. \fab{Hence, there exists a critical work cost $\Delta F_{\text{crit}}$ below which incoherent \gh{operations} perform better than coherent \gh{ones,} while the reverse is true above some \(\Delta F_{\text{crit}}' \geq \Delta F_{\text{crit}}\), with $ \Delta F_{\text{crit}}'=\Delta F_{\text{crit}}$ numerically strongly supported to be true}. We denote the temperature of the target qubit at the crossing point by $T_{\text{crit}}$. In \appref{app:cross} we study the behaviour of $T_{\text{crit}}$ and $\Delta F_{\text{crit}}$ as functions of $T_R$ and $E_C$.

\gh{
\section{Two-qubit machine: Repeated operations and asymptotic regime}
\label{sec:repeat}
}

\gh{In this section we go beyond the single-cycle regime discussed above. 
In the repeated and asymptotic regimes, the cooling unitaries of either scenario can be repeated a finite number of times \gh{or indefinitely}, inter-spaced by steps in which the machine qubits (B and C) are rethermalised to the temperatures of their baths, \textit{i.e.} \gh{respectively} $T_R$ and $T_H$ in scenario 1 and $T_R$ for both \gh{machine qubits} in scenario 2. 

The target qubit is assumed not to rethermalise during the cooling process. In this way, the bounds we obtain on achievable temperature and work cost are general. Moreover, these bounds can be attained in the limit where the thermal coupling of the target qubit is much smaller than other couplings in the system. } \\

Before going into details, we first summarize the main results of this section.\\

\begin{enumerate}
    \item Repeated operations do enhance the cooling, as the lowest achievable temperatures in both scenarios are strictly lower than in the single-cycle case.
    \item For incoherent operations (scenario 1), the asymptotic regime (the limit of infinite repetitions) corresponds to autonomous refrigeration. Specifically, we recover the cooling and work cost obtained in the steady-state of a three-qubit autonomous refrigerator \cite{auto0,footnote1}.
    \item For coherent operations (scenario 2), the asymptotic regime corresponds to algorithmic cooling. In particular, the cooling bounds correspond to known results \cite{raeisi,algo3}.
    \item In the asymptotic regime, incoherent operations (scenario 1, autonomous cooling) achieve the same maximal cooling (for $T_H \rightarrow \infty$) as that of a single-cycle coherent operation (scenario 2). See our accompanying article \cite{OurLetter} for more details on this relation. 
    \item In both scenarios, the approach to the asymptotic state of the target qubit (w.r.t. its ground state population) is exponential in the number of repetitions.
\end{enumerate}

In the following, we will start by discussing repeated operations in scenario 1 and then move to scenario 2.

\subsection{Scenario 1: repeated incoherent operations} \label{sec:repeatinc}

\begin{figure}[t]
	\centering
	\includegraphics[width=9cm]{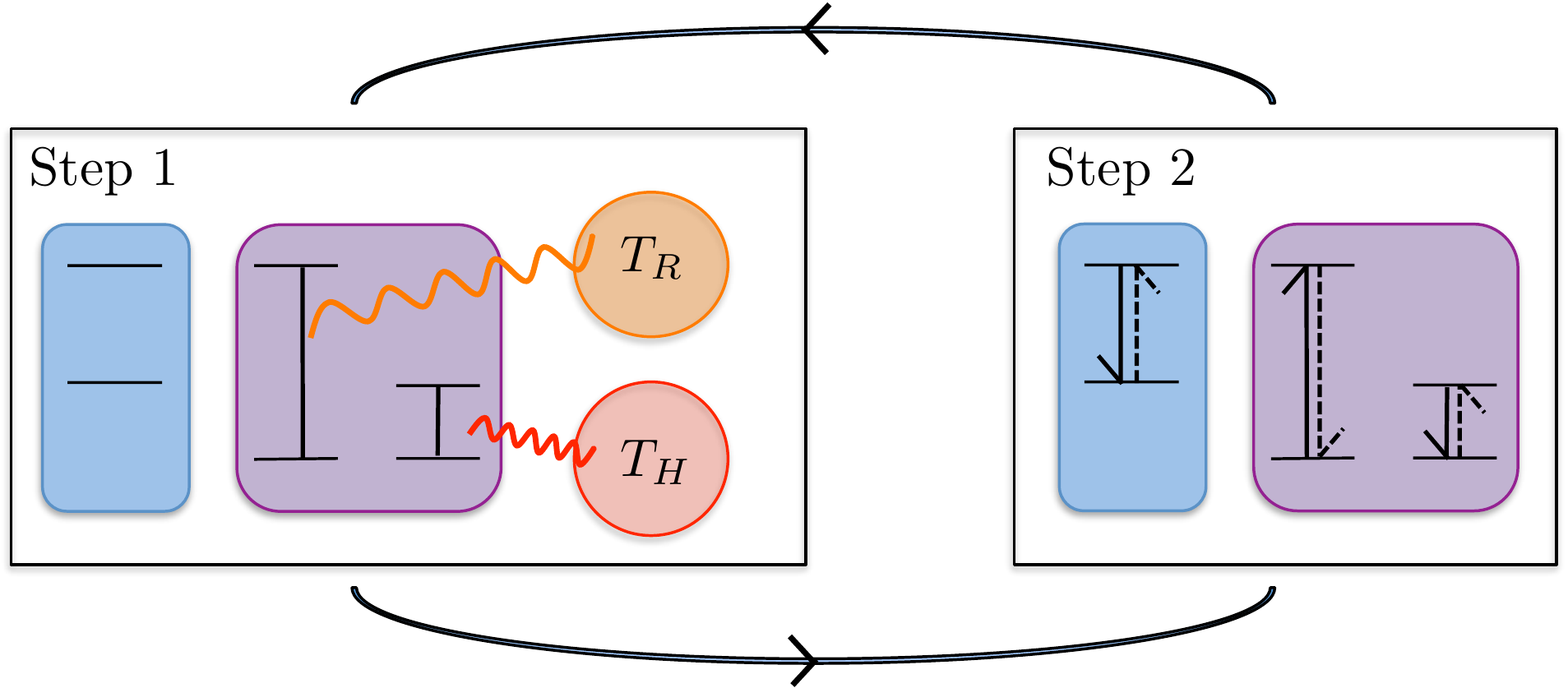}
	\caption{Scenario 1, repeated incoherent operations. Each cycle comprises the steps of 1. the environment reset of qubit $B$ and resource input into qubit $C$, and 2. the cooling unitary operation.\label{fig:repeat_inc}}
\end{figure}

As mentioned above, the scenario of repeated incoherent operations involves a rethermalisation of the machine qubits to their respective baths in every step. This is followed by an energy-conserving unitary operation between the machine and the target. Thus, the cooling cycle consists in the following steps (see \figref{fig:repeat_inc}), which can be repeated any number of times.
\begin{enumerate}
    \item \emph{Environment reset and resource input - } \fab{Qubit C is heated to \(T_H\) after the machine \gh{has} been brought back to the environment temperature \(T_R\).}
    \item \emph{Cooling step - } The energy-preserving unitary given by Eq.~\eqref{eq:u_incoh} (swapping the degenerate states $\ket{101} \leftrightarrow \ket{010}$) is applied.
    \end{enumerate}

Prior to the first step, all three qubits are at temperature $T_R$. Then qubit C is heated to $T_H$. After this, every cooling step lowers the temperature of the target qubit $A$, but also cools down qubit $C$ while heating qubit $B$, which necessitates the reset of $B$ to $T_R$ and the heating of $C$ to $T_H$ before the swap can be repeated. This process can be conveniently characterized using the notion of a virtual qubit, \cite{virtual}. The virtual qubit corresponds to the subspace of the machine which is involved in the cooling swap with the target qubit. See \appref{app:virtualqubit} and \appref{app:incohop}, for a detailed explanation. It is thus the properties of the virtual qubit that determine the cooling in each step. For the unitary operation here, the virtual qubit is spanned by the states $\{\ket{01}_{BC},\ket{10}_{BC}\}$. In each repetition, the rethermalisation of qubits $B$ and $C$ (Step 1) resets the virtual qubit.

In the asymptotic limit of infinite repetitions, we find that the ground state population of the target goes to (see \appref{app:incohop}) 
\begin{align} \label{equ:rincinf}
r_{\text{inc},\infty} = \frac{1}{1 + e^{-E/T_{\text{inc},\infty}}}, 
\end{align}
where $T_{\text{inc},\infty}$ is equal to the temperature of the virtual qubit,
\begin{align}\label{eq:TVinc}
	T_{\text{inc},\infty} = T_{V,\text{inc}} = \frac{E}{\frac{E_B}{T_R} - \frac{E_C}{T_H}}.
\end{align}

For a finite number $n$ of repetitions, the ground state population of the target qubit approaches the asymptotic value as
\begin{align}
\label{eq:incoherentnstep}
    r_{\text{inc,n}} &= r_{\text{inc},\infty} - \left( r_{\text{inc},\infty} - r \right) \left( 1 - N_{V,\text{inc}} \right)^n,
\end{align}
where $N_{V,\text{inc}} = r_B (1-r_C^H)+(1-r_B) r_C^H$ is the norm of the virtual qubit (i.e. the total population in the subspace $\{\ket{01}_{BC},\ket{10}_{BC}\}$). Note that all of the quantities in the above expressions are functions of $T_H$.

As argued also in \appref{app:incoherentauto}, the asymptotic temperature given by Eq.~\eqref{eq:TVinc} is exactly equal to the temperature obtained in the steady state of an autonomous refrigerator \cite{auto0}, and thus the asymptotic state of the target qubit under repeated incoherent operations is the same as the steady state of the autonomous fridge. More precisely,
\begin{align}
	r_{\text{inc},\infty} &= r_{\text{auto}} & &\text{i.e.} \quad T_{\text{inc},\infty} = T_{\text{auto}}.
\end{align}
This highlights an interesting connection between discrete and continuous cooling procedures; see also \cite{Raam}.

Furthermore, showcasing one of the result of our accompanying article \cite{OurLetter}, the maximal cooling in either case, obtained in the limit $T_H \rightarrow \infty$, is the same as for a single-cycle coherent operation (c.f. Eq.~\eqref{eq:Tmincoh})
\begin{equation}
   T_{\text{auto}}^* = \lim_{T_H \rightarrow \infty} T_{\text{auto}} = \frac{E}{E_B}T_R = T_{\text{coh}}^* \,.
\end{equation}
Note that in this limit we have that $N_{V,\text{inc}} = \frac{1}{2} $. Hence in each repetition the difference between the current and asymptotic ground state population is halved. 

Finally, we discuss the work cost of cooling. Detailed calculations are given in \appref{app:incohop}. Intuitively, the free energy drawn from the hot bath can be divided into two parts: i) the energy required in the first instance of step 1, to initially heat up qubit C to temperature $T_H$, and ii) the energy required in all subsequent repetitions of step 1, to bring qubit C back to $T_H$. This is straightforwardly calculated from the change in population of qubit $C$, which is equal to the change in population of qubit $A$, due to the form of the energy-preserving unitary in step 2. The total heat drawn from the hot bath for $n$ repetitions is
\begin{align}
	Q^H_n &= E_C \left( r_C - r_C^H \right) + E_C \left( r_{\text{inc},n-1} - r \right).
\end{align}

In the asymptotic case, we find that the total heat drawn from the hot bath is exactly the same as if we had run the autonomous refrigerator beginning from the initial state, i.e. $Q^H_{\infty} = Q^H_{\text{auto}}$. See \appref{app:incoherentauto} for a detailed proof.

In order to cool to a given temperature, it is possible to vary the number of repetitions as well as the temperature of the hot bath $T_H$. One may therefore ask which is the most cost-efficient strategy. Generically, we observe (see \figref{fig:compare_incoherent}) that for a given final temperature, implementing many cooling swaps has a lower work cost than using fewer swaps (at higher temperature $T_H$). As implementing a higher number of swaps would take longer time, this observation is reminiscent of the power vs efficiency trade-off in continuously operated machines \cite{ContinuousDevices}.

\begin{figure}[h]
	\centering
	\includegraphics[width=8.5cm]{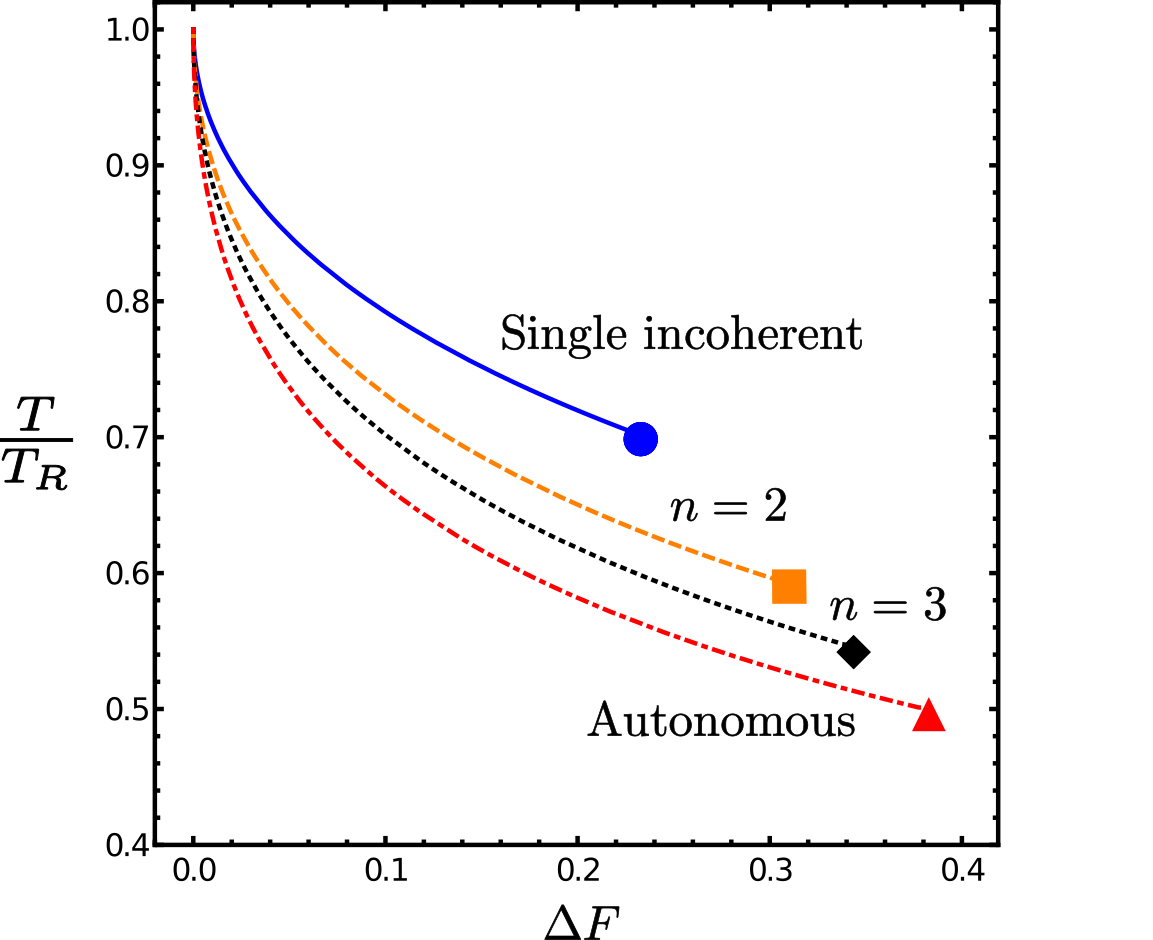}
	\caption{Cooling vs the work cost for different number of repetitions of incoherent operations. Each curve is parametrized by the temperature of the hot bath, $T_H$. \(E_C\), \(E\) and \(T_R\) are all set to \(1\).\label{fig:compare_incoherent}}
\end{figure}

\subsection{Scenario 2: repeated coherent operations}

\begin{figure}[t]
	\centering
	\includegraphics[width=9cm]{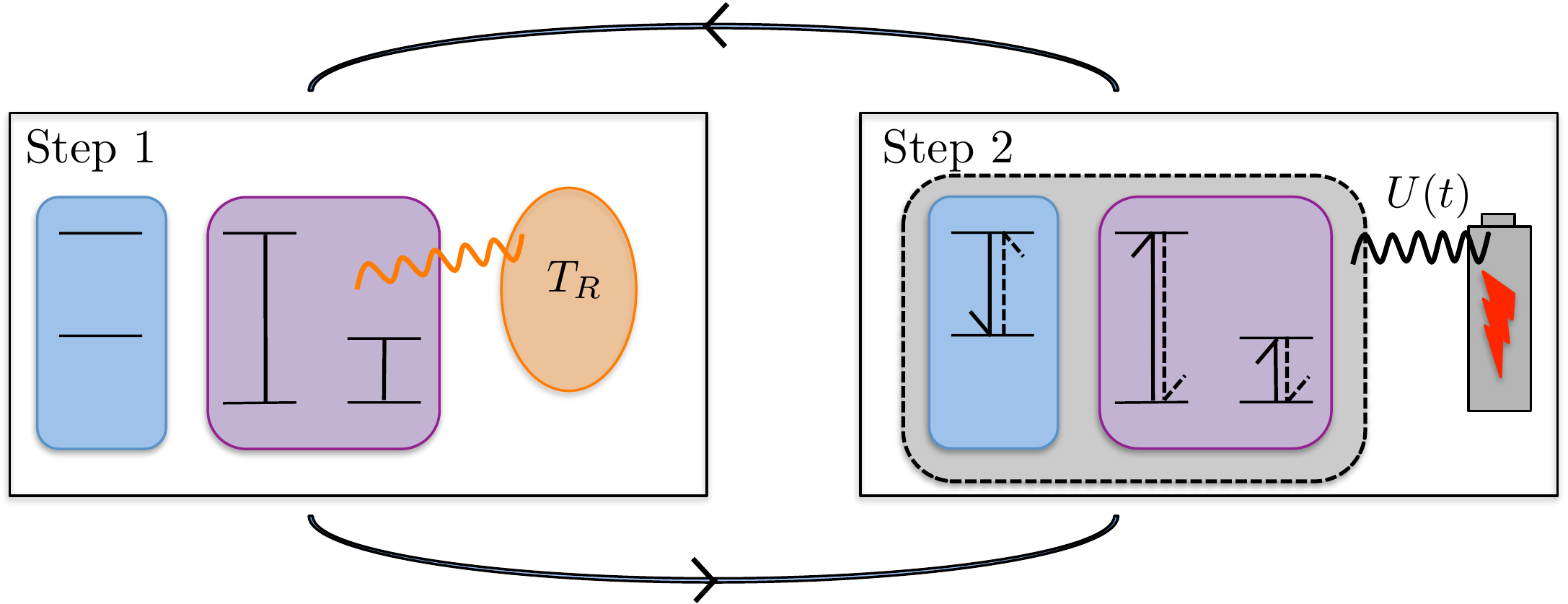}
	\caption{Scenario 2, coherent operations, in the regime of repeated operations. Each cycle comprises the steps of 1. the environment reset of the machine, and 2. cooling.\label{fig:repeat_coh}}
\end{figure}

When discussing single-cycle cooling via coherent operations in \secref{sec:single-cycle2}, we found that \fab{according to the relative size of \(E_C\) and \(E_A\), there were} two different sets of unitaries which lead to the lowest achievable temperature $T_{\text{coh}}^*$ of the target qubit. The first procedure involved only qubits A and B, and maximal cooling could be achieved with a single-qubit machine (i.e. without qubit C). \fab{This procedure was found to be optimal when \(E_C \leq E_A\).} However, \fab{although this procedure was also valid when \(E_C > E\)} we  showed that \fab{in this case} a different procedure, involving all three qubits, could reach the same temperature, but at a lower work cost.

In the present section we discuss cooling via repeated coherent operations. We find that \fab{after the first cycle a procedure similar to the second procedure in the single-cycle case must be applied in order to cool further}. In fact, one can immediately see that for a single-qubit machine, repetitions do not lower the temperature further beyond the single-cycle case. Since the single-qubit machine simply swaps qubits A and B, there is no unitary operation that can cool further, even after $B$ is re-thermalised to the ambient temperature $T_R$.

On the contrary, using a two-qubit machine one can enhance the cooling beyond the single-cycle case. This is achieved by repeating the following steps (see \figref{fig:repeat_coh}, and \appref{app:repeatcoherent} for more details):
\begin{enumerate}
    \item \emph{Environment reset - } Qubits B and C are brought back to the environment temperature $T_R$.
    \item \emph{Cooling step - } The unitary swapping the populations of the states $\{\ket{100},\ket{011}\}$ is applied.
\end{enumerate}

As qubit A is cooled by swapping with the subspace $\{\ket{00}_{BC},\ket{11}_{BC}\}$ of the machine, we identify this subspace as the relevant virtual qubit of the machine, and denote its norm as $N_{V,\text{coh}}$. Following calculations given in \appref{app:repeatcoherent}, one finds that in the asymptotic limit (infinite repetitions), the ground state population of the target goes to
\begin{align}
	r^*_{\text{coh},\infty}= \frac{1}{1+e^{-E/T^*_{\text{coh},\infty}}},
\end{align}
where the asymptotic temperature takes the simple form
\begin{align}
	T^*_{\text{coh},\infty} = T_R \frac{E}{E_B+E_C}. 
\end{align}

This recovers the result of our accompanying article \cite{OurLetter} and the results of heat bath algorithmic cooling with no compression qubit. Note that in the coherent case, the temperature of the virtual qubit is just $T_R$, since both the machine qubits are at $T_R$ after rethermalization. However, due to the swap, the final temperature of the target qubit is not simply the virtual temperature, but rather is modulated by the ratio of energies of the target and virtual qubits, see \appref{app:virtualqubit} for more detail. This is why maximal cooling in the asymptotic case is attained by picking the virtual qubit of the largest energy gap, which for the two qubit machine is $\{\ket{00}_{BC},\ket{11}_{BC}\}$.

For a finite number $n$ of repetitions, the ground state population of the target approaches its asymptotic value as
\begin{align}
\label{eq:coherentnstep}
    r_{\text{coh},n}^* &= r^*_{\text{coh},\infty} - \left( r^*_{\text{coh},\infty} - r \right) \left( 1 - N_{V,\text{coh}} \right)^n.
\end{align}

Thus we see that cooling is enhanced compared to the single-cycle case, i.e. $T^*_{\text{coh},n} < T_{\text{coh}}^*$. (Note that we use $*$ here to denote the lowest achievable temperature for a fixed number of repetitions.)

We proceed to discuss the work cost of this process. Note that the optimal work cost of the first coherent operation has already been discussed in \secref{sec:single-cycle2}, and is denoted by $\Delta F^*_{\text{coh}}$. For further repetitions of the steps presented above, free energy is needed to implement the unitary in step 2, as populations of states with different energies are swapped. (Step 1 is free as it involves thermalisation of the machine qubits to the environment temperature $T_R$). The work cost of $n$ full repetitions of the cycle is given by (details in \appref{app:repeatcoherent})
\begin{align}\label{eq:Fcohn}
	\Delta F_{\text{coh},n}^* &= \Delta F_{\text{coh}}^* + 2 E_C \left( r_{\text{coh},n}^* - r_B \right),
\end{align}
where $\Delta F_{\text{coh}}^*$ is the work cost in the single-cycle regime given by Eq.~\eqref{eq:deltaFmaxcohBig}. In the asymptotic regime, the work cost becomes
\begin{align}
	\Delta F^*_{\text{coh},\infty} &= \Delta F_{\text{coh}}^* + 2E_C \left( r^*_{\text{coh},\infty} - r_B \right),
\end{align}
where $r^*_{\text{coh},\infty}$ is the final ground-state population for the target qubit corresponding to $T^*_{\text{coh},\infty}$. Following the argument expanded in full detail in \appref{app:repeatcoherent}, the steps presented above are the only way to cool the target after the first (optimal) coherent operation, and thus $\Delta F^*_{\text{coh},n}$ represents the minimum work cost given the lowest achievable temperature after $n$ repetitions.

\subsection{Scenario 2: algorithmic cooling}

\begin{figure}[b]
	\centering
	\includegraphics[width=8cm]{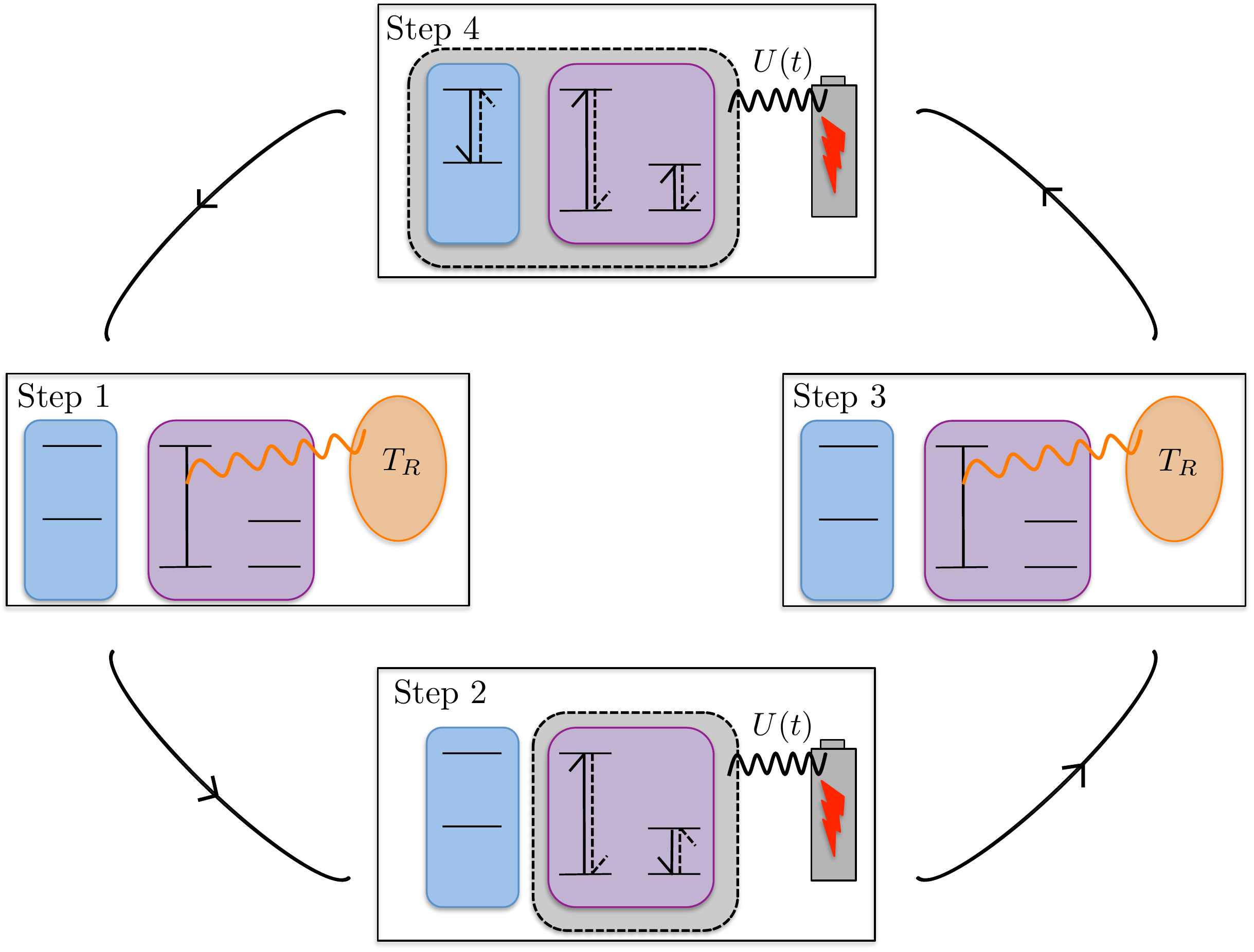}
	\caption{Scenario 2 in the regime of algorithmic cooling. Each cycle comprises the steps of 1. environment reset, 2. precooling, 3. environment reset, and 4. cooling.\label{fig:repeat_algo}}
\end{figure}

It turns out that even stronger cooling can be obtained, by increasing the level of control compared to the above model of repeated coherent operations, specifically, by allowing for individual rethermalisation of each machine qubit separately. This model is equivalent to heat bath algorithmic cooling, this time with a compression qubit, as we will demonstrate shortly. The procedure consists in repeating the following steps, shown schematically in \figref{fig:repeat_algo}:
\begin{enumerate}
    \item \emph{Environment reset - } Qubit $B$ is brought back to the environment temperature $T_R$.
    \item \emph{Precooling - } The states of qubits $B$ and $C$ are swapped.
    \item \emph{Environment reset - } Qubit $B$ is brought back to the environment temperature $T_R$.
    \item \emph{Cooling step - } The unitary swapping the populations of the states $\ket{100}\leftrightarrow\ket{011}$ is applied.
\end{enumerate}
As before, the target qubit is swapped with the qubit subspace of the machine that has the highest energy gap, spanned by $\ket{00}_{BC}$ and $\ket{11}_{BC}$. However, thanks to the precooling step, the virtual temperature of this coldest qubit subspace is decreased, from $T_R$ to
\begin{align}
	T_{V,\text{algo}} &= T_R \frac{E_B + E_C}{2E_B}.
\end{align}

The final temperature is again determined by the virtual temperature. Following calculations given in \appref{app:algo}, in the asymptotic limit of infinite repetitions, the ground state population of the target qubit tends to
\begin{align}
	r_{\text{algo},\infty}^* &= \frac{1}{1+e^{-E/T_{\text{algo},\infty}^*}},
\end{align}
where the aysmptotic temperature is given by
\begin{align}\label{eq:Talgo}
T_{\text{algo},\infty}^* = T_R \frac{E}{2E_B} = \frac{T_{\text{coh}}^*}{2}.
\end{align}

The final temperature is thus half the temperature achieved via single-cycle coherent operations. Note that it is also half of the minimal achievable temperature $T_{\text{auto}}^*$ in the asymptotic incoherent regime. Moreover, since $E_B >E_C$, we see that the lowest achievable temperature of algorithmic cooling is strictly colder than that of repeated coherent operations. It is worth noting that the expression for the minimal temperature of Eq.~\eqref{eq:Talgo} perfectly matches known results in algorithmic cooling: specifically Eq. (7) of Ref. \cite{raeisi} (for the case of two reset qubits), as well as Eq. (10) of Ref. \cite{algo3}.

For a finite number of repetitions of the above cycle of steps, one finds that the ground state population of the target approaches $r_{\text{algo},\infty}^*$ as
\begin{align}\label{eq:ralgon}
r_{\text{algo},n} & = r_{\text{algo},\infty}^* - \left( r_{\text{algo},\infty}^* - r_0 \right) \left( 1 - N_{V,\text{algo}} \right)^n,
\end{align}
where $r_0$ is the population of the ground state before the first application of the procedure, and $N_{V,\text{algo}}$ is the norm of the virtual qubit $\{\ket{00}_{BC},\ket{11}_{BC}\}$ right before step 4 (i.e. after qubit $C$ has been pre-cooled and qubit $B$ rethermalized).

Finally, we discuss the work cost of this process. Free energy is needed to implement the unitaries in steps 2 and 4, as populations of states with different energies are swapped. Steps 1 and 3 have zero cost, since they only involve the environment bath. As detailed in \appref{app:algo}, the work cost after $n$ full repetitions is given by
\begin{equation}\label{eq:Falgon}
\begin{aligned}
	\Delta F_{\text{algo},n} = &E ( r_B - r_C ) + 2E_C ( r_{\text{algo},n} - r_0 ) \\
	&+ E ( r_{\text{algo},n-1} - r_0 )
\end{aligned}
\end{equation}

Let us first remark that for cooling to a temperature that would be achievable with repeated coherent operations, algorithmic cooling has a higher work cost, as is argued in \appref{app:optimalcoherent}, and on comparison of Eqs.~\eqref{eq:Fcohn} and \eqref{eq:Falgon}. Thus, in order to minimize the work cost, a better strategy consists in first cooling using repeated coherent operations, until the temperature cannot be lowered any further, and only then switch to algorithmic cooling. A detailed discussion of this sequence of operations may be found in \appref{app:optimalcoherent}. In the asymptotic case of infinite repetitions, the work cost of this procedure (denoted by $\Delta F^*$) becomes
\begin{equation}
\begin{split}
	\Delta F^*_{\text{algo},\infty} &= \Delta F^*_{\text{coh},\infty} + E \left( r_B - r_C \right) \\
						&\quad + (2 E_C+E)  \left( r^*_{\text{algo},\infty} - r_{\text{coh},\infty}^* \right).
\end{split}
\end{equation}
This procedure turns out to be optimal with respect to the work cost if one is interested in reaching the lowest achievable temperature \(T_{\text{algo},\infty}^*\). If one is however interested in cooling the target to a temperature between \(T_{\text{algo},\infty}^*\) and \(T_{\text{coh},\infty}^*\), fully precooling qubit \(C\) is unnecessary and there exists a better manner of proceeding after repeated coherent operations, where given the desired final temperature of the target, one tunes the precooling of qubit $C$ to be a partial rather than a full swap.

In \figref{fig:plot_compare_coherent}, we compare the work cost of the optimal sequence of operations (first repeated coherent, then optimized algorithmic cooling) against that of using standard algorithmic cooling from the beginning.

\begin{figure}[t]
	\centering
	\includegraphics[width=9cm]{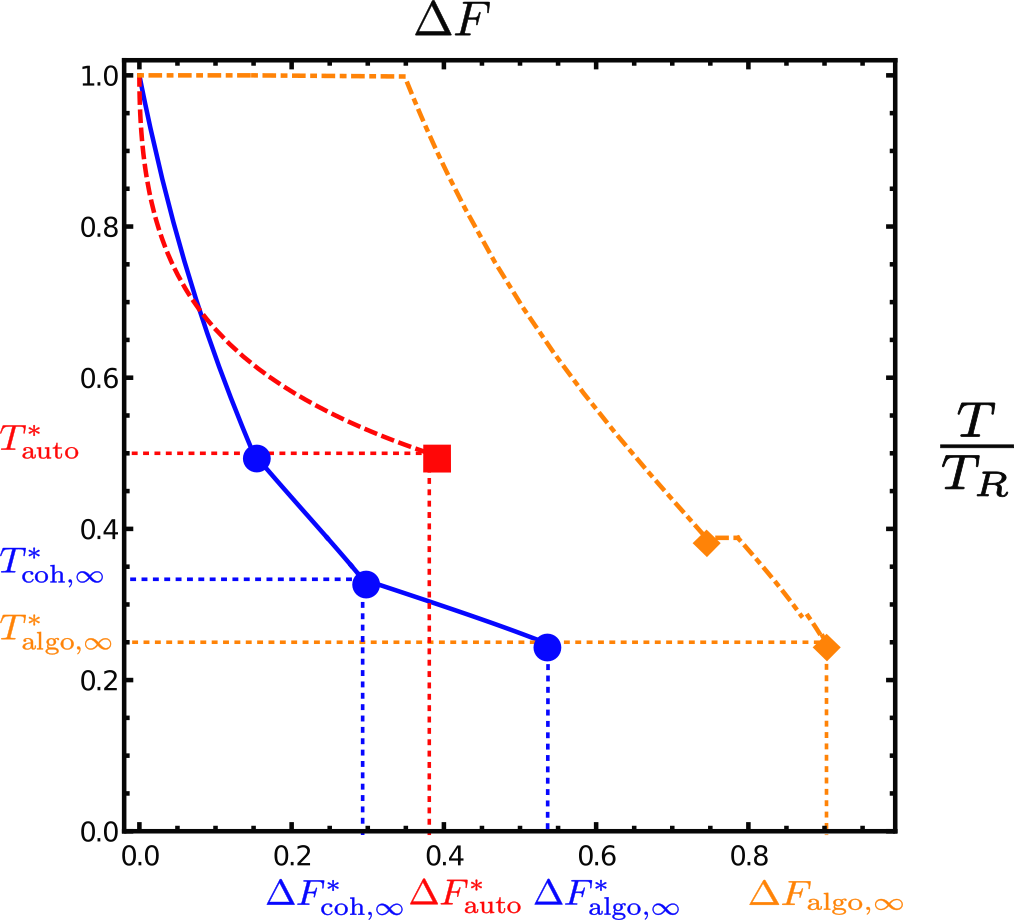}
	\caption{Comparison of the work cost of using algorithmic cooling from the beginning (orange dot-dashed), as opposed to the optimal sequence of coherent operations (blue solid line), and of an autonomous refrigerator (red dashed, parametrized w.r.t. $T_H$). \(E_C\), \(E\) and \(T_R\) are all set to \(1\).\label{fig:plot_compare_coherent}}
\end{figure}

Finally, it is worth noting that algorithmic cooling is rather expensive even when compared to autonomous cooling, for the same target temperature, see \figref{fig:plot_compare_coherent}. Thus, while algorithmic cooling can achieve the lowest temperatures, it may be the case, depending on the parameters of the problem, that an autonomous refrigerator can cool to any $T \geq T^*_{\text{auto}}$ more efficiently.

\jona{
\section{Saturating the second law \label{sec:secondlaw}}

Upon comparing the cooling performance of the minimal machines presented in this article with the ultimate performance bound set by each paradigm in our accompanying article \cite{OurLetter}, it is quite striking to notice that as simple as they are, the minimal machines already suffice to saturate the bound. The next natural question to ask is if these machines are also optimal in terms of the associated work cost. We in this section answer this question by the negative.

Clearly, fundamental limitations on the work cost arise from the second law. Specifically, the free energy change of the target qubit is a lower bound on the work cost. Here we present a family of $N$-qubit coherent machines which asymptotically saturate this bound. These machines have been introduced in Ref. \cite{paul} for demonstrating optimal work extraction from quantum states. Moreover, for any machine in the family, we construct an incoherent machine of $2N$ qubits achieving the same temperature. In the limit where the hot bath becomes infinite, the associated work cost is the same up to a constant offset \gh{that can be made arbitrarily small}.

As we have learned from section \ref{sec:onequbit}, a given temperature $T$ can be achieved via a single-qubit machine with energy gap $E_N = E \frac{T_R}{T}$. In order to minimize the work cost, we now introduce $N-1$ additional qubits with energy gaps (evenly) spaced between $E$ and $E_N$. The single swap is now replaced by a sequence of swaps between the target qubit and machine qubits in order of increasing energy gaps. This can be understood intuitively by noticing that the energy difference $\Delta E$ when swapping two qubits represents the work cost per unit of population transferred $\Delta r$ (see \appref{app:virtualqubit})
\begin{align}
	\Delta F &=  \Delta r \Delta E.
\end{align}
Hence, for a given final population transfer, replacing one single swap at large $\Delta E$ by $N$ swaps at smaller $\Delta E $ reduces the work cost. As shown in Ref. \cite{paul}, the total work cost of this procedure is given by
\begin{align}
	\Delta F &= \Delta F_{\text{target}} + O\left( \frac{1}{N} \right),
\end{align}
where $\Delta F_{\text{target}}$ is the increase in the free energy of the target qubit. In the limit $N \rightarrow \infty$, the work cost is exactly the free energy transferred to the system, which is the lower bound provided by the second law. See \appref{app:secondlaw} for details.

The next question is whether we can find an incoherent machine which also saturates the second law. 
A first possibility is to transform the above coherent machine into an incoherent one, using the same idea as discussed in our accompanying article \cite{OurLetter}. Specifically, each swap can be made energy preserving by adding an extra qubit to the machine. Therefore, the $N$-qubit coherent machine discussed above, can be made incoherent by adding $N$ extra qubits. The temperature achieved by the incoherent machine will match that of the coherent if either $T_H \rightarrow \infty$, or if the energies of the machine qubits are increased so as to match the desired temperature on the target qubit. In the former case, the work cost will diverge when $N \rightarrow \infty$, as each additional qubit must now be heated from $T_R$ to infinite temperature. Nevertheless, in the second case, this problem can be circumvented by noting that these $N$ additional qubits do not need to correspond to physical qubits, but can be taken as virtual qubits. For instance, one can consider a single evenly spaced $(N+1)$-level system providing all these virtual qubits. By embedding this $(N+1)$-level system into a larger system, the initial work cost can be made arbitrarily small, and we can thus approach the work cost of the corresponding coherent machine arbitrarily closely. Consequently, we have constructed an incoherent machine which also saturates the second law in the limit of $N \rightarrow \infty$. See \appref{app:incsecondlaw} for an explicit construction and proof.

}

\gh{
\section{Conclusion and outlook}
\label{sec:conclusion}
}
We have presented a unified view of quantum refrigeration, allowing us to compare various paradigms. In particular, our framework incorporates autonomous quantum thermal machines, algorithmic cooling, single-cycle refrigeration and the resource theory of thermodynamics.

\begin{boxfigure}{Lowest achievable temperature $T^*$ and associated free energy change of the resource $\Delta F^*$ for different cooling paradigms and machine sizes}{Tempi}
	\tcblower
	\center{\bf{One-qubit machines}}
	\begin{align}
T_{coh}^* = \frac{E}{E_B} T_R&&\Delta F_{coh}^* = (r_{coh} - r)(E_B - E)\nonumber
	 	\end{align} 
	\center{\bf{Two-qubit machines}}
	\begin{align}
	&T^*_{\text{inc}}=\frac{E}{\ln \left( \frac{1-r^*_{\text{inc}}}{r^*_{\text{inc}}}\right)} &&r^*_{\text{inc}}=\frac{1}{2} (r+r_B)	&\Delta F^*_{\text{inc}}= E_C (r_C- \frac{1}{2}) > \Delta F_{\text{coh}}^* \nonumber\\
	&T^*_{\text{auto}}=\frac{E}{E_B}T_R&& r^*_{\text{auto}}=r_B&\Delta F^*_{\text{auto}}= E_C (r_C- \frac{1}{2}+r_B-r)\nonumber\\
	&T^*_{\text{coh}}=T^*_{\text{auto}} && r_{\text{coh}}^*=r^*_{\text{auto}}&\Delta F^*_{\text{coh}}= \begin{cases} E_C (r_B- r),& E_C \leq E\nonumber\\
	E_C (r_B-r)-E(r_C-r),& E_C > E \end{cases}\nonumber\\
	&T^*_{\text{coh,}\infty}=\frac{E}{E_B+E_C} T_R && r^*_{\text{coh},\infty}= \frac{1}{1+e^{-\frac{E_B+E_C}{T_R}}}&\Delta F^*_{\text{coh,}\infty}=\Delta F^*_{\text{coh}}+2 E_C (r^*_{\text{coh},\infty}- r_B)\nonumber\\
	&T^*_{\text{algo,}\infty}=\frac{E}{2E_B} T_R && r^*_{\text{algo},\infty}= \frac{1}{1+e^{-\frac{2E_B}{T_R}}}&\Delta F^*_{\text{algo,}\infty}= \Delta F^*_{\text{coh,}\infty}+E(r_B-r_C)+\nonumber\\
	&\,&&\,&(2 E_C+E) (r^*_{\text{algo},\infty}-r^*_{\text{coh},\infty})\nonumber
	 	\end{align} 
		\center{\bf{Asymptotic machine size}}
		\begin{align}
		T^* =  \frac{E}{E_{max}}T_R&&\Delta F^* = \Delta F_{\text{target}} \nonumber
		\end{align}
		$T_R:$ Room temperature; $E:$ target qubit energy gap; $E_B:$ energy gap of first machine qubit; $E_C:$ energy gap of second machine qubit; $E_{max}:$ maximum energy gap of the machine; $r:$ initial target qubit ground state population; $r^*:$ final target qubit ground state population;
		
		The subscripts denote the different operational paradigms, i.e. \emph{coh} refers to the coherent scenario, \emph{inc} to the incoherent scenario, \emph{auto} to the autonomous scenario and \emph{algo} to algorithmic cooling.
\end{boxfigure}
 We characterize fundamental limits of cooling, in terms of achievable temperature and work cost, for both coherent and incoherent operations, in single-cycle, finite repetitions, and asymptotic regimes. The main formulas are summarized in the boxes shown.

\jona{We find that, contrary to classical thermodynamics, the fundamental limits crucially depend on the level of control available. In particular, this implies that the free energy does not uniquely determine the minimal achievable temperature. Moreover, the size of the machine represents an additional form of control, which also influences thermodynamic performance. On the one hand, for minimal machines, the difference between coherent and incoherent control is strongly pronounced. On the other, in the asymptotic limit, the two scenarios become mostly equivalent.}

\nico{While we focused here on the task of cooling a single qubit, it is natural to ask what the fundamental limits to cooling larger systems are. Understanding the qubit case already provides significant insight into the general case. For the task of increasing the ground-state population, we showed that qubit bounds apply in general. 
Repeating every scenario, for every possible notion of cooling would, while possible, not add much insight without a more physical motivation of the respective target Hamiltonian and setting. It would furthermore be interesting to characterize the work cost of cooling general systems, although this will also depend on the exact Hamiltonian structure of the target, and so one cannot expect to obtain a single-letter formula (as in the qubit case).
}


Finally, it would indeed be interesting to discuss different tasks than cooling, e.g. work extraction, and determine if similar conclusions can be drawn.

\textbf{Acknowledgements.}We are grateful to Flavien Hirsch, Patrick P. Hofer, Marc-Olivier Renou, and Tam\'as Krivachy for fruitful discussions. We would also like to acknowledge the referees of our initial submission for productive and challenging comments. Several authors acknowledge the support from the Swiss NFS and the NCCR QSIT: R.S. through the grants Nos. $200021\_169002$ and $200020\_165843$, N.B. through the Starting grant DIAQ and grant $200021\_169002$, F.C. through the AMBIZIONE grant PZ00P2$\_$161351, and G.H. through the PRIMA grant $PR00P2\_179748$ and Marie-Heim V\"ogtlin grant 164466. MH acknowledges support from the Austrian Science Fund (FWF) through the START project Y879-N27. JBB acknowledges support from the Independent Research Fund Denmark.

\appendix

\onecolumngrid
\section{Degeneracy condition for cooling\label{app:degeneracy}}

\fab{We want here to investigate the conditions for cooling to be possible in the incoherent scenario. First we will see that given an arbitrary machine, the system hamiltonian of target and machine together must have some degeneracy in that scenario for cooling to be possible at all with that machine. This is the content of Lemma \ref{lemma:degexist}. We will then move on to the specific case of the two-qubit machine and prove that given such a machine, cooling is only possible when \(E_B=E_A+E_C\). This is what Lemma \ref{lemma:twoqubitdeg} proves.}

\begin{lemma}
\label{lemma:degexist}
As \([U,H]=0\), \(U\) can only cool the target by acting on the degenerate eigenspaces of \(H\).
\end{lemma}

\begin{proof}
\fab{Let \(\text{Eig}_{H}(E)\) be the eigenspace of \(H\) with eigenvalue \(E\). Let \(\ket{v} \in \text{Eig}_{H}(E)\). Per definition \(H \ket{v} = E \ket{v}\). Furthermore as \([U,H]=0\) we have

\begin{equation}
\label{equ:UonE}
\begin{aligned}
&U H \ket{v} = H U \ket{v}\\
\Leftrightarrow & E (U \ket{v})= H (U \ket{v}),
\end{aligned}
\end{equation}
which shows that \(U \ket{v} \in \text{Eig}_{H}(E)\). This means that every energy eigenspace is invariant under \(U\) and so as the whole vector space can be decomposed as a direct sum of \(\text{Eig}_H(E), U= \oplus_E U_E\). We now have left to prove that if \(\text{dim}(\text{Eig}_H(E))=1\), \(U_E\) does not affect the temperature of the target. for this, let \(E\) be an eigenvalue of \(H\) with \(\text{dim}(\text{Eig}_H(E))=1\). Let \(\ket{v} \in \text{Eig}_H(E)\). From Equation \eqref{equ:UonE} and \(\text{dim}(\text{Eig}_H(E))=1\), \(U_E \ket{v}=U \ket{v}=\lambda \ket{v}\), meaning that \(\ket{v}\) is an eigenvector of \(U_E\) with eigenvalue \(\lambda\). Since \(U_E\) is unitary, \(\lambda = e^{i \theta}\) and so \( U_E \ket{v} \bra{v} U_E^{\dagger}= \ket{v}\bra{v}\), which proves our assertion as only the diagonal elements of the density matrix \(U_E \rho U_E^{\dagger}\) contribute to the temperature of the target and that we can write any \(\rho\) as
\begin{equation}
\rho= \sum_{ij} a_{ij} \ket{v_i} \bra{v_j},
\end{equation}
with \((\ket{v_i})_i\) being an ONB of eigenvectors of H and \(\ket{v_1}=\ket{v}\). So
\begin{equation}
\begin{aligned}
U_E \rho U_E^{\dagger} &= U_E (\sum_{ij} a_{ij} \ket{v_i} \bra{v_j}) U_E^{\dagger}\\
&= U_E (\sum_{i,j \neq 1} a_{ij} \ket{v_i} \bra{v_j}+\sum_{j \neq 1} a_{1j} \ket{v} \bra{v_j}+\sum_{i \neq 1} a_{i1} \ket{v_i} \bra{v}+
a_{11} \ket{v} \bra{v}) U_E^{\dagger}\\
&= \sum_{i,j \neq 1} a_{ij} \ket{v_i} \bra{v_j}+\sum_{j \neq 1} \lambda a_{1j} \ket{v} \bra{v_j}+\sum_{i \neq 1} a_{i1} \bar{\lambda} \ket{v_i} \bra{v}+
a_{11} \ket{v} \bra{v}.
\end{aligned}
\end{equation}
implying that
\begin{equation}
(U_E \rho U_E^{\dagger})_{kk}=a_{kk} \ket{v_k} \bra{v_k}=\rho_{kk},
\end{equation}
, i.e. the diagonal elements of \(U_E \rho U_E^{\dagger}\) are the original ones.}
\end{proof}

Next we want to argue that 

\begin{lemma}
\label{lemma:twoqubitdeg}
Among all the possible degeneracies of \(H\), only \(E_B=E_A+E_C\) enables cooling of qubit A.
\end{lemma}

\begin{proof}
 Going through all the possible eigenvalue degeneracies of \(H= H_A+H_B+H_C, \; H_i=E_i \ket{1} \bra{1}_i \otimes \mathds{1}_{\bar{i}}, i \in \{A,B,C\}\), we see that we can have 3 different types of degeneracies:

  \begin{enumerate}
  \item \(E_i = E_j , \quad i,j \in \{A,B,C\}\)
  \item\( E_i =0, \quad i \in \{A,B,C\}\)
  \item \(E_i = E_j + E_k, \quad i,j,k \in \{A,B,C\}\)
  \end{enumerate}
We first look at type 2. Imposing \(E_A=0\) we get 4 degenerate subspaces: \(A_{mn}:=\text{span}\{\ket{0}_A\ket{mn}_{BC}, \ket{1}_A\ket{mn}_{BC}\}\), where \(m,n \in \{0,1\}\). Our unitary can act within each \(A_{mn}\) subspace on \(\rho^H = \tau \otimes \tau_B \otimes \tau_C^H\). However note that as \(r_A= \frac{1}{1+e^{- \frac{E_A}{T_R}}}=\frac{1}{2} = 1-r_A\), we have 
\begin{equation}
\begin{aligned}
\rho^H_{0mn,0mn}&=r_A (m+ (-1)^m r_B) (n+ (-1)^n r_C^H)\\
&=(1-r_A) (m+ (-1)^m r_B) (n+ (-1)^n r_C^H)=\rho^H_{1mn,1mn}
\end{aligned}
\end{equation}
such that in each of the \(A_{mn}\) \(\rho\) is proportional to the identity and hence \(U \rho U^{\dagger}=\rho\) for all unitaries \(U\) acting only within those subspaces. Note that this argument actually holds for any permutation of A, B and C thus also treating the \(E_B=0\) and \(E_C=0\) cases and showing that Type 2 degeneracies do not enable cooling qubit A.\\

Turning to type 1, if \(E_B=E_C\), we have two 2-dim. degenerate subspaces \(\text{span}(\ket{001},\ket{010})\) and \(\text{span}(\ket{101},\ket{110})\). In order to cool qubit A, one should maximize \([\text{Tr}_{BC}(U \rho^H U^{\dagger})]_{1,1}= (U\rho^H U^{\dagger})_{000,000}+(U\rho^H U^{\dagger})_{001,001}+(U\rho^H U^{\dagger})_{010,010}+(U\rho^H U^{\dagger})_{011,011}\). As unitaries are trace preserving, acting with U within the first subspace leaves \(\rho^H_{001,001}+\rho^H_{010,010}\) unchanged. Acting with U within the second one does not alter any term in \([\text{Tr}_{BC}( \rho^H)]_{1,1}\), meaning that this degeneracy does not allow us to cool qubit A. For \(E_A=E_B\), the degenerate subspaces are \(\text{span}(\ket{010}, \ket{100})\) and \(\text{span}(\ket{011}, \ket{101})\). Doing the same analysis as before shows that in general the unitary doesn't leave \([\text{Tr}_{BC}( \rho )]_{1,1}\) invariant, unfortunately it does for our \(\rho^H\) since with this condition 
\begin{equation}
\rho^H_{010,010} = r_A (1-r_B) r^H_C= r_B(1-r_A)  r^H_C=\rho^H_{100,100}
\end{equation}
 and  similarly \(\rho_{011,011}=\rho_{101,101}\). Imposing \(E_A = E_C\) we have as degenerate subspaces \(\text{span}(\ket{001},\ket{100})\) and  \(\text{span}(\ket{011},\ket{110})\). As above the unitary does not in general leave \([\text{Tr}_{BC}(\rho )]_{1,1}\) invariant. For our \(\rho^H\) it is also the case but since \(T_H \geq T_R\), we have \(-\frac{E_C}{T_H}\geq -\frac{E_A}{T_R}\) meaning \(r_C^H \leq r_A\), such that 
 \begin{equation}
 \rho^H_{001,001}=r_A r_B (1-r_C^H) \geq r_C^H r_B (1-r_A) = \rho_{100,100}^H
 \end{equation}
  and 
  \begin{equation}
  \rho^H_{011,011}= r_A (1-r_B) (1-r_C^H)\geq r_C^H (1-r_B) (1-r_A)\geq \rho_{110,110}^H.
  \end{equation}
  The unitary that maximizes \([\text{Tr}_{BC}(U \rho^H U^{\dagger})]_{1,1}\) is therefore the trivial one. Indeed any 2 dimensional unitary can be written as
 
 \begin{equation}
U= \begin{pmatrix}
 a & b \\
 -b^* e^{i \theta} & a^* e^{i \theta}
 \end{pmatrix},
 \end{equation}
 with \(\lvert a \rvert ^2 + \lvert b \rvert ^2 =1\) and \(\theta \in [0,2\pi)\). And so
 
 \begin{equation}
 \begin{split}
\left[ U \begin{pmatrix} \rho^H_{001,001} &0 \\ 0 & \rho^H_{100,100} \end{pmatrix} U^{\dagger} \right]_{1,1} &= \begin{pmatrix} \lvert a \rvert^2 \rho^H_{001,001}  + \lvert b \rvert ^2 \rho^H_{100,100} &  a b e^{- i \theta} (\rho^H_{100,100}- \rho^H_{001,001} ) \\ a^* b^* e^{i \theta} (\rho^H_{100,100} - \rho^H_{001,001} ) & \lvert b \rvert ^2 \rho^H_{001,001}  + \lvert a \rvert ^2 \rho^H_{100,100} \end{pmatrix}_{1,1}\\
&=\lvert a \rvert^2 \rho^H_{001,001}  + \lvert b \rvert ^2 \rho^H_{100,100}
\end{split}
 \end{equation}
 is maximal for \(\lvert a \rvert =1\), \(b=0\) and any choice of \(\theta\), which exactly corresponds to the unitary of \(\text{span}(\ket{001}, \ket{100})\) acting trivially on our \(\rho^H\). The same obviously holds for the unitaries of \(\text{span}(\ket{011}, \ket{110})\).
  This type of degeneracy hence also does not allow any cooling to happen.\\
 
 We are left with the last type of degeneracy, type 3. Looking at \(E_A = E_B + E_C\) we have that the degenerate subspace is \(\text{span} (\ket{011}, \ket{100})\). As after heating we have that \(T_R\leq T_H\), which implies \(e^{-E_C/T_R} \leq e^{-E_C/T_H}\), we have, 
 \begin{equation}
 \begin{aligned}
 \rho^H_{011,011}&= r_A (1-r_B) (1-r_C^H)\\
 &= r_A e^{-\frac{E_B}{T_R}} e^{-\frac{E_C}{T_H}} r_B r_C^H\\
 &\geq r_A e^{-\frac{E_B+E_C}{T_R}} r_B r_C^H= (1-r_A) r_B r_C^H= \rho^H_{100,100}
 \end{aligned}
 \end{equation}
  meaning that our unitary can only decrease  \([\text{Tr}_{BC}(U \rho U^{\dagger})]_{1,1}\) by making use of this degeneracy (that corresponds to heating qubit A). Similarly for \(E_C = E_A + E_B\) (here the subspace is \(\text{span}(\ket{001}, \ket{110})\) and we have \(\rho^H_{001,001} \geq \rho^H_{110,110}\)).\\
 
 However, for \(E_B=E_A+E_C\), we have that our unitary can increase \([\text{Tr}_{BC}( \rho^H )]_{1,1}\) by acting in the degenerate subspace \(\text{span}(\ket{010}, \ket{101})\) since 
 \begin{equation}
 \begin{aligned}
 \rho^H_{010,010}&= r_A (1-r_B) r_C^H\\
 &= r_A e^{-\frac{E_A}{T_R}} e^{-\frac{E_C}{T_R}} r_B r_C^H\\
 &\leq r_A e^{-\frac{E_A}{T_R}} e^{-\frac{E_C}{T_H}} r_B r_C^H= (1-r_A) r_B r_C^H= \rho^H_{101,101}.
 \end{aligned}
 \end{equation}
  This shows that the only single degeneracy enabling cooling is \(E_B=E_A+E_C\). 
  
  To finish the proof one needs to prove that there is no way of selecting some of the above degeneracies and achieve cooling without also having to select the degeneracy \(E_B=E_A+E_C\). All the ways of selecting two degeneracies can be listed as
  
  \begin{enumerate}
  \item[a)] \(E_i=E_j=0, \quad i,j \in \{A,B,C\}\)
  \item[b)] \(E_i=E_j, E_k=0, \quad \{i,j,k\}=\{A,B,C\}\)
  \item[c)] \(E_A = E_B=E_C=0\)
  \item[d)] \(E_A=E_B=E_C\)
  \item[e)] \(E_i=E_j, E_k=2 E_i, \quad \{i,j,k\}=\{A,B,C\}\)
  \end{enumerate}
  
In case a), \(\rho \propto \mathds{1}\) within the degenerate subspaces and so no cooling can occur. In case b) the degenerate subspaces are \(\text{span}(\ket{00}_{ij} \ket{0}_k, \ket{00}_{ij} \ket{1}_k)\), \(\text{span}(\ket{11}_{ij} \ket{0}_k, \ket{11}_{ij} \ket{1}_k)\), and  \(\text{span}(\ket{01}_{ij} \ket{0}_k, \ket{01}_{ij} \ket{1}_k, \ket{10}_{ij} \ket{0}_k,\ket{10}_{ij} \ket{1}_k)\). In the first two subspaces \(\rho \propto \mathds{1} \). In the third if \((i,j,k)=(A,B,C)\) then \(\rho \propto \mathds{1}\), if \((i,j,k) = (B,C,A)\) then cooling is possible as \(\rho^H_{001,001}=\rho^H_{101,101} \geq \rho^H_{010,010}=\rho^H_{110,110}\), but this is no contradiction to our claim since in this case \(E_B= E_A+E_C\) holds, and if \((i,j,k)= (C,A,B)\), \(\rho^H_{100,100}= \rho^H_{110,110} \leq \rho^H_{001,001}= \rho^H_{011,011}\), meaning that no cooling is possible. In case c) \(\rho \propto \mathds{1}\) and so no cooling is possible. In case d) the degenerate subspaces are \(\text{span}(\ket{001},\ket{010},\ket{100})\) and \(\text{span}(\ket{011}, \ket{101}, \ket{110})\) and as \(\rho^H_{001,001} \geq \rho^H_{010,010}= \rho^H_{100,100}\) and \(\rho^H_{011,011} = \rho^H_{101,101} \geq \rho^H_{110,110}\), no cooling is possible. Finally in case e) the degenerate subspaces are \(\text{span}(\ket{01}_{ij} \ket{0}_k, \ket{10}_{ij} \ket{0}_k)\), \(\text{span}(\ket{00}_{ij} \ket{1}_k, \ket{11}_{ij} \ket{0}_k)\), and \(\text{span}(\ket{01}_{ij} \ket{1}_k, \ket{10}_{ij} \ket{1}_k)\). If \((i,j,k) = (A,B,C)\) then \(\rho^H_{010,010}=\rho^H_{100,100}\), \(\rho^H_{011,011}=\rho^H_{100,100}\), and \(\rho^H_{001,001}\geq \rho^H_{110,110}\) so that no cooling is possible. If \((i,j,k) = (B,C,A)\) then \(\rho^H_{001,001}\geq \rho^H_{010,010}\), \(\rho^H_{101,101}\geq \rho^H_{110,110}\), and \(\rho^H_{100,100}\leq \rho^H_{011,011}\) so that no cooling is possible either.  If \((i,j,k) = (C,A,B)\), \(\rho^H_{010,010}\leq \rho^H_{101,101}\) so that one can cool in that subspace but as in this case \(E_B=E_A+E_C\) also happens to hold; this again is no contradiction to our claim.

If one selects more than two different degeneracies from the list 1.,2., and 3., either three linearly independant degeneracies are selected, which results in \(E_A=E_B=E_C=0\) and leads to no cooling as shown above, or less than three of the selected degeneracies are linearly independant and the situation reduces to one of the above treated case. This ends the proof.
\end{proof}

\section{Optimal incoherent thermalisation} \label{app:thermalinc}

Here we want to argue that for the case of the two-qubit machine in order to cool the target qubit maximally, the best way to make use of both thermal baths at \(T_R\) and \(T_H\) respectively, is to thermalise qubit B at \(T_R\) and qubit C at \(T_H\).

To begin with, note that the only allowed unitaries are those within the degenerate subspace, as these are the only ones that preserve energy, see \secref{app:degeneracy}. Any unitary within this qubit subspace can be viewed as a partial swap between the populations of the two levels (up to a change in relative phase, which does not affect cooling). Thus the maximum cooling is achieved by either swapping the populations fully, or not at all, since these are the two extremes of the achievable populations.

Thus given two thermal baths, at temperatures $T_R$ and $T_H$, the optimal manner of cooling would be to thermalize qubits $B$ and $C$ in such a way as to maximize the difference in the populations of the two degenerate levels \(\ket{101}\) and \(\ket{010}\) before applying a full swap; i.e. maximize $p_{101} - p_{010}$, where \(p_{ijk}\) denotes the population of level \(\ket{ijk}\).

Consider that we thermalize $B$ and $C$ to temperatures between those of the environment and of the hot bath (these are the extremes of temperatures available to us, and thus any temperature in between is also attainable). It is straightforward to check that
\begin{align}
\frac{d}{d T_B} \Big( p_{101} - p_{010} \Big) &= - \frac{ E_B r_B (1-r_B)}{T_B^2} \left( r r_C + (1-r) (1-r_C) \right) < 0 \quad \forall T_B, \\
\text{and} \quad \frac{d}{d T_C} \Big( p_{101} - p_{010} \Big) &= + \frac{E_C r_C (1-r_C)}{T_C^2} \left( r (1-r_B) + (1-r) r_B \right) > 0 \quad \forall T_C.
\end{align}

Therefore, it is optimal to have qubit $B$ be as cold as possible (the environment temperature $T_R$), and qubit $C$ be as hot as possible (the temperature of the hot bath $T_H$). Thus, although the whole machine has access to a hotter thermal bath at temperature \(T_H\), it is best to only put qubit \(C\) in contact with it, leaving \(B\) at the room temperature \(T_R\).

Note that the above argument also holds if the population on qubit A is set to be some other value than \(r\), meaning that in the repeated incoherent operations one should also rethermalise qubit B to \(T_R\) and qubit C to \(T_H\) before applying the swap operation in order to maximally cool the target qubit.

\section{Single-cycle coherent machines}\label{app:Uopt}

We want here to discuss the solution of the single-cycle coherent machines presented in the main text. More precisely, we are interested in finding the unitary \(U_{\text{opt}}\) (or equivalently the state \(\rho_{\text{opt}}= U_{\text{opt}} \rho^{\text{in}} U_{\text{opt}}^{\dagger}\)) that enables us to cool the target to a given temperature, i.e. ground state, \(r_{\text{coh}} \in [r^{\text{in}},r_{\text{coh}}^*]\) at a minimal work cost. From the discussion of the main text we know that using the Schur-Horn theorem, finding \(\rho_{\text{opt}}\) for a system comprised of a target qubit and a machine of \(n/2\) energy gaps amounts to solving
\begin{equation}
\label{equ:genqubitprob}
\min_{\vv{\rho} \prec \vv{\rho^{\text{in}}}} \vv{\rho} \cdot \vv{H}, \quad \text{s.t. } \sum_{i=1}^{n/2} \rho_i=r_{\text{coh}}.
\end{equation}
Indeed, the solution of Equation \eqref{equ:genqubitprob} gives us \(\vv{\rho_{\text{coh}}}\) from which we can easily reconstruct \(\rho_{\text{coh}}\) and \(U_{\text{opt}}\).

We in the following solve the problem for the one qubit machine (\(n=4\)) and the two qubit machine (\(n=8\)). We then show that , given a general machine, it is sufficient to solve two marginal problems in order to find the optimal unitary cooling the target to the lowest temperature \(r_{\text{coh}}^*\). We also provide \(\vv{\rho_{\text{coh}}^*}\).

\subsection{Coherent One-Qubit Machine\label{sec:Uopt1}}
We want to solve
\begin{equation}
\label{equ:onequbitprob}
\min_{\vv{\rho} \prec \vv{\rho^{\text{in}}}} \vv{\rho} \cdot \vv{H}, \quad \text{s.t. }  \rho_1+\rho_2=r_{\text{coh}},
\end{equation}
where the majorization conditions are simply given by
\begin{equation}
\label{equ:onequbitMaj}
\sum_{i=1}^l \rho_i^{\downarrow} \leq \sum_{i=1}^l \rho_i^{\text{in},\downarrow}, \quad \forall l=1,\dots,4,
\end{equation}
with equality for \(l=4\). First note that \(\rho_1+\rho_2=r\) with the trace condition implies that \(\rho_3+\rho_4=1-r\) such that

\begin{equation}
\begin{aligned}
\vec{\rho} \cdot \vec{H} &= \rho_2 E_B + \rho_3 E_A + \rho_4 (E_A+E_B)\\
&= (r-\rho_1 +\rho_4) E_B+ \underbrace{(1-r)E_A}_{=\text{cste}}
\end{aligned}
\end{equation}
such that in order to minimise \(\vec{\rho} \cdot \vec{H}\), one should minimise \(\rho_4-\rho_1\). This means that \(\rho_1\) should be the greatest component of \(\vec{\rho}\), namely \(\rho_1= \rho_1^{\downarrow}\) and \(\rho_4\) the smallest, namely \(\rho_4=\rho_4^{\downarrow}\). From equation \eqref{equ:onequbitMaj} with \(l=1\) we have \(\rho_1=\rho_1^{\downarrow} \leq \rho_1^{\text{in}}\) and combining equation \eqref{equ:onequbitMaj} with \(l=3\) with the trace condition we get

\begin{equation}
\begin{aligned}
\rho_1^{\text{in},\downarrow}+\rho_2^{\text{in},\downarrow}+\rho_3^{\text{in},\downarrow}+\rho_4^{\text{in},\downarrow}&=\rho_1^{\downarrow}+\rho_2^{\downarrow}+\rho_3^{\downarrow}+\rho_4^{\downarrow}\\
& \leq \rho_1^{\text{in},\downarrow}+\rho_2^{\text{in},\downarrow}+\rho_3^{\text{in},\downarrow}+\rho_4^{\downarrow},
\end{aligned}
\end{equation}
such that \(\rho_4 = \rho_4^{\downarrow} \geq \rho_4^{\text{in},\downarrow}\). In order to minimise \(\rho_4-\rho_1\) we therefore have to choose \(\vec{\rho}\) such that 
\begin{equation}
\begin{aligned}
\rho_1&=\rho_1^{\text{in}}\\
\rho_4&=\rho_4^{\text{in}}.
\end{aligned}
\end{equation}
Plugging these values in the majorization conditions (equation \eqref{equ:onequbitMaj}), we are left with
\begin{equation}
\label{eq:majT}
\begin{aligned}
\rho_2^{\downarrow} &\leq \rho_3^{\text{in}}\\
\rho_2^{\downarrow}+\rho_3^{\downarrow}&= \rho_2^{\text{in}}+\rho_3^{\text{in}}.
\end{aligned}
\end{equation}

As \(\rho_2^{\downarrow} = \max \{\rho_2,\rho_3 \}\) and \(\rho_2^{\downarrow}+\rho_3^{\downarrow} = \rho_2+\rho_3\), these are exactly the conditions for
\begin{equation}
(\rho_2,\rho_3) \prec (\rho_2^{\text{in}},\rho_3^{\text{in}}),
\end{equation}
which means that one can get the majorized vector by applying a T-transform to the initial vector. That is, for some \(t \in [0,1]\),
\begin{equation}
\label{eq:Ttrafo}
\begin{pmatrix}
\rho_2\\
\rho_3
\end{pmatrix} = T \begin{pmatrix}
\rho_2^{\text{in}}\\
\rho_3^{\text{in}}
\end{pmatrix}, \quad T= \begin{pmatrix} t & 1-t \\
1-t & t
\end{pmatrix}.
\end{equation}
This simply follows from the fact that in general \(r \prec s\) iff there exists some doubly stochastic matrix \(D\) such that \(r=Ds\), and that the most general \(2 \times 2\) doubly stochastic matrices are T-tranforms. Now we just have to choose t such that \(\rho_1+\rho_2=r\), that is
\begin{equation}
t=\frac{\rho_1^{\text{in}}+\rho_3^{\text{in}}-r}{\rho_3^{\text{in}}-\rho_2^{\text{in}}},
\end{equation}	
or equivalently

\begin{equation}
1-t = \frac{r-r^{\text{in}}}{r_B-r^{\text{in}}}
\end{equation}
or
\begin{equation}
r=r^{\text{in}}+(1-t) (r_B-r^{\text{in}}).
\end{equation}
The associated work cost is

\begin{equation}
\begin{aligned}
\Delta F &= (\vec{\rho} - \vv{\rho^{\text{in}}}) \cdot \vec{H} \\
&=(1-t) (\rho_3^{\text{in}}-\rho_2^{\text{in}}) (E_B-E_A)\\
&= (1-t) (r_B-r_A) (E_B-E_A).
\end{aligned}
\end{equation}				

A unitary U such that  \(\vec{\rho}=\vv{\text{Diag}}(U\rho^{\text{in}}U^{\dagger})\) is for example given by

\begin{equation}
U = \begin{pmatrix}
1&0&0&0\\
0& \sqrt{1-\mu}& \sqrt{\mu}&0\\
0& - \sqrt{\mu} & \sqrt{1-\mu} &0\\
0&0&0&1
\end{pmatrix},
\end{equation}	
where \(\mu = 1-t\) and can be written more compactly as

\begin{equation}
U= e^{-i \arcsin(\sqrt{\mu}) L_{AB}},
\end{equation}
with \(L_{AB}= i \ket{01}\bra{10} - i \ket{10} \bra{01}\).	
\subsection{Coherent Two-Qubit Machine}
We want to solve
\begin{equation}
\label{equ:twoqubitprob}
\min_{\vv{\rho} \prec \vv{\rho^{\text{in}}}} \vv{\rho} \cdot \vv{H}, \quad \text{s.t. }  \sum_{i=1}^4\rho_i=r_{\text{coh}},
\end{equation}
where the majorization conditions are simply given by
\begin{equation}
\label{equ:majcond}
\begin{aligned}
\sum_{i=1}^k \rho^{\downarrow}_i &\leq \sum_{i=1}^k \rho^{\text{in},\downarrow}_i, \quad \forall k=1,\dots, 7\\
\sum_{i=1}^8 \rho^{\downarrow}_i &= \sum_{i=1}^8 \rho^{\text{in},\downarrow}_i.
\end{aligned}
\end{equation}

The ordering of the original entries is crucial to the solving of the problem. There are hence two regimes that one needs to investigate, namely \(E_C \leq E_A\) and \(E_C > E_A\). We begin with the \(E_C \leq E_A\) regime.

\paragraph{The \(E_C \leq E_A\) regime.}

In this regime the ordering of the original diagonal entries is given by
\begin{equation}
\label{equ:smallorder}
\rho_1^{\text{in}} >\rho_2^{\text{in}}>\rho_5^{\text{in}}>\rho_3^{\text{in}}=\rho_6^{\text{in}}>\rho_4^{\text{in}}>\rho_7^{\text{in}}>\rho_8^{\text{in}}.
\end{equation}

We furthermore have:
\begin{equation}
\label{equ:rhoHsmall}
\begin{aligned}
\vv{\rho} \cdot \vv{H} &= \rho_2 E_C + \rho_3 E_B+ \rho_4 (E_B+E_C)+ \rho_5 E_A \\
& \quad +\rho_6 (E_A+E_C) + \rho_7 (E_A+E_B) \\
& \quad + \rho_8 (E_A+E_B+E_C).
\end{aligned}
\end{equation}

We will next rewrite equation \eqref{equ:rhoHsmall}, keeping in mind the ordering of the original state of equation \eqref{equ:smallorder}, in a way that the majorization conditions of equation \eqref{equ:majcond} can easily be applied. First we use that \(\rho_5+\rho_6+\rho_7+\rho_8=1-r\) and get

\begin{equation}
\label{eq:reformsmall}
\begin{aligned}
\vv{\rho} \cdot \vv{H} &= \underbrace{(\rho_2+\rho_4)}_{r-\rho_1-\rho_3} E_C + \underbrace{(\rho_3+\rho_4)}_{r-\rho_1-\rho_2} E_B \\
& \quad + (1-r) E_A + \underbrace{(\rho_6+\rho_8)}_{1-r-\rho_5-\rho_7} E_C + (\rho_7+\rho_8) E_B\\
&=(r-\rho_1) E_C - \rho_3 E_C + (r-\rho_1-\rho_2) E_C \\
&\quad + (r-\rho_1-\rho_2)E_A+ (1-r) E_A + (1-r-\rho_5) E_C \\
&\quad - \rho_7 E_C + (\rho_7+\rho_8) E_C + (\rho_7+\rho_8) E_A\\
&= (r-\rho_1) E_C + (1-(\rho_1+\rho_2+\rho_5+\rho_3)) E_C\\
&\quad + (1-(\rho_1+\rho_2)) E_A+\rho_8 E_C + (\rho_7+\rho_8) E_A,
\end{aligned}
\end{equation}
where in the second step we used that \(E_B=E_A+E_C\). Note that the sum of the the minima of each summand of \(\vv{\rho} \cdot \vv{H}\) is for sure a lower bound to the minimum of \(\vv{\rho} \cdot \vv{H}\), such that if one can pick a \(\rho\) achieving this lower bound, we will have reached the minimum of \(\vv{\rho} \cdot \vv{H}\). Using the last reformulation of \(\vv{\rho} \cdot \vv{H}\), this is luckily possible, indeed:

\begin{equation}
\label{equ:mincond}
\begin{aligned}
\rho_1 &\leq \rho_1^{\downarrow} \leq \rho_1^{\text{in}, \downarrow} =\rho_1^{\text{in}}\\
\rho_1+\rho_2+\rho_5+\rho_3 &\leq \sum_{i=1}^4 \rho_i^{\downarrow}\leq \sum_{i=1}^4 \rho_{i}^{\text{in},\downarrow}=\rho_1^{\text{in}}+\rho_2^{\text{in}}+\rho_5^{\text{in}}+\rho_3^{\text{in}}\\
\rho_1+\rho_2 &\leq \rho_1^{\downarrow} + \rho_2^{\downarrow} \leq \rho_1^{\text{in},\downarrow}+\rho_2^{\text{in},\downarrow}=\rho_1^{\text{in}}+\rho_2^{\text{in}}\\
\rho_8& \geq \rho_8^{\downarrow} \geq \rho_8^{\text{in},\downarrow} = \rho_8^{\text{in}}\\
\rho_7+\rho_8 &\geq \rho_7^{\downarrow}+\rho_8^{\downarrow} \geq\rho_7^{\text{in},\downarrow}+\rho_8^{\text{in},\downarrow}=\rho_7^{\text{in}}+\rho_8^{\text{in}}.
\end{aligned}
\end{equation}
To minimise the first summand we hence have to choose \(\rho_1=\rho_1^{\text{in}}\). To minimise the third summand, since \(\rho_1=\rho_1^{\text{in}}\), we have to pick \(\rho_2=\rho_2^{\text{in}}\). To minimise the fourth summand we have to choose \(\rho_8=\rho_8^{\text{in}}\) which forces us to choose \(\rho_7=\rho_7^{\text{in}}\) in order to minimise the last summand. We are hence left with the minimisation of the second summand that is achieved if 
\begin{equation}
\label{equ:equal35}
\rho_5+\rho_3=\rho_5^{\text{in}}+\rho_3^{\text{in}}
\end{equation}
 is satisfied.

Now note that we have
\begin{equation}
\begin{aligned}
\rho_1 + \rho_2+ \max\{\rho_3, \rho_5\} &\leq\rho_1^{\downarrow} + \rho_2^{\downarrow} + \rho_3^{\downarrow} \\
&\leq \rho_1^{\text{in},\downarrow} + \rho_2^{\text{in},\downarrow} + \rho_3^{\text{in},\downarrow} \\
&= \rho_1^{\text{in}}+\rho_2^{\text{in}}+\rho_5^{\text{in}}\\
&= \rho_1+\rho_2+\rho_5^{\text{in}}
\end{aligned}
\end{equation}

such that
\begin{equation}
\label{equ:bound3}
\max\{\rho_3,\rho_5\} \leq \rho_5^{\text{in}}.
\end{equation}

Equation \eqref{equ:bound3} and \eqref{equ:equal35} together mean that \((\rho_3,\rho_5) \prec (\rho_5^{\text{in}},\rho_3^{\text{in}})\), which we know from \secref{sec:Uopt1} to be equivalent to
\begin{equation}
\begin{pmatrix} \rho_3 \\ \rho_5 \end{pmatrix} = T_1 \begin{pmatrix} \rho_3^{\text{in}} \\ \rho_5^{\text{in}} \end{pmatrix}, \quad T_1  \begin{pmatrix} t_1 & 1-t_1 \\ 1-t_1 & t_1 \end{pmatrix},
\end{equation}
for some \(t_1 \in [0,1]\). Similarly, as \(\sum_{i=1}^8 \rho_i = \sum_{i=1}^8 \rho_i^{\text{in}}\) we find that 
\begin{equation}
\label{eq:equal46}
\rho_4+\rho_6 = \rho_4^{\text{in}}+ \rho_6^{\text{in}}
\end{equation}

and using that the second line of equation \eqref{equ:mincond} is satisfied with equality we find that
\begin{equation}
\sum_{i=1}^4 \rho_i^{\downarrow} +\max\{\rho_4,\rho_6\} \leq \sum_{i=1}^5 \rho_i^{\downarrow} = \sum_{i=1}^5 \rho_i^{\text{in},\downarrow} = \sum_{i=1}^4 \rho_i^{\downarrow} + \rho_6^{\text{in}}
\end{equation}
such that 
\begin{equation}
\label{eq:maj46}
\max\{\rho_4,\rho_6\} \leq \rho_6^{\text{in}}.
\end{equation}
Now equations \eqref{eq:equal46} and \eqref{eq:maj46} together mean that \((\rho_4 , \rho_6) \prec (\rho_4^{\text{in}}, \rho_6^{\text{in}})\), which as before is equivalent to
\begin{equation}
\begin{pmatrix} \rho_4 \\ \rho_6 \end{pmatrix} = T_2 \begin{pmatrix} \rho_4^{\text{in}} \\ \rho_6^{\text{in}} \end{pmatrix}, \quad T_2  \begin{pmatrix} t_2 & 1-t_2 \\ 1-t_2 & t_2 \end{pmatrix},
\end{equation}
for some \(t_2 \in [0,1]\). This means that for any \(t_1\) and \(t_2\), we have found a vector \(\rho\) that achieves the minimum of each summands in \eqref{eq:reformsmall} and that therefore is the solution of our problem. Of course, for a given \(r\), only some \(t_1\) and \(t_2\) will solve our problem, namely the ones satisfying
\begin{equation}
\begin{aligned}
r&= \rho_1+\rho_2 +\rho_3+\rho_4\\
&= \rho_1^{\text{in}}+\rho_2^{\text{in}}+t_1 \rho_3^{\text{in}}+(1-t_1) \rho_5^{\text{in}}+ t_2\rho_4^{\text{in}}+ (1-t_2) \rho_6^{\text{in}}\\
&= \sum_{i=1}^4 \rho_i^{\text{in}} + (t_1-1) \rho_3^{\text{in}} + (1-t_1) \rho_5^{\text{in}} \\
& \quad + (t_2-1) \rho_4^{\text{in}}+ (1-t_2) \rho_6^{\text{in}}\\
&= r^{\text{in}} + (1-t_1) (\rho_5^{\text{in}}-\rho_3^{\text{in}}) + (1-t_1) (\rho_6^{\text{in}} - \rho_4^{\text{in}}).
\end{aligned}
\end{equation}

Next note that
\begin{equation}
\begin{aligned}
\rho_5^{\text{in}} - \rho_3^{\text{in}}&= (1-r_A) r_B r_C - r_A (1-r_B) r_C \\
&= r_B r_C - r_A r_B r_C - r_A r_C + r_A r_B r_C \\
&= (r_B-r_A) r_C\\
\rho_6^{\text{in}}-\rho_4^{\text{in}}&=(1-r_A) r_B (1-r_C) - r_A (1-r_B) (1-r_C)\\
&= (1-r_C) ( r_B-r_A r_B -r_A + r_A r_B)\\
&= (1-r_C) (r_B-r_A)
\end{aligned}
\end{equation}
such that 
\begin{equation}
r= r^{\text{in}} + [(1-t_1) r_C + (1-t_2) (1-r_C)] (r_B-r_A).
\end{equation}

If we were to choose \(t_1=t_2 = t\) then we would have
\begin{equation}
r= r^{\text{in}} + (1-t) (r_B-r_A).
\end{equation}
Now the work cost of carrying this process is
\begin{equation}
\begin{aligned}
\Delta F &= (\vv{\rho}- \vv{\rho^{\text{in}}}) \cdot \vv{H}\\
&= (t_1 \rho_3^{\text{in}}+ (1-t_1) \rho_5^{\text{in}}-\rho_3^{\text{in}}) E_B\\
& \quad + (t_2 \rho_4^{\text{in}} + (1-t_2) \rho_6^{\text{in}}-\rho_4^{\text{in}}) (E_B+E_C)\\
&\quad + ((1-t_1) \rho_3^{\text{in}}+ t_1 \rho_5^{\text{in}}-\rho_5^{\text{in}}) E_A\\
& \quad +((1-t_2) \rho_4^{\text{in}}+ t_2 \rho_6^{\text{in}}-\rho_6^{\text{in}}) (E_A+E_C)\\
&= (1-t_1) (\rho_5^{\text{in}}-\rho_3^{\text{in}}) E_B\\
&\quad + (1-t_2) (\rho_6^{\text{in}} -\rho_4^{\text{in}}) (E_B+E_C)\\
& \quad + (1-t_1) (\rho_3^{\text{in}}-\rho_5^{\text{in}}) E_A\\
& \quad + (1-t_2) (\rho_4^{\text{in}}-\rho_6^{\text{in}}) (E_A+ E_C)\\
&=[(1-t_1) (\rho_5^{\text{in}}-\rho_3^{\text{in}}) + (1-t_2) (\rho_6^{\text{in}}-\rho_4^{\text{in}}) ] E_C\\
&=[(1-t_1) r_C + (1-t_2) (1-r_C) ] (r_B-r_A) E_C.
\end{aligned}
\end{equation}
If we choose \(t_1=t_2=t\) then
\begin{equation}
\Delta F = (1-t) (r_B-r_A) E_C.
\end{equation}

A unitary \(U\) such that \(\vv{\rho} = \vv{\text{Diag}} (U \rho^{\text{in}} U^{\dagger})\) is for example given by
\begin{equation}
U=\begin{pmatrix}
1&0&0&0&0&0&0&0\\
0&1&0&0&0&0&0&0\\
0&0&\sqrt{1-\mu_1}&0&\sqrt{\mu_1}&0&0&0\\
0&0&0&\sqrt{1-\mu_2}&0&\sqrt{\mu_2}&0&0\\
0&0&- \sqrt{\mu_1}&0&\sqrt{1-\mu_1}&0&0&0\\
0&0&0&-\sqrt{\mu_2}&0&\sqrt{1-\mu_2}&0&0\\
1&0&0&0&0&0&0&0\\
\end{pmatrix},
\end{equation}

where \(\mu_1=1-t_1\) and \(\mu_2= 1-t_2\). If we choose \(t_1=t_2=t\) then \(\mu_1=\mu_2=\mu\) and
\begin{equation}
U=\begin{pmatrix}
1&0&0&0&0&0&0&0\\
0&1&0&0&0&0&0&0\\
0&0&\sqrt{1-\mu}&0&\sqrt{\mu}&0&0&0\\
0&0&0&\sqrt{1-\mu}&0&\sqrt{\mu}&0&0\\
0&0&- \sqrt{\mu}&0&\sqrt{1-\mu}&0&0&0\\
0&0&0&-\sqrt{\mu}&0&\sqrt{1-\mu}&0&0\\
1&0&0&0&0&0&0&0\\
\end{pmatrix}
\end{equation}
can be compactly written as 
\begin{equation}
U=e^{-i \arcsin(\sqrt{\mu}) L_{AB}}, \quad L_{AB}- i \ket{01} \bra{10}_{AB} - i \ket{10}\bra{01}_{AB}.
\end{equation}

\paragraph{The \(E_C > E_A\) regime}

In this regime the ordering of the original diagonal entries is given by
\begin{equation}
\label{equ:bigorder}
\rho_1^{\text{in}} >\rho_5^{\text{in}}>\rho_2^{\text{in}}>\rho_3^{\text{in}}=\rho_6^{\text{in}}>\rho_7^{\text{in}}>\rho_4^{\text{in}}>\rho_8^{\text{in}}.
\end{equation}

As before, we would like to reshuffle the terms of
\begin{equation}
\label{equ:rhoHbig}
\begin{aligned}
\vv{\rho} \cdot \vv{H} &= \rho_2 E_C + \rho_3 E_B+ \rho_4 (E_B+E_C)+ \rho_5 E_A \\
& \quad +\rho_6 (E_A+E_C) + \rho_7 (E_A+E_B) \\
&\quad + \rho_8 (E_A+E_B+E_C)
\end{aligned}
\end{equation}
 such that each summand can be minimised. So we get

\begin{equation}
\begin{aligned}
\vv{\rho} \cdot \vv{H} &= (\rho_2+\rho_4) E_C+ (r-\rho_1-\rho_2) \underbrace{E_B}_{=E_A+E_C}+\rho_5 E_A\\
&\quad +\rho_6 (E_A+E_C)+(1-r-\rho_5-\rho_6) (E_A+E_B)\\
& \quad +\rho_8 E_C\\
&= -\rho_1 E_C+ (-\rho_1-\rho_2) E_A +(-\rho_5) (E_A+E_C)\\
& \quad +\rho_4 E_C +(-\rho_6) E_A+\rho_8 E_C\\
&\quad +\underbrace{(1-r) (E_A+E_B)+r E_B}_c\\
&= -\rho_1 E_C+ (-\rho_1-\rho_5-\rho_2) E_A +\\
&\quad (r-1+\rho_6+\rho_7+\rho_8) E_C+ (\rho_4+\rho_8) E_C \\
& \quad + (-\rho_6) E_A+c\\
&= -\rho_1 E_C+ (-\rho_1-\rho_5-\rho_2) E_A\\
&\quad +  \rho_6 (E_C-E_A)+(\rho_7+\rho_4+\rho_8) E_C +\rho_8 E_C \\
& \quad +\underbrace{c+(r-1)E_C}_d\\
&= -\rho_1 E_C+ (-\rho_1-\rho_5-\rho_2) E_A + d+ \rho_8 E_C\\
& \quad+(\rho_6+\rho_7+\rho_4+\rho_8) (E_C-E_A) \\
& \quad +(\rho_7+\rho_4+\rho_8) E_A. 
\end{aligned}
\end{equation}
Now looking at each summand and applying the majorization conditions with the order of the original vector that we know we get:

\begin{equation}
\label{equ:fixbigEC}
\begin{alignedat}{3}
&\rho_1 \leq \rho_1^{\text{in}} &&\Rightarrow \rho_1=\rho_1^{\text{in}}\\
&\rho_1+\rho_5+\rho_2 \leq \rho_1^{\text{in}}+\rho_5^{\text{in}}+\rho_2^{\text{in}} && \hspace{-0.76cm}\Rightarrow \rho_5+\rho_2=\rho_5^{\text{in}}+\rho_2^{\text{in}}\\
&\rho_8 \geq \rho_8^{\text{in}}&& \Rightarrow \rho_8= \rho_8^{\text{in}}\\
&\rho_7+\rho_4+\rho_8 \geq \rho_7^{\text{in}}+\rho_4^{\text{in}}+\rho_8^{\text{in}} &&\hspace{-0.76cm}\Rightarrow \rho_7+\rho_4=\rho_7^{\text{in}}+\rho_4^{\text{in}}\\
&\rho_6+\rho_7+\rho_4+\rho_8 \geq \rho_6^{\text{in}}+\rho_7^{\text{in}}+\rho_4^{\text{in}}+\rho_8^{\text{in}} &&\Rightarrow \rho_6=\rho_6^{\text{in}}.
\end{alignedat}
\end{equation}

Furthermore, note that out of \(\sum_{i=1}^8 \rho_i=\sum_{i=1}^8 \rho_i^{\text{in}}\) and the above fixed values we get
\begin{equation}
\rho_3=\rho_3^{\text{in}}.
\end{equation}
Also using the majorization conditions, we have
\begin{equation}
\begin{aligned}
\rho_1+\max\{\rho_5,\rho_2\}\leq \rho_1^{\text{in}}+\rho_5^{\text{in}} \Rightarrow \max\{\rho_5,\rho_2\}\leq \rho_5^{\text{in}}\\
\min\{\rho_5,\rho_2\}+\rho_8 \geq \rho_4^{\text{in}}+\rho_8^{\text{in}}\Rightarrow \min\{\rho_5,\rho_2\} \geq \rho_4^{\text{in}},\\
\end{aligned}
\end{equation}
which together with \eqref{equ:fixbigEC} means that \((\rho_5 , \rho_2) \prec (\rho_5^{\text{in}}, \rho_2^{\text{in}})\) and \((\rho_4 , \rho_7) \prec (\rho_4^{\text{in}}, \rho_7^{\text{in}})\) which is equivalent to
\begin{equation}
\begin{pmatrix} \rho_5 \\ \rho_2 \end{pmatrix} = T_1 \begin{pmatrix} \rho_5^{\text{in}} \\ \rho_2^{\text{in}} \end{pmatrix}, \quad T_1  \begin{pmatrix} t_1 & 1-t_1 \\ 1-t_1 & t_1 \end{pmatrix}
\end{equation}
and
\begin{equation}
\begin{pmatrix} \rho_4 \\ \rho_7 \end{pmatrix} = T_2 \begin{pmatrix} \rho_4^{\text{in}} \\ \rho_7^{\text{in}} \end{pmatrix}, \quad T_2  \begin{pmatrix} t_2 & 1-t_2 \\ 1-t_2 & t_2 \end{pmatrix},
\end{equation}
for some \(t_1 \in [0,1]\) and \(t_2 \in [0,1]\). This means that for any \(t_1\) and \(t_2\), the vector \(\rho\) is the solution of our problem. Of course, for a given \(r\), only some \(t_1\) and \(t_2\) will solve our problem, namely the ones satisfying
\begin{equation}
\begin{aligned}
r&= \rho_1+\rho_2 +\rho_3+\rho_4\\
&= r^{\text{in}} + (1-t_1) (\rho_5^{\text{in}}-\rho_2^{\text{in}}) + (1-t_1) (\rho_7^{\text{in}} - \rho_4^{\text{in}})\\
&= r^{\text{in}} + [(1-t_1) r_B + (1-t_1) (1-r_B)] (r_C-r_A).
\end{aligned}
\end{equation}

If we were to choose \(t_1=t_2 = t\) then we would have
\begin{equation}
r= r^{\text{in}} + (1-t) (r_C-r_A).
\end{equation}
Now the work cost of carrying this process is
\begin{equation}
\begin{aligned}
\Delta F &= (\vv{\rho}- \vv{\rho^{\text{in}}}) \cdot \vv{H}\\
&=[(1-t_1) r_B + (1-t_2) (1-r_B) ]\\
&\quad \cdot (r_B-r_A) (E_C-E_A).
\end{aligned}
\end{equation}
If we choose \(t_1=t_2=t\) then
\begin{equation}
\Delta F = (1-t) (r_C-r_A) (E_C-E_A).
\end{equation}

A unitary \(U\) such that \(\vv{\rho} = \vv{\text{Diag}} (U \rho^{\text{in}} U^{\dagger})\) is for example given by
\begin{equation}
U=\begin{pmatrix}
1&0&0&0&0&0&0&0\\
0&\sqrt{1-\mu_1}&0&0&\sqrt{\mu_1}&0&0&0\\
0&0&1&0&0&0&0&0\\
0&0&0&\sqrt{1-\mu_2}&0&0&\sqrt{\mu_2}&0\\
0&- \sqrt{\mu_1}&0&0&\sqrt{1-\mu_1}&0&0&0\\
0&0&0&0&0&1&0&0\\
0&0&0&-\sqrt{\mu_2}&0&0&\sqrt{1-\mu_2}&0\\
0&0&0&0&0&0&0&1\\
\end{pmatrix},
\end{equation}

where \(\mu_1=1-t_1\) and \(\mu_2= 1-t_2\). If we choose \(t_1=t_2=t\) then \(\mu_1=\mu_2=\mu\) then U
can be compactly written as 
\begin{equation}
U=e^{-i \arcsin(\sqrt{\mu}) L_{AC}}, \quad L_{AC}=- i \ket{01} \bra{10}_{AC} - i \ket{10}\bra{01}_{AC}.
\end{equation}
Note however that upon applying this procedure one only finds the solution of our problem for
\begin{equation}
r_A \leq r \leq \rho_1^{\text{in}}+\rho_5^{\text{in}}+\rho_3^{\text{in}}+\rho_7^{\text{in}}.
\end{equation}
To find the solution for
\begin{equation}
\rho_1^{\text{in}}+\rho_5^{\text{in}}+\rho_3^{\text{in}}+\rho_7^{\text{in}} \leq r \leq \rho_1^{\text{in}}+\rho_5^{\text{in}}+\rho_2^{\text{in}}+\rho_3^{\text{in}}
\end{equation}
we use another rearrangement of terms of \(\vv{\rho} \cdot \vv{H}\), namely
\begin{equation}
\begin{aligned}
\vv{\rho} \cdot \vv{H} &= \rho_2 E_C + (r-\rho_1-\rho_2) E_B+\rho_4 E_C\\
& \quad +(1-r-\rho_7-\rho_8) E_A + \rho_6 E_C \\
&\quad + (\rho_7+\rho_8) (E_A+E_B) +\rho_8 E_C\\
&= -\rho_1 E_C + (-\rho_1-\rho_2) E_A + r E_B + (1-r) E_A\\
&\quad (\rho_4+\rho_6+\rho_7+\rho_8) E_C + (\rho_7+ \rho_8) E_A + \rho_8 E_C.
\end{aligned}
\end{equation}

By looking at each summand individually we find that
\begin{equation}
\begin{alignedat}{3}
&\rho_1 \leq \rho_1^{\text{in}} &&\Rightarrow \rho_1=\rho_1^{\text{in}}\\
&\rho_1+\rho_2 \leq \rho_1^{\text{in}}+\rho_5^{\text{in}} &&\Rightarrow \rho_2=\rho_5^{\text{in}}\\
&\rho_8 \geq \rho_8^{\text{in}}&& \Rightarrow \rho_8= \rho_8^{\text{in}}\\
&\rho_7+\rho_8 \geq \rho_4^{\text{in}}+\rho_8^{\text{in}} &&\Rightarrow \rho_7=\rho_4^{\text{in}}\\
&\rho_6+\rho_7+\rho_4+\rho_8 \geq \rho_6^{\text{in}}+\rho_7^{\text{in}}+\rho_4^{\text{in}}+\rho_8^{\text{in}} && \\
& &&\hspace{-1.6 cm}\Rightarrow \rho_4+\rho_6=\rho_6^{\text{in}}+\rho_7^{\text{in}}.
\end{alignedat}
\end{equation}
Using the trace condition we find that 
\begin{equation}
\rho_3+\rho_5=\rho_2^{\text{in}}+\rho_3^{\text{in}}.
\end{equation}
As before this leads to  \((\rho_5 , \rho_3) \prec (\rho_2^{\text{in}}, \rho_3^{\text{in}})\)and \((\rho_6 , \rho_4) \prec (\rho_6^{\text{in}}, \rho_7^{\text{in}})\) which is equivalent to
\begin{equation}
\begin{pmatrix} \rho_5 \\ \rho_3 \end{pmatrix} = T_1 \begin{pmatrix} \rho_2^{\text{in}} \\ \rho_3^{\text{in}} \end{pmatrix}, \quad T_1  \begin{pmatrix} t_1 & 1-t_1 \\ 1-t_1 & t_1 \end{pmatrix}
\end{equation}
and
\begin{equation}
\begin{pmatrix} \rho_6 \\ \rho_4 \end{pmatrix} = T_2 \begin{pmatrix} \rho_6^{\text{in}} \\ \rho_7^{\text{in}} \end{pmatrix}, \quad T_2  \begin{pmatrix} t_2 & 1-t_2 \\ 1-t_2 & t_2 \end{pmatrix},
\end{equation}
for some \(t_1 \in [0,1]\) and \(t_2 \in [0,1]\). This means that for any \(t_1\) and \(t_2\), the vector \(\rho\) is the solution of our problem. Of course, for a given \(r\), only some \(t_1\) and \(t_2\) will solve our problem, namely the ones satisfying
\begin{equation}
\begin{aligned}
r&= \rho_1+\rho_2 +\rho_3+\rho_4\\
&= r_C + [(1-t_1) r_A + (1-t_1) (1-r_A)] (r_B-r_C).
\end{aligned}
\end{equation}

If we were to choose \(t_1=t_2 = t\) then we would have
\begin{equation}
r= r_C+ (1-t) (r_B-r_C).
\end{equation}
Now the work cost of carrying this process is
\begin{equation}
\begin{aligned}
\Delta F &= (\vv{\rho}- \vv{\rho^{\text{in}}}) \cdot \vv{H}\\
&=(r_C-r_A) (E_C-E_A)\\
&\quad+[(1-t_1) r_A + (1-t_2) (1-r_A) ] (r_B-r_C) E_C.
\end{aligned}
\end{equation}
If we choose \(t_1=t_2=t\) then
\begin{equation}
\Delta F = (r_C-r_A) (E_C-E_A)+(1-t) (r_B-r_C) E_C.
\end{equation}
A unitary U such that \(\vv{\rho}=\vv{\text{Diag}} (U \rho U^{\dagger})\) is given by
\begin{equation}
U=U_{35}(\mu_1) U_{46}(\mu_2) U_{25}(1) U_{57}(1),
\end{equation}

where \(\mu_1=1-t_1\) and \(\mu_2= 1-t_2\) and \(\mu \in [0,1]\)
\begin{equation}
U_{ij}(\mu):= \begin{pmatrix}
\sqrt{1-\mu}&\sqrt{\mu}\\
-\sqrt{\mu}&\sqrt{1-\mu}
\end{pmatrix}_{ij} \oplus \mathds{1}_{\bar{ij}}.
\end{equation}
If we choose \(t_1=t_2=t\) then \(\mu_1=\mu_2=\mu\) then U can be written as 
\begin{equation}
U=e^{-i \arcsin(\sqrt{\mu}) L_{AB}} e^{-i \pi/2 L_{AC}}, 
\end{equation}
with \(L_{AC}=- i \ket{01} \bra{10}_{AC} - i \ket{10}\bra{01}_{AC}\) and \(L_{AB}=- i \ket{01} \bra{10}_{AB} - i \ket{10}\bra{01}_{AB}\). We can also summarise both parts of the solution in one unitary. Then U looks like
\begin{equation}
U=U_{35}(\mu_2) U_{46}(\mu_2) U_{25}(\mu_1) U_{57}(\mu_1),
\end{equation}
with \(\mu_1= \min(2 \mu,1), \mu_2= \max(2 \mu-1,0)\), and \(\mu \in [0,1]\). Or

\begin{equation}
U=e^{-i \arcsin(\sqrt{\mu_2}) L_{AB}} e^{-i \arcsin(\sqrt{\mu_1}) L_{AC}}. 
\end{equation}
\subsection{Endpoint of Arbitrary Single-Cycle Machines\label{subsec:endpoint}}

We here want to find the solution of the problem of Equation \eqref{equ:genqubitprob} when r is chosen to be the maximally allowed r. We set \(k=n/2\). We know that r is at most \(r_{\text{coh}}^*=\sum_{i=1}^k \rho_i^{\text{in},\downarrow}\) since 
\begin{equation}
r= \sum_{i=1}^k \rho_i \leq \sum_{i=1}^k \rho_i^{\downarrow} \leq \sum_{i=1}^k \rho_i^{\text{in},\downarrow}
\end{equation}
and choosing \(\rho_i=\rho_i^{\text{in},\downarrow},\, i=1,\dots,n\) achieves this upper bound, i.e. then \(\vv{\rho} \prec \vv{\rho^{\text{in}}}\) and \(r=r_{\text{coh}}^*\). We next want to show that

\begin{lemma} \label{lemma:biggest_entries}
For any state \(\rho\) such that \(r_{\rho}=r_{\text{coh}}^*\), the first k entries of the state are its biggest entries.
\end{lemma}

\begin{proof}
Suppose not, i.e. there exists a \(\rho\) as in the statement for which there exist \( i \leq k, \text{ and } j >k \text{ such that } \rho_i < \rho_j\). Then \(\vv{\rho'}:= P_{ij} \vv{\rho} \prec \vv{\rho} \prec \vv{\rho^{\text{in}}}\) and \(r' = \sum_{l=1}^k \rho_l' = \sum_{l=1, l \neq i}^k \rho_l+\rho_j > \sum_{l=1, l \neq i}^k \rho_l+\rho_i= \sum_{l=1}^k \rho_l=r_{\rho}\). As \(\vv{\rho'} \prec \vv{\rho^{\text{in}}}, \, r' \leq r_{\text{coh}}^*\) so \(r_{\rho} < r_{\text{coh}}^*\) in contractiction with the assumption.
\end{proof}

Writing \(\vv{\rho}\) as 
\begin{equation}
\vv{\rho} =(v_{\rho}, w_{\rho}),
\end{equation}
 with 
 \begin{equation}
 \begin{aligned}
 v_{\rho}&=((v_{\rho})_1,\dots,(v_{\rho})_k) :=(\rho_1,\dots, \rho_k)\\
 w_{\rho}&=((w_{\rho})_{1},\dots,(w_{\rho})_{n-k}) :=(\rho_{k+1},\dots, \rho_n),
 \end{aligned}
 \end{equation}
 the above lemma can be reformulated as
 \begin{equation}
 v_{\rho}^{\downarrow} = (\rho_1^{\downarrow}, \dots, 
 \rho_k^{\downarrow}).
 \end{equation}
 What makes the above equation non trivial is that on the left hand side only the first k entries of \(\rho\) are reordered whereas on the right hand side all the entries of \(\rho\) are reordered. Using that \(\vv{\rho} \prec \vv{\rho^{\text{in}}}\) we have for all \(l=1,\dots, k\) that
 \begin{equation}
 \sum_{i=1}^l (v_{\rho}^{\downarrow})_i = \sum_{i=1}^l \rho_i^{\downarrow} \leq \sum_{i=1}^l \rho_i^{\text{in},\downarrow},
 \end{equation}
 with equality for \(l=k\). This is equivalent to \(v_{\rho} \prec v_{\rho^{\text{in},\downarrow}}\). Also note that the lemma implies that the last \(n-k\) entries of \(\vv{\rho}\) are the smallest \(n-k\) ones, that is
 \begin{equation}
 w_{\rho}^{\downarrow} =(\rho_{k+1}^{\downarrow}, \dots, \rho_{n}^{\downarrow}).
 \end{equation}
 Again using that \(\vv{\rho} \prec \vv{\rho^{\text{in}}}\) we find that for all \(l=1,\dots,n-k\),
 \begin{equation}
 \sum_{i=1}^{l} (w_{\rho}^{\downarrow})_i= \sum_{i=1}^{l} \rho_{k+i}^{\downarrow} \leq \sum_{i=1}^{l} \rho_{k+i}^{\text{in},\downarrow}
 \end{equation}
 with equality for \(l=n-k\). This is equivalent to \(w_{\rho} \prec w_{\rho^{\text{in},\downarrow}}\). So we have proven that 
 \begin{lemma}
 If \(\vv{\rho}\) satisfies \(r_{\rho} = r_{\text{coh}}^*\), then
 \begin{equation}
 \vv{\rho} \prec \vv{\rho^{\text{in}}} \Leftrightarrow v_{\rho} \prec v_{\rho^{\text{in},\downarrow}} \text{ and } w_{\rho} \prec w_{\rho^{\text{in},\downarrow}},
 \end{equation}
 where \(\vv{\rho} = (\underbrace{\rho_1,\dots \rho_k}_{v_{\rho}},\underbrace{\rho_{k+1},\dots \rho_n}_{ w_{\rho}})\).
 \end{lemma}
 Indeed, the reverse implication is trivially satisfied. This means that \(\rho\) is the solution of 
 \begin{equation}
 \label{equ:endprob}
 \min_{\rho \prec \rho^{\text{in}}} \vv{\rho}\cdot \vv{H} , \text{ s.t. } \sum_{i=1}^k \rho_i = r_{\text{coh}}^*
 \end{equation}
  iff \(v_{\rho}\) and \(w_{\rho}\) are solutions of
 \begin{equation}
 \begin{aligned}
 \min_{v_{\rho} \prec v_{\rho^{\text{in},\downarrow}}}& v_{H} \cdot v_{\rho}\\
 \min_{w_{\rho} \prec w_{\rho^{\text{in},\downarrow}}}& w_{H} \cdot w_{\rho},
 \end{aligned}
 \end{equation}
where we split \(H\) in the same way as \(\rho\) in \(H=(v_H,w_H)\). That is, we have reformulated the original constaint problem into two marginal unconstraint problems. The minimums of \(v_H \cdot v_{\rho}\) and \(w_H \cdot w_{\rho}\) are attained by \(\sum_{i=1}^k (v_H^{\downarrow})_i (v_{\rho}^{\uparrow})_i\) and \(\sum_{i=1}^{n-k} (w_H^{\downarrow})_i (w_{\rho}^{\uparrow})_i\), that is when the entries of \(v_{\rho}\) and \(w_{\rho}\) are inversely ordered with respect to the ones of \(v_H\) respectively \(w_H\). This uniquely defines the \(\vv{\rho}\) that minimises \eqref{equ:endprob} and solves the endpoint problem.

\section{single-cycle endpoint free energy}\label{app:endpoint}

We want here to argue that in the two-qubit machine one always needs less free energy to reach the endpoint in the single-cycle coherent scenario than in the single-cycle incoherent scenario. This is formulated in the following

\begin{claim}
\(\Delta F_{\text{coh}}^* \leq \Delta F^*_{\text{inc}}\) with equality iff \(T_R \rightarrow +\infty\), \(E \rightarrow 0\), \(E_C \rightarrow 0\), or \(E_C \rightarrow +\infty\).
\end{claim}

\begin{proof}
If \(T_R \rightarrow +\infty\) one sees directly that for both cases of \(E_C \leq E\) and \(E > E_C\), \(\Delta F_{\text{coh}}^* =0= \Delta F^*_{\text{inc}}\). If \(E \rightarrow 0\), then also in both cases \(\Delta F_{\text{coh}}^*=\Delta F^*_{\text{inc}}\). If \(E_C \rightarrow 0\) we also trivially have \(\Delta F_{\text{coh}}^* =0= \Delta F^*_{\text{inc}}\). If \(E_C \rightarrow +\infty\) both terms go to infinity as \(\mathcal{O}(E_C)\) and are in that sense equal. Similarly one sees that if \(T_R \rightarrow 0\) or \(E \rightarrow + \infty\), \(E_C (r_B-r) < \frac{1}{2}= \Delta F^*_{\text{inc}}.\) Else, assuming \(E_C, E, T_R \notin \{0, + \infty\}\), note that as for \(E_C >E\) we have
\begin{equation}
\Delta F_{\text{coh}}^*= E_C (r_B-r)-E (r_C-r),
\end{equation}
the work cost in the coherent scenario is always bounded by \(E_C (r_B-r)\). In order to prove our point we hence only need to prove that \(E_C (r_B-r) < \Delta F^*_{\text{inc}}\). To do so we look at 
\begin{equation}
f(T_R) = r_C+r-r_B-\frac{1}{2}
\end{equation}
 for fixed \(E_C, E \in ]0, + \infty[\). As

\begin{equation}
f(0)=\frac{1}{2}, \quad f(+\infty)=0,
\end{equation}

if \(f'(T_R) <0\) our point is proven. We hence calculate

\begin{equation}
\begin{aligned}
f'(T_R)&=  -\frac{1}{T_R^2} \left[ E_C r_C (1-r_C) + E r (1-r)-E_B r_B (1-r_B) \right]\\
&= -\frac{1}{T_R^2} \left[ E_C \underbrace{[ r_C (1-r_C)- r_B (1-r_B)]}_{>0}+ E \underbrace{ [ r (1-r)- r_B (1-r_B)]}_{>0}\right]<0,
\end{aligned}
\end{equation}
where in the second step we used that \(g(r)= r (1-r)\) is stricly decreasing on \(]\frac{1}{2}. 1[\) as well as \(\frac{1}{2} < r_C < r_B <1\) and \(\frac{1}{2} < r < r_B <1\). This ends the proof.
\end{proof}

\section{Crossing point}\label{app:cross}

In Sec.~VI~C  of the main text we demonstrate the existence of a critical point \((\Delta F_{\text{crit}},T_{\text{crit}})\) beyond which the coherent scenario outperforms the incoherent one in the single-cycle regime. Note that as both curves start at the same point, this critical point is not the only crossing point of both curves. Our numerical results though strongly suggest that those are the only two. We want here to study the behavior of the more interesting crossing point, \((\Delta F_{\text{crit}},T_{\text{crit}})\), when one varies the environment temperature \(T_R\) and the energy gap \(E_C\). In \figref{fig:crossingFT} we analyse the behaviour of \(\Delta F_{\text{crit}}\) as a function of \(T_R\) for fixed \(E_C\). Apart from the fact that the curves seem smooth, it is interesting to note that they all exhibit a maximum for some environmental temperature. This point corresponds to the environmental temperature for which the crossing between coherent and incoherent
occurs at the lowest temperature of the target qubit (i.e. at maximum cooling).

\begin{figure}[!tbp]
  \centering
    \includegraphics[width=0.7 \textwidth]{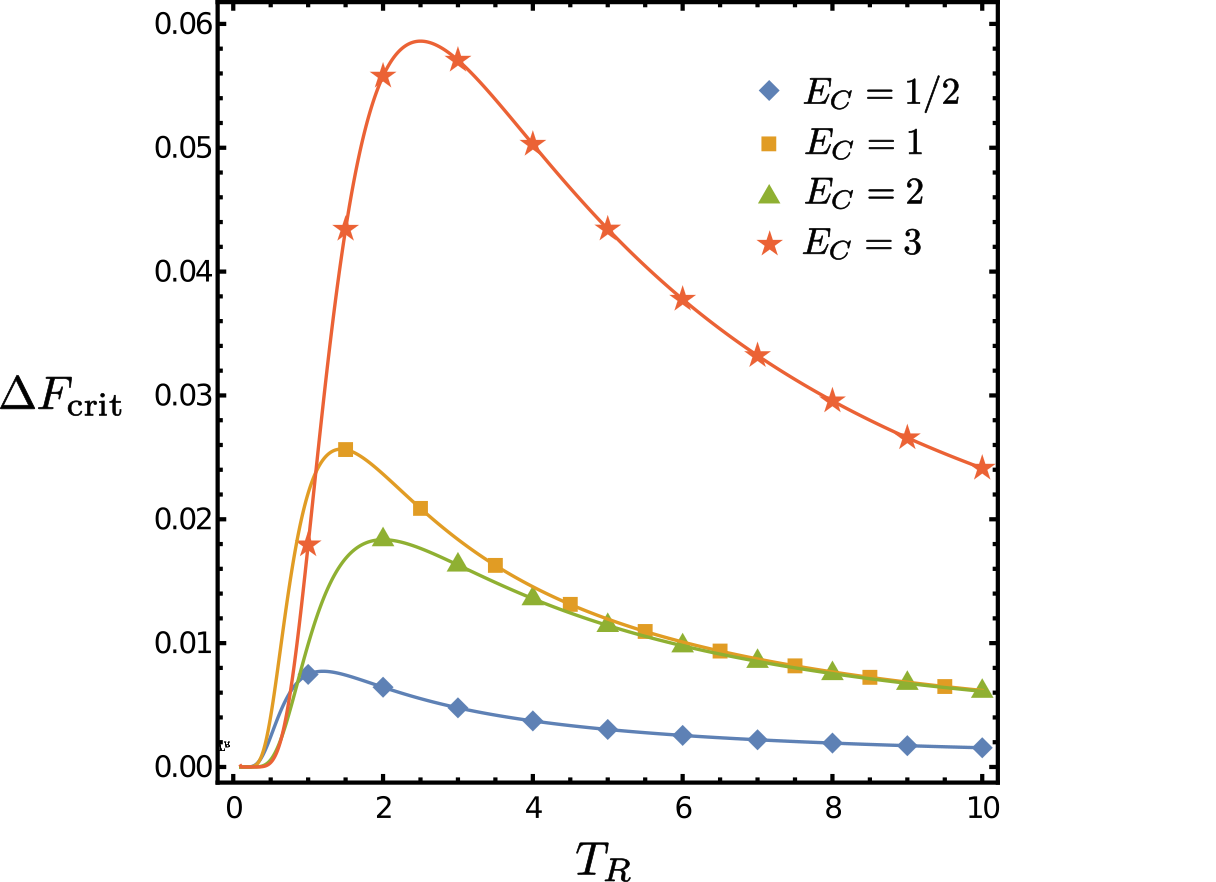}
    \caption{\label{fig:crossingFT}\(\Delta F_{\text{crit}}\) is plotted as a function of \(T_R\) for various fixed \(E_C\).}
\end{figure}

\section{Treating the resource internally}\label{app:internal}

Instead of treating the resource as an external supply, \fab{one can instead consider part of the machine to be the resource itself. We showcase here what such a standpoint would lead to for the two-qubit machine when considering qubit C to be the resource.} One can then ask the same question, namely how do the fully entropic (incoherent) vs. the non-entropic (coherent) supply of free energy scenarios compare in terms of 

\begin{itemize}
\item reachable temperatures
\item reachable temperatures for a given work cost.
\end{itemize}
The incoherent scenario translates to exchanging qubit C with a qubit at a hotter temperature \(T_H\) and then perfoming the energy conserving unitary in the subspace \(\text{span}(\ket{010},\ket{101})\). The free energy difference is now calculated in terms of the system state since the state itself is the resource. We hence have for the final free energy

\begin{align*}
F^{\text{fin}} &= \langle  H\rangle_{\rho^{\text{fin}}} -T_R S_{\rho^{\text{fin}}}\\
&=\text{Tr}(\rho^{\text{fin}} H) + T_R \; \text{Tr}[\rho^{\text{fin}} \ln(\rho^{\text{fin}})]\\
&=\text{Tr}(\rho^{\text{fin}} [H + T_R \ln(\rho^{\text{fin}})])\\
&=T_R \ln (r r_B r^H_C) + E_C ( 1- \frac{T_R}{T_H}) (1-r^H_C) .
        \end{align*}
To calculate the initial free energy note that the initial state \(\rho^{\text{in}} = \tau \otimes \tau_B \otimes \tau_C\) is the same as \(\rho^H = \tau \otimes \tau_B \otimes \tau^H_C\) with \(T_H = T_R\). Hence by setting \(T_H=T_R\) in the above result

\begin{equation}
F^{\text{in}} = T_R \ln(r r_B r_C).
\end{equation}

Therefore 

\begin{equation}
\Delta F_{\text{inc,int}} = F^{\text{fin}} - F^{\text{in}}= E_C \frac{T_H-T_R}{T_H} (1-r^H_C) + T_R \ln(\frac{r^H_C}{r_C})
\end{equation}
The temperature achieved on the target qubit is the same as in the single-cycle incoherent scenario of Sec.~VI~A of the main text and reads 

\begin{align}
r_{\text{inc,int}}&= r r_B+ (1-r_C^H)((1-r) r_B+r (1-r_B))\\
T_{\text{inc,int}}&= \frac{E}{\ln{\frac{r_{\text{inc,int}}}{1-r_{\text{inc,int}}}}}.
\end{align}

The coherent scenario allows one to implement any unitary on qubit C and then perfoming the energy conserving unitary on the 3 qubit system in the subspace \(\text{span }(\ket{010},\ket{101})\). After applying the unitary to qubit C the state looks like

\begin{equation}
\rho^U= \tau \otimes \tau_B \otimes U \tau_C U^{\dagger},
\end{equation}

where

\begin{equation}
U=\begin{pmatrix}
a&b\\
-b^* e^{i \theta}&a^* e^{i \theta}
\end{pmatrix},
\end{equation}
with \(\theta \in [0, 2 \pi]\) and \(\lvert a \rvert ^2+ \lvert b \rvert ^2 =1, \, a,b \in \mathbb{C}\).
And hence

\begin{equation}
\begin{split}
U \tau_C U^{\dagger} &= \begin{pmatrix}
(1-\lvert b \rvert^2) r_C+ \lvert b \rvert^2 (1-r_C)& ab e^{-i \theta} (1-2r_C)\\
a^* b^* e^{i \theta} (1-2r_C) & \lvert b\rvert^2 r_C+ (1-\lvert b \rvert^2) (1-r_C)
\end{pmatrix}\\
&=\begin{pmatrix}
(1-\mu) r_C+\mu(1-r_C)& \sqrt{\mu (1-\mu)} (1-2r_C)\\
\sqrt{\mu (1-\mu)} (1-2r_C)&\mu r_C+(1-\mu) (1-r_C)
\end{pmatrix}\\
&=:\begin{pmatrix}
r_C^U&z\\
z& 1-r_C^U
\end{pmatrix}
\end{split},
\end{equation}
where in the second step we made the choice of a and b being real, \(\theta=0\), and \(b^2=\mu\). Note that making this choice does not influence the perfomance of U since for this only the value of \(r_C^U\), which is not altered by the choice, matters. In any case, to maximally cool the target qubit for a given state of qubit C, one notices that the energy conserving unitary \(U_{\text{cons}}\) need be chosen as \(U_{\text{cons}}= \begin{pmatrix}
0&1\\
1&0
\end{pmatrix}\) in the \(\text{span}(\ket{010}, \ket{101})\) subspace and as identity elsewhere, such that for the final state \(\rho^{\text{fin}}:= U_{\text{cons}} \rho^U U_{\text{cons}}^{\dagger}\) we have

\begin{equation}
\text{Tr}_{BC}(\rho^{\text{fin}})=\begin{pmatrix}
r_{\text{coh,int}}&0\\
0&1-r_{\text{coh,int}}^f
\end{pmatrix},
\end{equation}

with \(r_{\text{coh,int}}:=r r_B+ (1-r_C^U) [(1-r) r_B+ r (1-r_B)]\). And so the final temperature is obtained as usual by

\begin{equation}
T_{\text{coh,int}}= \frac{E}{\ln{\frac{r_{\text{coh,int}}}{1-r_{\text{coh,int}}}}}.
\end{equation}

The free energy cost is obtained as

\begin{equation}
\Delta F_{\text{coh,int}}= \Delta \langle  H\rangle_{\rho} - T_R \Delta S_{\rho}
\end{equation}

Note that as the transformations are all unitaries, \(\Delta S_{\rho}=0\) and so we have

\begin{equation}
\begin{aligned}
\Delta F_{\text{coh,int}}&= \Delta \langle H\rangle_{\rho}\\
&=\text{Tr}\left((\rho^U-\rho^{\text{in}}) H\right)\\
&=(r_C-r_C^U) E_C.
\end{aligned}
\end{equation}

We are now in a position to map out the amount of cooling vs. the associated work cost for both scenarios and compare them. This is displayed in \figref{fig:cohVSincohinternal}. 

\begin{figure}[!tbp]
  \centering
 \includegraphics[width=0.7 \textwidth]{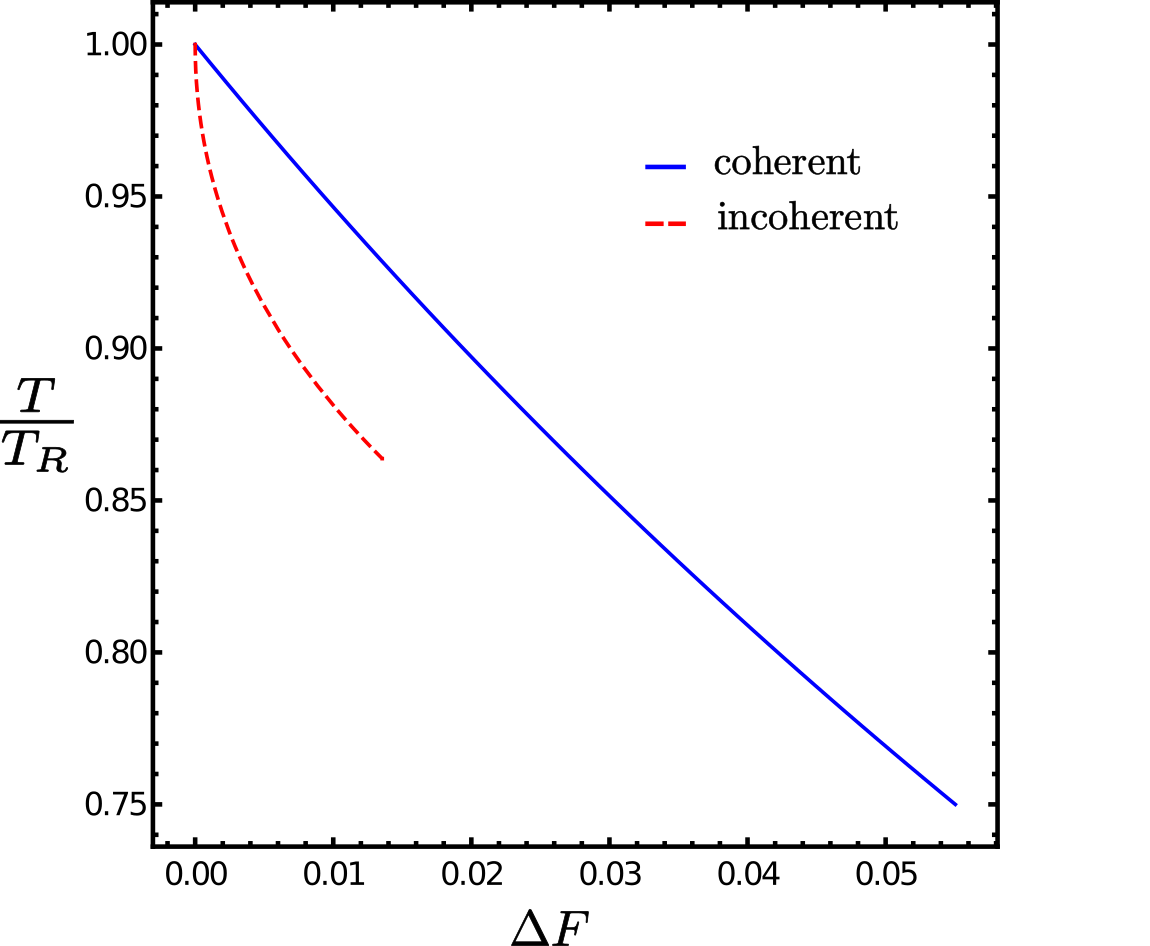}
  \caption{\label{fig:cohVSincohinternal}The internal resource versions of the coherent (solid blue) and the incoherent (dashed red) scenarios are compared. The energy gaps are fixed to \(E=1\) and \(E_C=\frac{1}{3}\) and the environment temperature to \(T_R=1\). One sees that the incoherent scenario always outperforms the coherent one for temperatures that are reachable to both scenarios.}
\end{figure}

Note however that those plots will never cross. Indeed by choosing the same cooling in both scenarios, i.e. \(T_{\text{inc,int}}=T_{\text{coh,int}}\), we have
\begin{equation}
\begin{aligned}
T_{\text{coh,int}}=T_{\text{inc,int}} &\Leftrightarrow r_{\text{coh,int}}=r_{\text{inc,int}}\\
&\Leftrightarrow r_C^U=r_C^H\\
&\Rightarrow \langle H \rangle_{\rho^U}=\langle H \rangle_{\rho^H}\\
&\Rightarrow \Delta F_{\text{coh,int}}=\langle H \rangle_{\rho^U} > \Delta F_{\text{inc,int}}=\langle H \rangle_{\rho^H}-T_R\Delta S_{\rho^H};
\end{aligned}
\end{equation}

meaning that for each temperature that both the incoherent and the coherent scenarios can reach, the incoherent scenario outperforms the coherent one. However, the coherent scenario can always reach lower temperatures than the incoherent one,  that is \(T_{\text{coh,int}}^* < T_{\text{inc,int}}^*\). This hence settles the comparison of both scenarios in a much more trivial way than in the external resource case.

\section{The swap operation and the virtual qubit as a basis for cooling operations}\label{app:virtualqubit}

In all of the paradigms discussed in this work, the operation that causes the target qubit to be cooled down is a swap operation between the target qubit and a qubit subspace in the joint system of the machine qubits. The latter can, but need not be, either one of the machine qubits. The effect of this swap operation can be very simply understood in terms of the ``virtual qubit" subspace of the machine qubits. This section presents the cooling effect of the swap in terms of the virtual qubit, as was done in \cite{silva16}. All of the results in the case of repeated operations (and some of those in the single-cycle regime) follow from this argument. For a proof of the statement, see \cite{silva16}, Appendix A.

Let $A$ be a real (target) qubit system that begins in a state that is diagonal w.r.t. the energy eigenbasis (denoted by $\{\ket{0},\ket{1}\}$, with the population of its ground state (i.e. the corresponding diagonal element in the density matrix) denoted by $r$. Denote the energy difference between the two levels by $E$. In addition, consider another system $M$ (representing the machine), that has in particular a two-dimensional subspace spanned by the energy eigenstates $\{\ket{E_g},\ket{E_e}\}$, this subspace is referred to as the ``virtual qubit". We denote by $E_V = E_e - E_g$ the energy gap of the virtual qubit. The initial state of the machine, expressed as a density matrix in the energy eigenbasis, is assumed to have no coherence w.r.t. the eigenstates of the virtual qubit, i.e. the coefficients of $\ket{E_g}\bra{E_i}$ are zero for all $i$ (except the diagonal element $i=g$), and similarly for $\ket{E_e}\bra{E_i}$.

Let the population in the $\ket{E_g}$ state (the coefficient of $\ket{E_g}\bra{E_g}$ in the density matrix) be denoted as $p_g$, and that in the $\ket{E_e}$ state be denoted by $p_e$. We label by
\begin{itemize}
	\item $N_V$ (the ``norm" of the virtual qubit), the total population in the virtual qubit, $N_V = p_g + p_e$.
	\item $r_V$ the normalized ground state population of the virtual qubit, $r_V = p_g/N_V$, i.e. the population if the virtual qubit was normalized to have $N_V = 1$,
	\item $Z_V$ the bias of the virtual qubit, also normalized, $Z_V = (p_g - p_e)/N_V$.
	\item $T_V$ the virtual temperature of the virtual qubit, calculated by inverting its Gibb's ratio,
	\begin{align}
	\frac{p_e}{p_g} &= e^{-E_V/T_V}.
	\end{align}
	Alternatively, the virtual temperature can also be expressed in terms of the bias, via the relation
	\begin{align}
	Z_V &= \tanh\left( \frac{E_V}{2T_V} \right).
	\end{align}
\end{itemize}

Let a swap operation be performed between the real and virtual qubits, described by the unitary
\begin{align}
U = \openone_{AM} - \ket{0,E_e}_{AM}\!\bra{0,E_e} - \ket{1,E_g}_{AM}\!\bra{1,E_g} + \ket{1,E_g}_{AM}\!\bra{0,E_e} + \ket{0,E_e}_{AM}\!\bra{1,E_g},
\end{align}

Then the final reduced state of the target qubit will have a new ground state population given by
\begin{equation}\label{app:virtualsingler}
\begin{aligned}
	r^\prime &= r + (1-r) p_g- r p_e\\
	&= N_V r_V + \left( 1 - N_V \right) r,
\end{aligned}
\end{equation}
i.e. with probability $N_V$, the new populations of the target qubit are those of the normalized virtual qubit, and with probability $1-N_V$, there is no change. We assume $N_V \neq 0$, as this corresponds to the virtual qubit being empty. 

One can also express the above in the form
\begin{align}\label{app:virtualrecursiver}
\frac{r_V - r^\prime}{r_V - r} &= 1 - N_V.
\end{align}

Thus, if after a single swap, the machine is restored to its state before the unitary, and then the swap is repeated, the recursive relation between $r$ and $r^\prime$ will hold for the new population $r^{\prime\prime}$ in terms of $r^\prime$. In general, if the reset of the machine and the swap are repeated in turn $n$ times, then the ground state population of the target qubit after the $n^{th}$ step will be
\begin{subequations}\label{app:virtualrepeatedr}
	\begin{align}
	\frac{r_V - r^{(n)}}{r_V - r} &= \left( 1 - N_V \right)^n, \\
	\text{Equivalently,} \quad r^{(n)} &= r_V - \left( r_V - r \right) \left( 1 - N_V \right)^n.
	\end{align}
\end{subequations}
In the asymptotic limit of infinite swaps, $r \rightarrow r_V$. This is equivalent to the Gibbs ratio of the target qubit approaching that of the virtual qubit, and the bias of the target qubit approaching $Z_V$.

In terms of temperature, if the target qubit and virtual qubit have the same energy gap ($E=E_V$), then the temperature of the target qubit approaches the virtual temperature with each swap, and in the asymptotic limit, $T\rightarrow T_V$. However, if the energies are unequal, then
\begin{align}\label{app:virtualtemperaturefinal}
T \longrightarrow T_V \frac{E}{E_V},
\end{align}
since it is the Gibbs ratio that equilibrates to that of the virtual qubit.

Finally, one can calculate the work cost of the swap operation. Since it is unitary, the energy difference and free energy difference are the same, and given by
\begin{align}
	\Delta F &= Tr \left( \rho^\prime H \right) - Tr \left( \rho H \right),
\end{align}
where $\{\rho,\rho^\prime\}$ are the initial and final states of the system and machine, and $H$ is the Hamiltonian of the system and machine.

For the degenerate case, i.e. $E=E_V$, one finds the work cost to be zero. For the non-degenerate case, the work cost of a single swap is given by
\begin{align}\label{app:virtualworkcost}
	\Delta F = \left( r^\prime - r \right) \left( E_V - E \right),
\end{align}

To end this section, we list the relevant virtual qubits for each of the paradigms used in this work: (see further sections for details)
\begin{itemize}
	\item For single shot and repeated incoherent operations, the virtual qubit is spanned by the two levels $\{\ket{01}_{BC},\ket{10}_{BC}\}$ of the machine qubits, with the energy gap of the virtual qubit equal to that of the target qubit $E_B-E_C=E$.
	\item For repeated coherent operations and algorithmic cooling, the virtual qubit is spanned by the levels $\{\ket{00}_{BC},\ket{11}_{BC}\}$, with the energy gap being $E_B + E_C$.
	\item For single shot coherent operations, one requires the swap between the target qubit $A$ and the machine qubit $B$, which also falls under the above analysis, here the virtual qubit is simply the machine qubit $B$ (thus $N_V = 1$). The energy gap is thus $E_B$. If $E_C > E$, one also requires the swap between qubits $A$ and $C$, where $C$ can be treated as a virtual qubit.
\end{itemize}

\section{Repeated incoherent operations}\label{app:incohop}

\subsection{The rate of cooling with repeated incoherent operations}\label{app:coolingincoherent}

In the case of incoherent operations, the relevant virtual qubit (see \secref{app:virtualqubit}) is the subspace $\{\ket{01}_{BC},\ket{10}_{BC}\}$ of the machine qubits. When qubit $B$ is at the environment temperature $T_R$ and qubit $C$ at the hot temperature $T_H$, one can calculate the populations and variables of the virtual qubit:
\begin{align}
p_{01} &= r_B \left( 1 - r_C^H \right) \\
p_{10} &= \left( 1 - r_B \right) r_C^H \\
N_{V,\text{inc}} = p_{01} + p_{10} &= r_B \left( 1 - r_C^H \right) + \left( 1 - r_B \right) r_C^H \\
r_{V,\text{inc}} \; (=r_{\text{inc},\infty}) \; &= \frac{ r_B \left( 1 - r_C^H \right)}{ r_B \left( 1 - r_C^H \right) + \left( 1 - r_B \right) r_C^H },
\end{align}
where the labelling of $r_{V,\text{inc}}$ as $r_{\text{inc},\infty}$ will become clear shortly. Equivalently, $r_{V,\text{inc}}$ can be expressed in terms of the virtual temperature of the virtual qubit,
\begin{align}\label{app:incohrepfinalT}
r_{V,\text{inc}} &= \frac{1}{1 + e^{-E/T_V}}, & &\text{where} \quad T_{V,\text{inc}} (=T_{\text{inc},\infty}) = \frac{E}{\frac{E_B}{T_R} - \frac{E_C}{T_H}}.
\end{align}

Thus following the argument in \secref{app:virtualqubit}, the ground state population after $n$ repetitions of the incoherent cycle will be given by
\begin{align}
r_{\text{inc},n} &= r_{V,\text{inc}} - \left( r_{V,\text{inc}} - r \right) \left( 1 - N_{V,\text{inc}} \right)^n.
\end{align}

Thus in the asymptotic limit of infinite repetitions, as $0<N_V\leq 1$, we recover $r_{\text{inc},\infty} = r_{V,\text{inc}}$, and the temperature of the target qubit in this limit is the virtual temperature, i.e. $T_{\text{inc},\infty}=T_{V,\text{inc}}$.

In particular, in the limit that the hot bath is at infinite temperature, $T_H \rightarrow \infty$,
\begin{align}
N_{V,\text{inc}}^* &= \frac{1}{2}, \\
r_{\text{inc},\infty}^* &= r_B, \\
T_{\text{inc},\infty}^* &= T_R \frac{E}{E_B}, \\
r_{\text{inc},n}^* &= r_B - \frac{(r_B - r)}{2^n}.
\end{align}

\subsection{The free energy cost of repeated incoherent operations}\label{app:costincoherent}

Here we calculate the free energy cost of repeating the incoherent operations a finite number of times. Since the resource is the hot bath, we will account for $Q_h$, the heat drawn from it. Among all of the steps involved, only the thermalization of qubit $C$ involves the hot bath, and so it is sufficient to keep track of the populations of the reduced state of qubit $C$ in order to calculate $Q^H$.

We can divide the total heat current into two parts, first off, the amount required to heat up qubit $C$ from the environment temperature $T_R$ to the temperature of the hot bath $T_H$. Following that, there is the repeated heat current required to bring back qubit $C$ to $T_H$ after a cooling swap has been performed.

The first heat current is trivial to calculate, from the difference in the ground state population of $C$ due to heating,
\begin{align}
Q^H_1 &= E_C \left( r_C - r_C^H \right).
\end{align}

For the second part, we have to determine the population change, specifically the reduction in the excited state population of qubit $C$, every time that the cooling swap is performed. However, since the swap is between the levels $\ket{010}$ and $\ket{101}$, we see that whatever the change in the reduced state populations of $C$, the change in the corresponding reduced state populations of $A$ is exactly the same. More precisely, the heat required to reset qubit $C$ before the $n^{th}$ swap (i.e. after the $(n-1)^{th}$ cooling swap) is
\begin{align}
Q^H_n &= E_C \left( r_{\text{inc},n-1} - r_{\text{inc},n-2} \right),
\end{align}
which holds for $n\geq 2$. From the above two expressions, we thus have the cumulative heat current required for $n$ cooling steps,
\begin{align}\label{app:heatcurrentrepeatedincoherent}
Q^H_n &= E_C \left( r_C - r_C^H \right) + E_C \left( r_{\text{inc},n-1} - r \right).
\end{align}
In the asymptotic limit of infinite repetitions, $r_{\text{inc},n-1}$ goes to \(r_{\text{inc},\infty}\), and the resultant expression demonstrates that the total heat current is asymptotically finite.

The work cost is given by the decrease in the free energy of the hot bath w.r.t. the temperature of the environment, which is defined as
\begin{align}
\Delta F_{\text{inc},n} &= Q^H_n - T_R \Delta S_{H,n},
\end{align}
where $\Delta S_{H,n}$ is the \emph{decrease} in the entropy of the bath after $n$ repetitions of the swap. For a thermal bath that stays at equilibrium, as we assume throughout, $\Delta S_{H,n} = Q^H_n/T_H$, leading to
\begin{align}
\Delta F_{\text{inc},n} &= Q^H_n \left( 1 - \frac{T_R}{T_H} \right).
\end{align}

In particular, for the case that $T_H \rightarrow \infty$, in the asymptotic limit of infinite repetitions of the swap,
\begin{align}
\Delta F_{\text{inc},\infty}^* &= E_C \left( r_C - \frac{1}{2} + r_B - r \right).
\end{align}

\section{Asymptotic equivalence of incoherent operations and autonomous refrigerator}\label{app:incoherentauto}

In this section, we show that in the two-qubit machine the final state, and hence the final temperature of the target, as well as the total work cost, are the same as if we had run an autonomous refrigerator between the 3 qubits and waited for the steady state. In other words, since the autonomous refrigerator runs continuously, repeated incoherent operations can be understood as a discretized version of the continuous process. For a discussion on the connection between continuous and discretized versions of quantum thermal machines, see \cite{Raam}. Here we simply review the autonomous 3-qubit fridge introduced in \cite{auto0}, and the equivalence of its steady state parameters with the asymptotic end point of repeated incoherent operations.

In the case of the autonomous fridge, rather than having repeated unitary operations, there is a time-independent interaction Hamiltonian between the three qubits given by
\begin{equation}\label{eq:intHauto}
H_{\text{int}} = g \left( \ket{010}_{ABC}\bra{101} + h.c. \right),
\end{equation}
that acts on the degenerate subspace. Note that this Hamiltonian is a generator of the unitary that swaps the population of the degenerate levels, specifically, $U = \text{exp} \left(-i \frac{\pi}{2g} H_{\text{int}} \right)$.

At the same time, each qubit is coupled to a thermal bath, qubit $B$ to the environment, qubit $C$ to the hot bath. For completeness one could also consider qubit $A$ to be coupled to its own environment, but for simplicity we ignore this effect here. This is to be consistent with the repeated incoherent operations picture, where we did not take into account any coupling between qubit $A$ and an environment in between the cooling operations.

As proven in \cite{auto0}, the three qubits approach a steady state, that is particularly simple in the case that qubit $A$ has no coupling to a bath,
\begin{equation}\label{eq:finalautostate}
\tau_{\text{auto}} \otimes \tau_B \otimes \tau_C^H.
\end{equation}
That is, the steady state is a tensor product state, with qubits $B$ and $C$ thermal at the temperatures of the baths they are respectively coupled to, and qubit $A$ in Gibbs state with temperature
\begin{equation}
T_{\text{auto}} = \frac{E}{\frac{E_B}{T_R} - \frac{E_C}{T_H}}.
\end{equation}

This is the same as $T_{\text{inc},\infty}$, see Eq.~\eqref{app:incohrepfinalT}, that is the asymptotic limit of repeated incoherent operations. Furthermore, it is clear that in the repeated operations, when the number of operations approaches infinite, the cooling swaps stop having an effect, and thus the final states of qubits $B$ and $C$ are Gibbs states at $T_R$ and $T_H$ respectively as these are the temperatures they are reset to after each cooling cycle. Thus the final state of all three qubits is the same in both the autonomous and repeated operations scenario.

\subsection{Free energy equivalence.}\label{app:costauto}

Here we calculate the free energy consumed by the autonomous fridge to go from the initial state to the final state. As the resource is the hot bath, we will calculate the free energy from $Q^H_{\text{auto}}$, the heat drawn from the hot bath. The initial state is that of all three qubits being at the environment temperature $T_R$, while the final state is the tensor product of Gibbs states derived above, see Eq.~\eqref{eq:finalautostate}.

Consider the entire system to be comprised of three parts. Each part consists one of the qubits and the bath that it is attached to (in the case of qubits $B$ and $C$). The only way that energy is exchanged between the different parts is via the energy-preserving interaction Hamiltonian given by Eq.~\eqref{eq:intHauto}. This swaps the populations of the two energy eigenstates $\ket{010}$ and $\ket{101}$, and thus the change in population of qubit $A$ due to the interaction is exactly the same as that in qubit $C$. Since the energy change is given by the population times the energy gap this implies that the energy change of the three parts (at all times during the operation of the machine) must be in proportion to $E:-E_B:E_C$, from the form of $H_{\text{int}}$.

Since part $A$ consists only of the target qubit, the total energy change is simply the difference in energy from the initial to the final state, $E(r - r_{\text{auto}})$. For part $C$, the total energy change is the sum of that of qubit $C$, and that of the hot bath, $E_C (r_C - r_C^H) - Q^H_{\text{auto}}$. Via the preceding argument,
\begin{align}
	\frac{E(r-r_{\text{auto}})}{E} &= \frac{E_C (r_C - r_C^H) - Q^H_{\text{auto}}}{E_C}.
\end{align}

Solving for $Q^H_{\text{auto}}$, we find that
\begin{equation}
Q_{\text{auto}}^H = E_C \left( r_C - r_C^H + r_{\text{auto}} - r \right).
\end{equation}
As $r_{\text{auto}} = r_{\text{inc},\infty}$, this is the same heat current as in the asymptotic limit of infinite repetitions of the incoherent operation, see Eq.~\eqref{app:heatcurrentrepeatedincoherent}.

\section{Repeated coherent operations}\label{app:repeatcoherent}

\subsection{Choosing the best virtual qubit from the machine}\label{app:repeatedcoherentproof}

In this section we investigate the effect and optimal strategy for repeated coherent operations. Here we are allowed to repeatedly perform arbitrary unitary operations on the joint system of the target and machine qubits, with the machine qubits being reset to the temperature of the environment in between (see \figref{fig:repeatedcoherentoperations}). To begin with, we demonstrate that in terms of asymptotic cooling, the best virtual qubit of the machine to choose is that spanned by $\{\ket{00}_{BC},\ket{11}_{BC}\}$.

\begin{figure}[h]
	\centering
	\includegraphics[width=10cm]{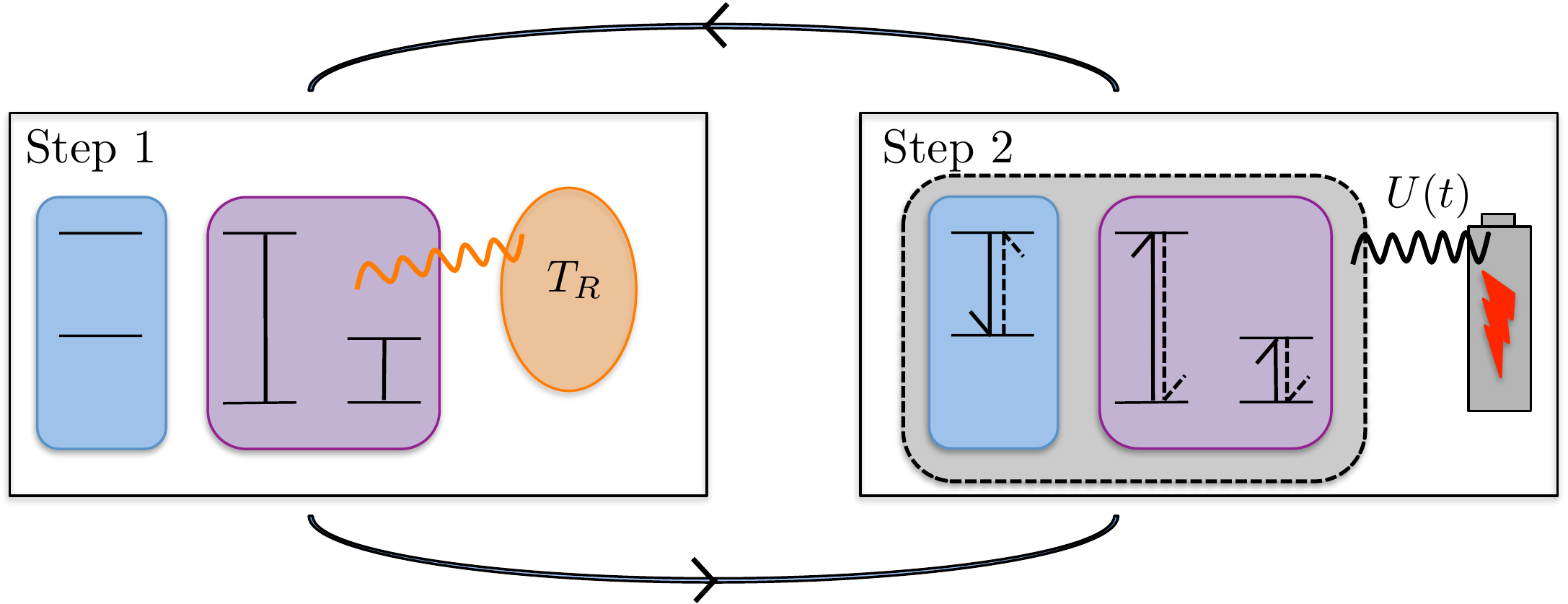}
	\caption{Cooling via repeated coherent operations after the first coherent operation is completed. First the machine qubits $B$ and $C$ are thermalized to the environment temperature $T_R$, following which one performs a unitary that swaps the populations of the levels $\ket{011}$ and $\ket{100}$.\label{fig:repeatedcoherentoperations}}
\end{figure}

First off, w.r.t. the virtual qubit picture, here we can choose any qubit subspace of the machine qubits to swap with the target qubit, unlike in the incoherent case, where we were forced to choose the subspace $\{\ket{01}_{BC},\ket{10}_{BC}\}$, so as to be degenerate ($E_V = E$) with the target system.

However, in the coherent case, there is only a single temperature available ($T_R$), thus the state of the machine after it is rethermalized to the environment will simply be the thermal state of qubits $B$ and $C$ at $T_R$. Given that the entire state of $B$ and $C$ is thermal, every qubit subspace of the machine has the same virtual temperature, $T_V = T_R$.

From \secref{app:virtualqubit}, Eq.~\eqref{app:virtualtemperaturefinal}, we conclude that if we pick a virtual qubit from the machine with energy gap $E_V$, then the temperature of the target system after many repetitions of the swap between the virtual qubit and the target (with the reset of the machine in between) will tend to
\begin{align}
T \longrightarrow T_R \frac{E}{E_V}.
\end{align}
Thus to cooling the maximum amount in the asymptotic case of infinite repetitions, we should pick the largest energy gap, i.e. the qubit subspace $\{\ket{00}_{BC},\ket{11}_{BC}\}$. In what follows, we show that in fact, after the first coherent operation (that was dealt with in \secref{sec:single-cycle2}), this is the \emph{only} virtual qubit that allows for cooling the target.

\subsection{The target qubit after $n$ repetitions of coherent operations}\label{app:repeatedcoherentcooling}

Consider the state of the three qubits at the end of the single coherent operation. The initial state before the operation was the thermal state of all three qubits at the environment temperature $T$. If the energies satisfy $E \geq E_C$, then the optimal coherent operation is simply to swap the states of $A$ and $B$, leaving the three qubits in the state
\begin{align}
\rho^\prime &= \tau_B \otimes \tau_A \otimes \tau_C,
\end{align}
whereas if $E < E_C$, then the optimal coherent operation is to first swap the states of $A$ and $C$, and then proceed by swapping $A$ with $B$, leading to the final state of
\begin{align}
\rho^\prime &= \tau_B \otimes \tau_C \otimes \tau_A.
\end{align}

In either case, there is no further cooling on qubit $A$ possible by any unitary operation. Thus the only option to continue is to bring the machine back to the environment temperature. At this point, the state is now given by
\begin{align}
\tilde{\rho} &= \tau_B \otimes \tau_B \otimes \tau_C\ \\
&=\begin{pmatrix}
r_B^2 r_C &  \multicolumn{7}{c}{\multirow{4}{*}{\Large0}}\\
& r_B^2 \bar{r_C} & & & & & &\\
& & r_B \bar{r}_B r_C & & & & &\\
& & & r_B \bar{r}_B \bar{r}_C & & & & \\
& & & & r_B \bar{r}_B r_C &  & & \\
\multicolumn{5}{c}{\multirow{2}{*}{\Large0}} & r_B \bar{r}_B \bar{r}_C & &\\
& & & &   & & \bar{r}_B^2 r_C &\\
& & & & & & & \bar{r}_B^2 \bar{r}_C \\
\end{pmatrix}
\end{align}

Recall that the first four populations (i.e. eigenvalues) are those in the ground state of qubit $A$. Labelling all of the populations from $p_{000}$ to $p_{111}$, one can verify (using $E_B > E_C$) that
\begin{align}
p_{000} > p_{001} > p_{010} = p_{100} > p_{011} = p_{101} > p_{110} > p_{111}.
\end{align}

Thus from the perspective of maximizing the ground state population of $A$, the only two populations that are not in the optimal location are $p_{011}$ and $p_{100}$, which should be swapped, corresponding to unitarily swapping the two energy levels $\ket{011}$ and $\ket{100}$. This is unlike the initial state before the first coherent operation, where there were a number of possible level swaps that achieved cooling. There one had to optimize over all possible swap operations to minimize the work cost, whereas here there is only one possible cooling swap.

Thus the second coherent operation continues with the $\ket{100}\leftrightarrow\ket{011}$ swap, cooling down qubit $A$ further, followed by bringing back the machine qubits $B$ and $C$ to the environment temperature. One can verify that after resetting the machine qubits, the populations once again satisfy $p_{011} < p_{100}$, allowing cooling to continue by repetition of this cycle of steps. In the same manner as for repeated incoherent operations, from the arguments of \secref{app:virtualqubit}, one can identify the properties of the relevant virtual qubit in this case, the states $\ket{00}_{BC}$ and $\ket{11}_{BC}$ of the machine,
\begin{align}
p_{00} &= r_B r_C \\
p_{11} &= \left( 1 - r_B \right) \left( 1 - r_C \right) \\
N_{V,\text{coh}} = p_{00} + p_{11} &= r_B r_C + \left( 1 - r_B \right) \left( 1 - r_C \right) \\
r_{V,\text{coh}} \; (=r_{\text{coh},\infty}) \; &= \frac{ r_B r_C }{ r_B r_C + \left( 1 - r_B \right) \left( 1 - r_C \right) },
\end{align}
where the labelling of $r_{V,\text{coh}}$ as $r_{\text{coh},\infty}$ will become clear shortly. Equivalently, $r_{V,\text{coh}}$ can be expressed in terms of the virtual temperature of the virtual qubit,
\begin{align}\label{app:icohrepfinalT}
r_{V,\text{coh}} &= \frac{1}{1 + E^{-E_V/T_{V,\text{coh}}}}, & &\text{where} \quad E_V = E_B + E_C \quad \text{and} \quad T_{V,\text{coh}} = T_R.
\end{align}

Thus following the argument in \secref{app:virtualqubit}, the ground state population after $n$ repetitions of the incoherent cycle will be given by
\begin{align}
r_{\text{coh},n}^* &= r_{V,\text{coh}} - \left( r_{V,\text{coh}} - r \right) \left( 1 - N_{V,\text{coh}} \right)^n.
\end{align}

Thus in the asymptotic limit of infinite repetitions, $r^*_{\text{coh},\infty} = r_{V,\text{inc}}$, and the temperature of the target qubit in this limit is
\begin{align}
T^*_{\text{coh},\infty} &= T_{V,\text{coh}} \frac{E}{E_{V,\text{coh}}} = T_R \frac{E}{E_B + E_C}.
\end{align}

\subsection{The free energy cost of cooling with repeated coherent operations}\label{app:costcoherent}

In the case of repeated coherent operations, the work cost is only calculated from the unitary swap operations, as the other step is the thermalization of the machine to the environment temperature, which comes at no cost. To calculate the work cost of the unitary operations, we follow the argument in \secref{app:virtualqubit}. From the argument therein (Eq.~\eqref{app:virtualworkcost}), the free energy input in each repeated coherent operation is given by
\begin{align}
\Delta F_{\text{coh},n}^* - \Delta F_{\text{coh},n-1}^* &= \left( r_{\text{coh},n}^* - r_{\text{coh},n-1}^* \right) \left( E_B + E_C - E \right) \\
&= 2 E_C \left( r_{\text{coh},n} - r_{\text{coh},n-1} \right).
\end{align}

This only applies for $n \geq 2$ since the first coherent operation is different, and the optimal work cost of the latter ($\Delta F^*_{\text{coh}}$) has been calculated in Sec.~VI~B of the main text. Recalling that the ground state population of the target qubit after a single coherent operation is $r_B$, we can calculate the work cost of $n$ repetitions of coherent operations,
\begin{align}
\Delta F_{\text{coh},n}^* &= \Delta F_{\text{coh}}^* + 2 E_C \left( r_{\text{coh},n}^* - r_B \right), \\
\text{where} \quad \Delta F_{\text{coh}}^* &= \begin{cases}
E_C \left( r_B - r \right) & \text{if $E_C \leq E$,} \\
\left( E_C - E \right) \left( r_C - r \right) + E_C \left( r_B - r_C \right) & \text{if $E_C \geq E$.}
\end{cases}
\end{align} 

\section{Algorithmic cooling}\label{app:algo}

In the case of repeated coherent operations, the minimum temperature achievable by the target qubit is bound by the maximum bias $Z_V$ (see \secref{app:virtualqubit})  that can be engineered on any qubit subspace of the machine qubits $B$ and $C$. When the qubits are both thermalized to the environment temperature $T_R$, the maximum bias is on the virtual qubit of $\{\ket{00}_{BC},\ket{11}_{BC}\}$.

However, if one is allowed to thermalize the machine qubits separately, then an even higher bias can be engineered on the same subspace, by pre-cooling qubit $C$. Specifically, after the cooling swap of the target qubit with the virtual qubit $\{\ket{00}_{BC},\ket{11}_{BC}\}$, only qubit $B$ is rethermalized to the environment temperature, and then its state is swapped with that of qubit $C$, thus cooling the state of $C$. Qubit $B$ is then rethermalized to $T_R$, and then the cooling swap involving all three qubits is repeated.

\begin{figure}[h]
	\centering
	\includegraphics[width=10cm]{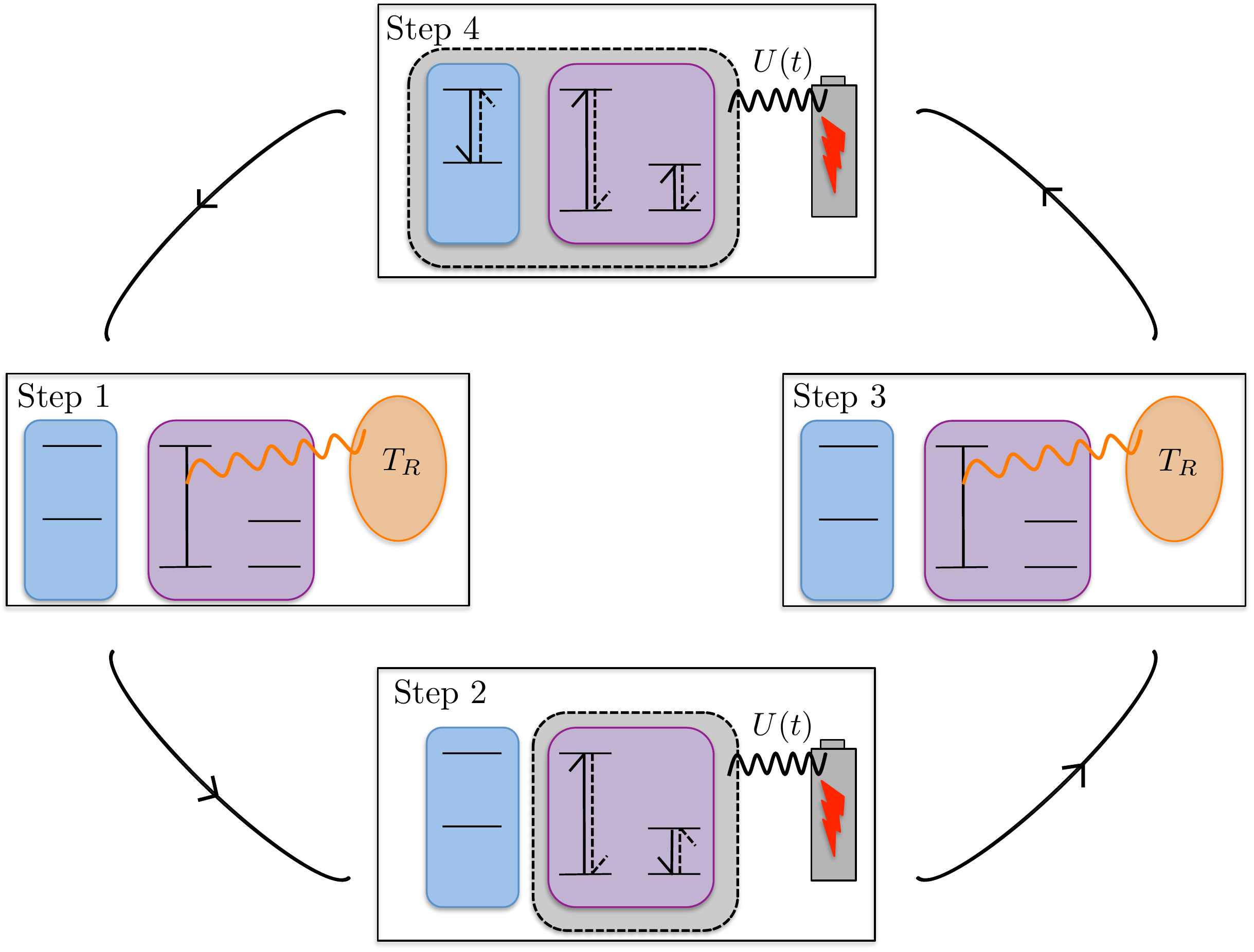}
	\caption{The cycle of steps corresponding to algorithmic cooling. Steps 1 and 3 thermalize qubit $B$ to the environment. Step 2 is the pre-cooling of qubit $C$ by a swap with $B$. Step 4 is the cooling of the target qubit via the usual coherent operation. In the case of optimizing algorithmic cooling w.r.t. the work cost (see \secref{app:algopara}), Step 2 is replaced by a partial rather than full swap.\label{fig:algorithmiccooling}}
\end{figure}

The state of the machine qubits prior to the swap is now a tensor product of two copies of the thermal state of qubit $B$ w.r.t. $T_R$, and so the virtual qubit $\{\ket{00}_{BC},\ket{11}_{BC}\}$ has the following properties
\begin{align}
p_{00} &= r_B^2 \\
p_{11} &= \left( 1 - r_B \right)^2 \\
N_{V,\text{algo}} = p_{00} + p_{11} &= r_B^2 + \left( 1 - r_B \right)^2 \\
r_{V,\text{algo}} \; (=r_{\text{algo},\infty}^*) \; &= \frac{ r_B^2}{ r_B^2 + \left( 1 - r_B \right)^2 },
\end{align}
where the labelling of $r_{V,\text{algo}}$ as $r_{\text{algo},\infty}^*$ will become clear shortly. Equivalently, $r_{V,\text{algo}}$ can be expressed in terms of the virtual temperature of the virtual qubit,
\begin{align}\label{app:algofinalT}
r_{V,\text{algo}} &= \frac{1}{1 + e^{-E_{V,\text{algo}} /T_{V,\text{algo}}}}, & &\text{where} \quad E_{V,\text{algo}} = E_B + E_C \quad \text{and} \quad T_{V,\text{algo}} = T_R \frac{E_B+E_C}{2 E_B}.
\end{align}

Thus following the argument in \secref{app:virtualqubit}, the ground state population after $n$ repetitions of algorithmic cooling will be given by
\begin{align}
r_{\text{algo},n} &= r_{V,\text{algo}} - \left( r_{V,\text{algo}} - r_0 \right) \left( 1 - N_{V,\text{algo}} \right)^n,
\end{align}
where \(r_0\) is the ground state population of the target before starting the algorithmic cooling procedure. \(r_0\) can be \(r\), in the case that we begin with algorithmic cooling from the initial state, but can also be anything else, in particular some population greater than \(r\), corresponding to the endpoint of a different type of cooling operation. Finally note that in the asymptotic limit of infinite repetitions, $r_{\text{algo},\infty}^* = r_{V,\text{algo}}$, and the temperature of the target qubit in this limit is given by
\begin{align}
T_{\text{algo},\infty}^* &= T_{V,\text{algo}} \frac{E}{E_{V,\text{algo}}} = T_R \frac{E}{2E_B},
\end{align}
which is independant of \(r_0\), the initial ground state population of the target.

\subsection{The free energy cost of algorithmic cooling}\label{app:costalgo}

Analogous to the case of repeated coherent operations, here the work cost is invested during the unitary operations. However, in addition to the cooling swap involving all three qubits, whose cost is calculated in exactly the same way as in the repeated coherent case, see \secref{app:costcoherent}, there is also the pre-cooling of qubit $C$, which is a non-energy preserving unitary operation. Since this is effected by a swap between qubits $B$ and $C$, the work cost per population swapped (in the direction of cooling $C$) is $E_B - E_C = E$.

The work cost of pre-cooling $C$ can be split into two contributions: first, the initial cost of cooling $C$ from the environment temperature $T_R$ to the state that has the same populations as that of the initial state of $B$, that costs $E (r_B - r_C)$, and then the work cost of returning it to the pre-cooled state after every successive three-qubit swap. Since the three qubit swap is between the states $\ket{011}$ and $\ket{100}$, we see that whatever the change in the population of the ground state of the target qubit, there is exactly the same decrease in the ground state population of qubit $C$.

Adding up all of these contributions, one finds that the free energy cost of algorithmic cooling is given by
\begin{align}
\Delta F_{\text{algo},n} &= 2 E_C \left( r_{\text{algo},n} - r_0 \right) + E \left( r_B - r_C \right) + E \left( r_{\text{algo},n-1} - r_0 \right),
\end{align}
where the first term is the total work cost of the cooling swap on all three qubits, the second term is the cost of pre-cooling qubit $C$ from its initial state thermal at $T_R$, and the third represents the cost of returning qubit $C$ to the pre-cooled state prior to the $n^{th}$ cooling swap. As before, \(r_0\) is the ground state population of the target before starting the algorithmic cooling procedure.

\section{Optimizing the repetition of coherent operations w.r.t. the work cost}\label{app:optimalcoherent}

In the case of coherent operations, we now have a number of different procedures for cooling. Recall that in the single-cycle case, we found that we could cool by simply swapping the target qubit $A$ with $B$. Furthermore, if it is the case that $E < E_C$, then a lower work-cost can be achieved by swapping the target qubit with $C$ to begin with. For repeated coherent operations, we have to swap the target qubit with the virtual qubit $\{\ket{00}_{BC},\ket{11}_{BC}\}$. And finally, to cool the maximum we should precool qubit $C$ (which is a swap between qubits $B$ and $C$) prior to the same cooling swap.

Each of these processes has a different work cost, and it is illuminating to construct the optimal manner of combining them to have the minimum work cost. Following the argument in \secref{app:virtualqubit}, Eq.~\eqref{app:virtualworkcost}, we understand that to optimize the work cost, we should always seek to swap the target qubit with a virtual (or real) qubit of as small an energy gap as we can find, given it has a greater normalized ground state population \(r_V\) than the ground state population of the target. This way we minimize the energy gradient over which we move population, and thus minimize the work cost.

At the beginning, when the target and machine qubits are at the environment temperature, if $E_C > E$, then one can verify from the machine state that among all the virtual qubits of the machine with greater normalized ground state population \(r_V\) than r, qubit $C$ (here it is a real qubit, rather than virtual) is the one that has the smallest energy difference with E, \(E_V-E\). Thus the minimal cost of cooling is to swap these states, taking $r\rightarrow r_C$ at a gradient of $E_C - E$.

Once this procedure is exhausted and the ground state population of qubit A has become \(r_C\), we find that among the above virtual qubits of the machine, qubit B has the second smallest energy difference with \(E\), and so one proceeds by swapping the target qubit with qubit $B$, taking $r\rightarrow r_B$, at a gradient of $E_B - E$. One then rethermalises the machine qubits to \(T_R\). Note that qubit C could have equivalently been rethermalised at any point between the end of the first swap and now without affecting the cooling and the work-cost of the procedure.

At this point, after resetting the machine qubits, we find that the only virtual or real qubit in the machine that allows for cooling is the virtual qubit $\{\ket{00}_{BC},\ket{11}_{BC}\}$, and one proceeds by repeatedly swapping the target qubit with this virtual qubit, until $r\rightarrow r_{V,\text{coh}}$. This is performed at a gradient of $E_B + E_C - E$.

Finally, one proceeds via algorithmic cooling, where one precools qubit $C$, at a gradient of $E_B - E_C$, before applying the same cooling swap as in the case of repeated coherent operations. The reason one exhausts the repeated coherent operations procedure before proceeding with algorithmic cooling is that precooling qubit \(C\) has a work-cost that arguably enables one to cool more but still at the same energy rate, \(2 E_C\). Thus, as long as cooling without this extra work-cost is possible, it is more efficient to do so.

The work cost at an intermediate stage in this process can be simply calculated from the above, we present here the total work cost of the entire procedure,
\begin{align}\label{eq:finalworkcost}
\Delta F_{\text{algo},\infty}^* &= \quad \left( r_C - r \right) \left( E_C - E \right) +  \left( r_B - r_C \right) \left( E_B - E \right) \nonumber\\
&\quad + \left( r_{\text{coh},\infty}^* - r_B \right) 2E_C \nonumber\\
&\quad + \left( r_B - r_C \right) \left( E_B - E_C \right) + \left( r_{\text{algo},\infty}^* - r_{\text{coh},\infty}^* \right) \left( \left( E_B - E_C \right) + 2 E_C \right),
\end{align}
where the first, second and third lines correspond to the work cost of the single-cycle, repeated and algorithmic sections of the protocol repsectively. In the case of $E_C \leq E$, the single shot case simplifies to directly swapping the target qubit with qubit $B$, and thus the first line of the work cost becomes $\left( r_B - r \right) \left( E_B - E \right)$.

It is interesting to observe that subdividing the entire procedure in this manner, the temperature of the target qubit evolves due to each subsection as
\begin{align}
T \xrightarrow{E<E_C} T_R \frac{E}{E_C}  \quad\longrightarrow\quad T_R \frac{E}{E_B} \quad\longrightarrow\quad T_R \frac{E}{E_B+E_C} \quad\longrightarrow\quad T_R \frac{E}{2E_B}.
\end{align}

\subsection{An optimal cooling sequence in the regime of algortihmic cooling}\label{app:algopara}

In the analysis above, we noted that algorithmic cooling is more expensive as it requires the pre-cooling of qubit $C$. Furthermore, if one pre-cools $C$ via a full swap with $B$, as presented above, this represents an initial work cost which does not cool down the target at all, representing a discontinuity in the curve of cooling vs work cost. This is especially relevant if the desired final temperature is not that corresponding to algorithmic cooling, but is rather somewhere in-between algorithmic cooling and the endpoint of repeated coherent operations.

In this case, one can optimize the work cost by using the same cycle of steps as in Fig. \ref{fig:algorithmiccooling}, but only \emph{partially} pre-cooling qubit $C$ in Step 2, to exactly the temperature required to achieve the desired final temperature on the target.

More precisely, consider that during Step 2, one performs a \emph{partial} swap between qubits $B$ and $C$, such that the final ground state population of qubit $C$ is given by
\begin{align}
r_C(\nu) &= r_C + \nu \left( r_B - r_C \right),
\end{align}
where $\nu \in [0,1]$. On inspection of the virtual qubit $\{\ket{00}_{BC},\ket{11}_{BC}\}$, we can calculate the normalized ground state population $r_V (\nu)$,
\begin{align}\label{eq:algopara}
r_{V,\nu \text{algo}} &= \frac{ r_B \cdot r_C(\nu)}{r_B \cdot r_C(\nu) + (1-r_B) (1-r_C(\nu))}.
\end{align}
Note that $r^*_{\text{coh},\infty} < r_{V,\nu \text{algo}} < r^*_{\text{algo},\infty}$, with $r_{V,0 \text{algo}} = r_{\text{coh},\infty}^*$ and $r_{V,1 \text{algo}} = r^*_{\text{algo},\infty}$, and thus $\nu$ parametrizes the entire regime of cooling between the endpoint repeated coherent operations, and full algorithmic cooling.

In the limit of infinite repetitions of the cycle of steps, the ground state population of the target becomes $r_{V,\nu \text{algo}}$, such that, given the desired final temperature of the target, \(T_{\nu \text{algo},\infty}^*\), the swapping parameter \(\nu\) need be chosen such that
\begin{equation}
r_{V,\nu \text{algo}}= \frac{1}{1+e^{-\frac{E}{T_{\nu \text{algo},\infty}^*}}}.
\end{equation}

The work cost of cooling the target to $r_{V,\nu \text{algo}}$, given that we began with the target ground state population of $r_0$, is found by adding up the cost of pre-cooling qubit $C$, the cost of returning it the pre-cooled state, and the cost of the repeated cooling swaps on the target,
\begin{align}\label{eq:algoparawork}
\Delta F_{\nu \text{algo},\infty} &= E \left( r_C (\nu) - r_C \right) + E \left( r_{V,\nu \text{algo}} - r_0 \right) + 2 E_C \left( r_{V,\nu \text{algo}} - r_0 \right).
\end{align}

Thus given the endpoint of repeated coherent operations, (where $r_0 = r^*_{\text{coh},\infty}$), the above expression represents the optimal extra work cost for cooling the target to a ground state population (Eq. \ref{eq:algopara}) that is between the end points of repeated coherent operations and algorithmic cooling. The total work cost of the optimal sequence is in this case therfore given by
\begin{align}\label{eq:nufinalworkcost}
\Delta F_{\nu \text{algo},\infty}^* &= \quad \left( r_C - r \right) \left( E_C - E \right) +  \left( r_B - r_C \right) \left( E_B - E \right) \nonumber\\
&\quad + \left( r_{\text{coh},\infty}^* - r_B \right) 2E_C \nonumber\\
&\quad + \left( r_C(\nu) - r_C \right) \left( E_B - E_C \right) + \left( r_{V,\nu \text{algo}} - r_{\text{coh},\infty}^* \right) \left( \left( E_B - E_C \right) + 2 E_C \right).
\end{align}
 Note that for $\nu=1$, we recover the previously discussed total work cost of the optimal sequence of coherent operations of Eq. \ref{eq:finalworkcost}.

\section{Optimizing the work cost}\label{app:secondlaw}

\subsection{The N qubit coherent machine}

In this section we review a result demonstrated in \cite{paul} (within a different context), that given a final cold temperature, there exists a family of coherent machines, each member of increasing size, that can attain the final temperature, and that saturate the second law of thermodynamics in the limit of infinite size. We do this for coherent machines first, and prove the same for incoherent machines in the next section.

Consider the system to be a qubit of energy $E$ (the result may be generalized by cooling individual qubit subspaces of a more complicated system), that begins at the environment temperature $T_R$. The final temperature that we would like to attain is labelled $T_C$, where $T_C < T_R$.

The simplest machine to do so would be a single qubit of energy
\begin{align}
	E_{coh,max} &= E \frac{T_R}{T_C},
\end{align}
as in Sec. IV B of the main text, and perform a swap in the energy eigenbasis. Note that the machine is assumed, as always, to begin at $T_R$.

As discussed in the main text and in \secref{app:virtualqubit} on the virtual qubit, the work cost of this protocol involves pushing population against the energy gradient between the machine and system, $E_{coh,max} - E$,
\begin{align}
	W &= (r_{max} - r)(E_{coh,max} - E) = (r_{max} - r)E \left( \frac{T_R}{T_C} - 1 \right).
\end{align}
where $r$ and $r_{max}$ are the initial and final ground state populations of the target.

One can reduce this work cost by splitting the protocol into a number of steps. Consider that the machine is constructed out of a sequence of $N$ qubits, with linearly increasing energy,
\begin{align}
	E_{coh,i} &= E + \frac{i}{N} (E_{coh,max} - E) = E \left( 1 + \frac{i}{N} \left( \frac{T_R}{T_C} - 1 \right) \right), \quad i \in \{1,2,...,N\}.
\end{align}

The protocol now consists in performing swap operations between the target qubit and each of the machine qubits in sequence. The final temperature is the same as before, since the final qubit has energy $E_{coh,max}$. At each intermediate step, the temperature attained by the target is given by
\begin{align}
	\frac{E}{T_i} &= \frac{E_{coh,i}}{T_R} = \frac{E}{T_R} + \frac{i}{N} \left( \frac{E_{coh,max} - E}{T_R} \right) \\
	\therefore \frac{1}{T_i} &= \frac{1}{T_R} + \frac{i}{N} \left( \frac{1}{T_C} - \frac{1}{T_R} \right).\label{eq:cohseqtemp}
\end{align}

Correspondingly, the ground state population of the target after the $i^{th}$ step is given by
\begin{align}
	r_i &= \frac{1}{1 + e^{-E/T_i}}.
\end{align}

The work cost of the $i^{th}$ step is now
\begin{align}
	W_{coh,i} &= (r_i - r_{i-1}) (E_{coh,i} - E),
\end{align}
from which the total cost follows as
\begin{align}
	W_{coh} = \sum_i W_{coh,i} &= \sum_{i=1}^N \left( r_i - r_{i-1} \right) \left( E_{coh,i} - E \right) \\
		&= E\sum_{i=1}^N \left( r_i - r_{i-1} \right) \frac{i}{N} \left( \frac{T_R}{T_C} - 1 \right).\label{eq:workcostcohseq}
\end{align}

In \cite{paul}, this protocol was studied, and it was shown that the total work cost was equal to
\begin{align}
	W_{coh} &= \Delta F + O \left( \frac{1}{N} \right),
\end{align}
where $\Delta F$ is the increase in free energy of the system from its initial to final temperature, and where the free energy is defined w.r.t. the environment temperature,
\begin{align}
	F &= \langle{E}\rangle - T_R S,
\end{align}
$\langle{E}\rangle$ and $S$ being the average energy and entropy of the system. Thus one can get arbitrarily close to saturating the second law of thermodynamics by increasing the number of steps involved in the protocol.

Note that the qubits in the coherent machine need not be real qubits, but qubit subspaces (virtual qubits) embedded in a larger space. In this case, rather than a single swap for each of the machine qubits, one would require repeated swaps (inter-spaced with rethermalization of the machine) to approach the asymptotic temperature corresponding to that qubit. This does not however change the work cost of the procedure, since the cost is always given by the amount of population changed multiplied by the energy gradient, so repeating the swap with the same virtual qubit a number of times to achieve the same population difference as with a real qubit of the same energy gap results in the same work cost.

\subsection{The 2N qubit incoherent machine}\label{app:incsecondlaw}

Consider as before that we wish to cool a target qubit of energy $E$ from the environment temperature $T_R$ to $T_C$, but only using incoherent operations, that include energy-preserving unitaries, and heating up parts of our machine to a given hot temperature $T_H$.

The simplest manner of achieving this temperature is via the simplest possible incoherent machine, comprised of two qubits (as in VII.A of the main text), of energies $E_B = E_{inc,max}$ and $E_C = E_{inc,max} - E$. The machine may be run in the repeated operations regime, where the energy preserving swap operation between the states $\ket{010}_{ABC}$ and $\ket{101}_{ABC}$ is repeatedly applied, inter-spaced by re-thermalising qubits $B$ and $C$ to the environment $T_R$ and hot bath $T_H$ respectively, or in the autonomous mode, where the Hamiltonion that generates the swap is left running continuously, while the qubits are kept coupled to their respective baths.

The final temperature achieved by such a machine is given by
\begin{align}
	\frac{E}{T^f} &= \frac{E_{inc,max}}{T_R} - \frac{E_{inc,max} - E}{T_H},
\end{align}
and so we choose $E_{inc,max}$ such that $T^f = T_C$, the final desired cold temperature, resulting in
\begin{align}\label{eq:cohincmatchtemp}
	E_{inc,max} &= E \left( \frac{ \frac{1}{T_C} - \frac{1}{T_H} }{ \frac{1}{T_R} - \frac{1}{T_H} } \right) \quad\quad \Leftrightarrow \quad\quad E_{inc,max} - E = E \left( \frac{ \frac{1}{T_C} - \frac{1}{T_R} }{ \frac{1}{T_R} - \frac{1}{T_H} } \right).
\end{align}

The work cost of this protocol was discussed in  \secref{app:incoherentauto}, and may be calculated from the heat drawn from the hot bath during the protocol. The heat is comprised of two parts, the preheating of qubit $C$, that we label $Q_{init}$, followed by the heat required to keep it at the hot temperature after repeated incoherent operations. In the limit of infinite repetitions (or the steady state of the autonomous machine), this heat is given by
\begin{align}
	Q_H &= (r_{max} - r) E_C = \left( r_{max} - r \right) \left( E_{inc,max} - E \right).
\end{align}

Note that as the final temperature of the target is the same as in the coherent case, the final ground state population $r_{max}$ is also identical.

Note that from the two heat contributions, the initial heat cost to bring qubit $C$ to $T_H$ from $T_R$ depends on whether it is a real or virtual qubit, and in the latter case, depends on the spectrum of the larger space in which the virtual qubit is embedded. However, the heat required to keep it at $T_H$ remains the same, as it only depends on the population flow between the system and the machine, which in the limit of infinite repetitions or the autonomous steady state, only depends on the Gibb's ratio of qubit $C$.

In a similar manner as in the coherent case, one can decrease the work cost by using a machine made out of a sequence of $N$ two-qubit systems, with linearly increasing energies given by
\begin{subequations}\label{eq:incseqenergies}
\begin{align}
	E_{B,i} &= E + \frac{i}{N} \left( E_{inc,max} - E \right) \\
	E_{C,i} &= E_{B,i} - E = \frac{i}{N} \left( E_{inc,max} - E \right).
\end{align}
\end{subequations}

For each of the two-qubit systems, one runs the same protocol as before, and thus the temperature attained by the $i^{th}$ step is given by
\begin{align}
	\frac{E}{T_{inc,i}} &= \frac{E_{B,i}}{T_R} - \frac{E_{C,i}}{T_H} \\
		&= \frac{E}{T_R} + \left( \frac{1}{T_R} - \frac{1}{T_H} \right) \frac{i}{N} \left( E_{inc,max} - E \right) \\
		&= \frac{E}{T_R} + \frac{i}{N} E \left( \frac{1}{T_C} - \frac{1}{T_R} \right),
\end{align}
by using \eqref{eq:cohincmatchtemp} for $E_{inc,max}$. Thus $T_{inc,i} = T_i$ from the coherent machine (see Eq. \ref{eq:cohseqtemp}), and we keep the notation $T_i$. This implies that the ground state population of the target after the $i^{th}$ step is also the same as in the coherent machine, and we keep the notation $r_i$.

The heat drawn in each step is again comprised of the two contributions of pre-heating and maintenance of the $i^{th}$ qubit $C$,
\begin{align}
	Q_i &= Q_{init,i} + \left( r_i - r_{i-1} \right) E_{C,i} \\
	&= Q_{init,i} + \left( r_i - r_{i-1} \right) \frac{i}{N} \left( E_{inc,max} - E \right).
\end{align}
where we label the initial heat drawn for pre-heating as $Q_{init,i}$. Simplifying the rest of the expression using \eqref{eq:cohincmatchtemp}, and summing up for the total heat,
\begin{align}
	Q &= \sum_{i=1}^N Q_{init,i} + E \sum_{i=1}^N \left( r_i - r_{i-1} \right) \frac{i}{N} \left( \frac{ \frac{1}{T_C} - \frac{1}{T_R} }{ \frac{1}{T_R} - \frac{1}{T_H} } \right) \\
	&= Q_{init} + \frac{W_{coh}}{ 1 - \frac{T_R}{T_H} },
\end{align}
using \eqref{eq:workcostcohseq} for the total work cost in the coherent case. We have denoted by $Q_{init}$ the total cost of pre-heating each of the qubits $\{C,i\}$.

The above is the heat drawn from the hot bath. To compare with the coherent case we take the work cost instead, which is the decrease in free energy of the hot bath,
\begin{align}
	W = \Delta F_H = Q - T_R \Delta S = Q \left( 1 - \frac{T_R}{T_H} \right).
\end{align}

Thus the work cost in the incoherent case is
\begin{align}
	W_{inc} &= Q_{init} \left( 1 - \frac{T_R}{T_H} \right) + W_{coh}.
\end{align}

Thus the incoherent cost is very similar to the coherent cost, with the sole addition of bringing the additional qubits from $T_R$ to $T_H$. At first glance, this may appear to be a finite disadvantage, however it is possible to make this additional cost as small as possible, as we now demonstrate.

The $N$ qubits $\{C,i\}$ need not be real qubits, but virtual ones. Consider for instance, a system with Hamiltonian
\begin{align}
H_C &= - E_g \ket{E_g}\bra{E_g} + \sum_{i=0}^N \frac{i}{N} (E_{inc,max} - E),
\end{align}
which is an evenly spaced ladder of $N+1$ levels plus a single ground state that lies at an energy $E_g$ below the ladder. Labelling the levels by $\{g,0,1,2,...,N\}$, we observe that $E_i - E_0 = E_{C,i}$, and thus the pair of levels $0$ and $i$ may be employed as the $i^{th}$ virtual qubit $C$ in the incoherent machine.

However, for any fixed $N$, $T_H$ and $E_{inc,max}$, the cost of preheating this system can be made as small as we like, by pushing the ground state energy further downward. For high enough values of $E_g$, the population in the ladder will be small enough for both $T_H$ and $T_R$ such that the difference in average energy is vanishingly small.

Note that this implies that the machine will run slower (in the autonomous case) or require many more repeated operations in the discrete case. However, in principle the final temperature attained is still the same and thus the incoherent machine can achieve as close a work cost as one likes to the coherent case. Together with the fact that the coherent machine can get as close to saturating the second law, we thus have the statement that for arbitrary sized machines, both coherent and incoherent machines can approach the limit of the second law of thermodynamics.


\begin{thebibliography}{99}
\bibitem{footnote0} We consider the system to be isolated from the environment during the operations. Taking the system-environment interaction into account can only weaken cooling, and requires a comparison of the time-scales of different operations, which we do not go into in this work.
\end{thebibliography}
\end{document}